\documentclass[12pt]{article}
\usepackage{amsmath,bm,amsfonts,color,amsthm, amssymb}
\usepackage{graphicx}
\usepackage{natbib}
\usepackage{xcolor}
\usepackage{tabu}
\usepackage{enumitem}
\usepackage{subcaption,siunitx,booktabs}
\usepackage[citecolor=blue,urlcolor=blue,allcolors=blue]{hyperref}
\usepackage{arydshln}

\usepackage{titlesec}
\titlelabel{\thetitle.\quad}

\usepackage{tikz}
\usepackage{transparent}

\usepackage[linesnumbered,ruled,vlined]{algorithm2e}

\usetikzlibrary{positioning}
\usetikzlibrary{shadows}
\usepgflibrary{shapes}

\usepackage{mathtools}

\usepackage{apptools}
\newtheorem{theorem}{Theorem}[section]
\newtheorem{lemma}{Lemma}[section]
\newtheorem{definition}{Definition}[section]

\newtheorem{remark}{Remark}

\usepackage{titlesec}

\titleformat*{\section}{\large\bfseries}
\titleformat*{\subsection}{\large\bfseries}
\titleformat*{\subsubsection}{\large\bfseries}
\titleformat*{\paragraph}{\large\bfseries}
\titleformat*{\subparagraph}{\large\bfseries}

\usepackage{titlesec}
\titlelabel{\thetitle\quad}

\theoremstyle{definition}

\numberwithin{mytheorem}{section}
\usepackage{setspace}

\SetKwInput{KwInput}{Input}
\SetKwInput{KwOutput}{Output}
\SetKwInput{KwInit}{Initialization}

\DeclareMathOperator*{\argmin}{arg\,min}
\DeclareMathOperator*{\argmax}{arg\,max}

\begin{document}

\def\spacingset#1{\renewcommand{\baselinestretch}%
{#1}\small\normalsize} \spacingset{1}

\title{\bf Optimal Change Point Detection and Inference in the Spectral Density of General Time Series Models}
  \author{Sepideh Mosaferi \hspace{.2cm}\\
    Department of Mathematics and Statistics\\ University of Massachusetts Amherst\\
    and \\
    Abolfazl Safikhani \\
    Department of Statistics \\ George Mason University\\
    and \\
    Peiliang Bai \\
    Microsoft}

  \date{}
    
  \maketitle


\bigskip
\begin{abstract}
This paper addresses the problem of detecting change points in the spectral density of time series, motivated by EEG analysis of seizure patients. Seizures disrupt coherence and functional connectivity, necessitating precise detection. Departing from traditional parametric approaches, we utilize the Wold decomposition, representing general time series as autoregressive processes with infinite lags, which are truncated and estimated around the change point. Our detection procedure employs an initial estimator that systematically searches across time points. We examine the localization error and its dependence on time series properties and sample size. To enhance accuracy, we introduce an optimal rate method with an asymptotic distribution, facilitating the construction of confidence intervals. The proposed method effectively identifies seizure onset in EEG data and extends to event detection in video data. Comprehensive numerical experiments demonstrate its superior performance compared to existing techniques.
\end{abstract}

\noindent%

\newpage
\spacingset{1.9} 

\section{Introduction}
\label{sec:intro}

Change point detection is a critical field in time series analysis, focused on identifying points in data where statistical properties, such as mean, variance, or distribution, change abruptly \citep{basseville1993detection}. These change points often indicate significant shifts in the underlying processes, which are important for various applications in fields. For example, in neuroscience, change point detection helps identify transitions in brain activity, such as shifts in neural dynamics during cognitive tasks or disease progression \citep{safikhani2018joint}. In finance, detecting shifts in market trends can be essential for identifying regime changes and risk assessment \citep{frisen2008financial}. Also, in signal processing, it is used to detect abrupt changes in sensor data, such as anomalies in speech or sound signals \citep{foggia2015audio}. Finally, in manufacturing, it can help identify faults or quality control problems by detecting changes in production processes \citep{qiu2013introduction}. Detection algorithms can be broadly categorized into online/sequential and offline approaches. In the online/sequential setting, change points must be detected in real-time as data are observed sequentially \citep{wang2015large,chan2021self}. In contrast, the offline approach involves a retrospective analysis of the entire data set to identify all change points \citep{yu2020review}. In this paper, we focus on the offline detection problem, which is particularly relevant for applications requiring a comprehensive post hoc analysis of structural changes in time series data.

There are various methods in the literature (CUSUM-based, fused lasso, dynamic programming, PCA-based, rolling window-based, etc.) to perform offline change point detection in univariate or high-dimensional data across different statistical models. Many of these methods are mainly focused on detecting mean changes. Examples include pruned exact linear time (PELT) method \citep{killick2012optimal} which uses optimal partitioning approach of \cite{jackson2005algorithm} and a pruning step within the dynamic program to detect the structural breaks, \cite{fryzl2018} that apply a tail-greedy Haar transformation to consistently estimate the number and locations of multiple change points in the univariate piecewise-constant model, Inspect method \citep{wang2018high} which proposes a high-dimensional change-point detection method that uses a sparse projection to project the high-dimensional data into a univariate case while \cite{aue2018detecting} develop a method based on the (scaled) functional cumulative sum (CUSUM) statistic for detecting shifts in the mean of a functional data model  \citep{wang2020univariate,fryzlewicz2014wild}. While some other methods accommodate a broader range of change types (e.g. changes in regression coefficients), they often assume a parametric structure, which can be limiting in certain contexts  \citep{pein2017heterogeneous,bai2022unified,kolar2012estimating,aue2009break,avanesov2018change}. Additionally, more recent methods have emerged that are capable of handling changes in temporal dependence, which manifest as shifts in spectral density \citep{davis2006structural,safikhani2018joint,cho2015multiple,cho2016change,wang2019localizing,bai2021multiple}. These methods offer a more flexible modeling framework while focusing on detecting changes. However, they typically do not extend to providing inference on the location of the change points, which is a critical aspect for a comprehensive analysis of temporal dynamics.

Providing inference for change point locations has become a major focus in the recent literature. For example, \cite{xu2024change} performs estimation and inference for change point locations and allows the covariate and noise sequences to be temporally dependent in the functional dependence framework, but the changes are only in the regression coefficients, and it is unclear whether their method can handle changes in the temporal dependence. \cite{kaul2021inferencenew} and \cite{kaul2021inference} provide inference for multiple change points in mean shift models, and thus cannot handle changes in spectral density \citep{wang2023dating}. Further, \cite{madrid2021optimal} provides inference for change points in multivariate data, but focuses on independent observations. More recently, \cite{madrid2023change} and \cite{padilla2021optimal} provide detection and inference in multivariate non-parametric models. However, their detection method focuses on estimating marginal distributions of data, and thus, it is unclear whether they are suitable for change point inference due to changes in temporal dependence/spectral density  \citep{chen2022inference,fryzlewicz2024narrowest}. To the best of our knowledge, there are no methods in the literature designed for detection and inference for the spectral density of general time series models. Thus, there remains a gap in the literature that our paper aims to address by developing a method that not only accommodates changes in the temporal dependence through different spectral densities but also offers inference on the location of the change point.










In this article, we investigate the problem of change point detection in the spectral density of time series models, motivated by the analysis of EEG data from patients undergoing seizures. Seizure events are known to alter the coherence and functional connectivity, making precise detection critical for understanding brain dynamics \citep{safikhani2018joint}. While existing methods in the literature primarily assume specific parametric forms before and after the change point, we take a nonparametric approach by leveraging the well-known Wold decomposition. This allows us to represent the pre- and post-change time series models as autoregressive processes with infinite lags \citep{wang2021consistent}. The core idea is to appropriately truncate these infinite-lag autoregressive models and estimate model parameters before and after the change point.

Our detection procedure begins with an initial estimation of the change point location by systematically searching over all possible time points, while excluding boundary regions as per standard practice in change point analysis. We theoretically establish the localization error for this estimator, linking its non-asymptotic rate to key characteristics of the time series before and after the change point, as well as the sample size (Theorem~\ref{thm:1}). To further enhance detection accuracy, we develop an improved method that not only achieves an optimal localization error rate but also admits an asymptotic distribution (Theorem~\ref{thm:4.1}). This is particularly significant, as it facilitates the construction of confidence intervals for the change point location, an aspect that has received limited attention in the literature (Theorem~\ref{thm:4.2}). Key technical developments include (1) utilizing central limit theorems for martingale differences to derive asymptotic distributions for the change point estimator (due to the existence of temporal dependence), (2) controlling the impact of temporal dependence properly by imposing mixing-type assumptions on time series models, and (3) carefully quantifying the effect of model misspecification (due to existence of change point) in spectral density estimations and Yule-Walker estimations. Additionally, by allowing time series observations to be sub-Weibull random variables, we allow observations to follow several heavy-tail distributions, unlike several existing methods which focus on Gaussian or sub-Gaussian random variables. The effectiveness of our proposed method is demonstrated through its application to EEG data recorded during seizures, where it successfully identifies the onset of seizures in channels near the affected brain region. Furthermore, we illustrate the versatility of our method by applying it to video data for event detection, highlighting its broader applicability beyond EEG analysis. Finally, extensive numerical studies benchmark our method against several existing approaches (both in terms of detection and inference), showcasing its superior performance.

The structure of the paper is as follows. In Section \ref{sec:model-formulation}, we begin by presenting the Wold representation for general time series data, along with the preliminary step of spectral density estimation. The proposed change point modeling framework is then introduced. In Section \ref{sec:near-optimal}, we propose a near-optimal estimator for the change point parameter and provide the corresponding theoretical results to demonstrate its convergence rate. Section \ref{sec:optimal} outlines the construction of the optimal estimator and presents the theoretical analysis, including the asymptotic distribution. In Section \ref{sec:simulations}, we assess the empirical performance of our methodology through several time series processes and compare it with existing approaches in the literature. Section \ref{sec:application} includes two real data applications, using EEG and video data sets. Finally, Section \ref{sec:discussion} summarizes our findings and suggests directions for future research. Technical details and additional numerical results are available in the Supplementary Material.

\noindent
\textit{Notation.} Throughout this paper, we use $\mathbb{R}$ to represent the real line. For any vector $v \in \mathbb{R}^p$, $\|v\|_1$, $\|v\|_2$, and $\|v\|_\infty$ indicate the $\ell_1$ norm, the Euclidean norm, and the maximum norm, respectively. For the sub-Weibull distribution, we use $\|\cdot\|_{\psi_{\gamma}}$ to denote the Orlicz norm with parameter $\gamma$. 
We denote $C,C_1,C_2,...$ as generic constants that may differ in each appearance.
Suppose $A$ is a symmetric or Hermitian matrix, we denote $\lambda_{\min}(A)$ and $\lambda_{\max}(A)$ as the smallest and largest eigenvalues of $A$. For any set $A$, we use $|A|$ and $A^c$ to represent the cardinality and complement of $A$. We also use $a \vee b = \max\{a, b\}$ and $a \wedge b = \min\{a, b\}$ for any $a, b \in \mathbb{R}$, and the notation $\lfloor \cdot \rfloor$ is used to represent the greatest integer function. All limits in this paper are with respect to the sample size $T \to +\infty$.

\section{Preliminaries and Model Formulation}\label{sec:model-formulation}

In this section, we first present some necessary preliminaries and definitions. Then, we formulate the model based on the Wold representation theorem.

\subsection{Preliminaries} \label{sec:preliminary}

We know from the well-known Wold representation that every purely nondeterministic, (weakly) stationary
and zero-mean time series $X_t$ with nonvanishing
spectral density and absolutely summable autoregressive coefficients \citep{wang2021consistent} can be written as

\begin{equation}
    \label{eq:1}
    X_t = \sum_{j=1}^\infty \phi^\star_jX_{t-j} + \epsilon_t,
\end{equation}
where $\epsilon_t$ is white noise with mean zero and some variance $\sigma^2$, i.e. $\epsilon_t \sim \text{WN}(0, \sigma^2)$. Here we have the additional assumption that $\sum_{j=1}^\infty |\phi^\star_j| < \infty$. We can obtain the sample autocovariance as $\hat{\gamma}_k = \frac{1}{T}\sum_{t=1}^{T-|k|}(X_t - \Bar{X})(X_{t+|k|} - \Bar{X})$ based on the data $X_1, X_2, \dots, X_T$ where $\Bar{X}$ is the average of these observations and $T$ is the sample size. Then, for some choice of lag $p$, we fit an AR($p$) model and estimate parameter ($\hat{\phi}$) by using the Yule-Walker equations \citep{brockwell2002introduction}. The spectral density of the fitted AR($p$) model is given by $\hat{f}^{AR}(\lambda) = \frac{\hat{\sigma}^2}{2\pi |\hat{\phi}(e^{i\lambda})|^2}$, where $\hat{\phi}(z) \overset{\text{def}}{=} 1 - \hat{\phi}_1z - \hat{\phi}_2z^2 - \cdots - \hat{\phi}_p z^p$ and $\lambda \in [-\pi, \pi]$. 


\subsection{Time series model with a change point}

Consider the zero-mean univariate process $X_t$ defined in \eqref{eq:1} with a change point. More 

\noindent specifically, there exists a change point $\lfloor T\tau^\star \rfloor \in (1,T)$ such that
\begin{equation}
    \label{eq:3}
    \begin{aligned}
        X_t &= \sum_{j=1}^\infty \phi_{1,j}^\star X_{t-j} + \epsilon_t, \quad t = 1, \dots, \lfloor T\tau^\star \rfloor, \\
        X_t &= \sum_{j=1}^\infty \phi_{2,j}^\star X_{t-j} + \tilde{\epsilon}_t, \quad t = \lfloor T\tau^\star \rfloor +1, \dots, T,
    \end{aligned}
\end{equation}
where $\{ \phi^\star_{1,j}, \phi^\star_{2,j} \}  \in \mathbb{R}$ are model coefficients corresponding to the $j$-th lag before and after the change point, respectively. The terms $\epsilon_t, \tilde{\epsilon}_t \in \mathbb{R}$ are the unobserved zero-mean white noises and are independent of each other. To simplify the notation, in the remainder, we refer to both error terms as $\epsilon_t$ if there are no ambiguities. The change point $\lfloor T\tau^\star \rfloor \in (1,T)$ is the main parameter of interest in this article. Furthermore, the true spectral density functions before and after the change point are denoted as $f_1(\lambda)$ and $f_2(\lambda)$, respectively, and the corresponding characteristic functions are given by $\Phi_1^\star(z) = \sum_{j=1}^\infty \phi_{1,j}^\star z^{j}$ and $\Phi_2^\star(z) = \sum_{j=1}^\infty \phi^\star_{2,j} z^j$.

The main goal of this article is to provide an estimate of the change point $\lfloor T\tau^\star \rfloor$, which achieves the optimal convergence rate, as well as the corresponding limiting distribution, so that it allows the construction of confidence intervals for the unknown change point.

\subsection{Definition of sub-Weibull random variable}

In the following, we provide the definition of an Orlicz norm for random variables and properties of the sub-Weibull distribution. Here, we refer to Section 2.7.1 of \cite{vershynin2018high} and the references therein for more details.
\begin{definition}[Orlicz norms]\label{def:1}
Let $g: [0, \infty) \to [0, \infty)$ be a nondecreasing function with $g(0) = 0$. The \emph{$g$-Orlicz norm} is defined as $\|X\|_g \overset{\text{def}}{=} \inf \left\{ \eta > 0: \mathbb{E}(g(|X| / \eta)) \leq 1 \right\}$ for a real-valued random variable $X$.
\end{definition}
We now define sub-Weibull random variables as follows:
\begin{definition}[Sub-Weibull($\gamma$) random variable and norm]\label{def:2}
A random variable $X$ is called to be sub-Weibull with parameter $\gamma>0$, denoted as sub-Weibull($\gamma$), if there exist constants $K_1, K_2$, and $K_3$ differ from each other at most by a constant depending only on $\gamma$, such that one of the following equivalent statements holds:
\begin{enumerate}
    \item Tails of $X$ satisfy: $\mathbb{P}(|X|>t) \leq 2\exp\left\{ -\left( \frac{t}{K_1} \right)^{\gamma} \right\},\quad \forall \,  t\geq 0$.
    
    \item Moments of $X$ satisfy: $\|X\|_p := (\mathbb{E}|X|^p)^{1/p} \leq K_2 \, p^{1/\gamma},\quad \forall \,  p \geq 1\wedge \gamma$.
    
    \item The moment generating function of $|X|^\gamma$ is finite at some point; i.e. $\mathbb{E}(\exp(|X|/K_3)^\gamma) \leq 2$.
\end{enumerate}
Then the associated norm of $X$, denoted as $\|X\|_{\psi_{\gamma}}$, is defined as:
\begin{equation*}
    \|X\|_{\psi_{\gamma}} \overset{\text{def}}{=} \sup_{p\geq 1}\left(\mathbb{E}|X|^{p}\right)^{1/p}p^{-1/\gamma} < +\infty.
\end{equation*}
\end{definition}
Now we introduce the definition of a sub-Weibull random vector. It is defined via using one-dimensional projections of the random vector along the same lines as the definitions of sub-Gaussian and sub-exponential random vectors.
\begin{definition}\label{def:3}
    Let $\gamma > 0$. A random vector $X \in \mathbb{R}^p$ is called a sub-Weibull($\gamma$) random vector if all of its one dimensional projections are sub-Weibull($\gamma$) random variables. Then we define the norm of $X$ as $\|X\|_{\psi_\gamma} := \sup_{v\in S^{p-1}}\|v^\prime X\|_{\psi_\gamma}$.
\end{definition}

\section{Construction of A Near Optimal Estimator of \texorpdfstring{$\lfloor T\tau^\star \rfloor$}{}}\label{sec:near-optimal}
Let $X_1, X_2, \dots, X_T$ be a sequence of observations with a single change point $\lfloor T\tau^\star \rfloor$ generated by the model presented in \eqref{eq:3}. Then, for any time point $\lfloor T\tau \rfloor \in (1,T)$, we first fit AR($p$) processes using the intervals $[1, \lfloor T\tau \rfloor]$ and $[\lfloor T\tau \rfloor+1, T]$. The fitting procedure is described next.

First, for some choice of $p$ (to be specified later), we consider the sets of observations $\mathcal{T}(1,\lfloor T\tau \rfloor) \overset{\text{def}}{=} \{X_1, X_2, \dots, X_{\lfloor T\tau \rfloor}\}$ and $\mathcal{T}(\lfloor T\tau \rfloor+1, T) \overset{\text{def}}{=} \{X_{\lfloor T\tau \rfloor+1}, X_{\lfloor T \tau \rfloor+2}, \dots, X_T\}$. Then, we obtain the sample autocovariances using $\mathcal{T}(1,\lfloor T\tau \rfloor)$ and $\mathcal{T}(\lfloor T\tau \rfloor+1, T)$, respectively:
\begin{equation}
    \label{eq:4}
    \hat{\gamma}_{1,k} = \frac{1}{\lfloor T\tau \rfloor}\sum_{t=1}^{\lfloor T\tau \rfloor-|k|}(X_t - \Bar{X}_1)(X_{t+|k|} - \Bar{X}_1),\ \  \hat{\gamma}_{2,k} = \frac{1}{T-\lfloor T\tau \rfloor}\sum_{t=\lfloor T\tau \rfloor+1}^{T-\lfloor T\tau \rfloor-|k|}(X_t - \Bar{X}_2)(X_{t+|k|} - \Bar{X}_2),
\end{equation}
where $k=1,2,\dots, p$ and $\Bar{X}_1$ and $\Bar{X}_2$ are the average of observed data in intervals $[1, \lfloor T\tau \rfloor]$ and $[\lfloor T\tau \rfloor+1, T]$, respectively. Then we construct the $p\times p$ autocovariance matrices $\hat{\Gamma}_1 = \left(\hat{\gamma}_{1,|i-j|}\right)_{i,j=1}^p$ and $\hat{\Gamma}_2 = \left(\hat{\gamma}_{2,|i-j|}\right)_{i,j=1}^p$, as well as the vectors $\hat{\gamma}_1 = (\hat{\gamma}_{1,1}, \dots, \hat{\gamma}_{1,p})^\prime$ and $\hat{\gamma}_2 = (\hat{\gamma}_{2,1}, \dots, \hat{\gamma}_{2,p})^\prime$, respectively.


Next, we use them to fit AR($p$) models on $\mathcal{T}(1,\lfloor T\tau \rfloor)$ and $\mathcal{T}(\lfloor T\tau \rfloor+1, T)$ via solving the Yule-Walker equations and find two $p$-dimensional parameter estimates:
\begin{equation}
    \label{eq:5}
    \hat{\phi}_1 = \hat{\Gamma}_1^{-1} \hat{\gamma}_1, \quad 
    \hat{\phi}_2 = \hat{\Gamma}_2^{-1} \hat{\gamma}_2.
\end{equation}

Note that the selected lag $p$ increases by increasing sample size, and thus, is a function of  $T$, i.e. $p = p(T)$. However, for ease of exposition, we omit this dependence and denote the lag as $p$. Also, the selected lag $p$ can be different for observations before and after $\lfloor T\tau \rfloor$. Here, we use the same lag for notational simplicity (see Section~\ref{sec:simulations} for more details on how the lag is selected in practice). Finally, we construct the objective function according to sum of squared error with respect to the time point $\lfloor T\tau \rfloor$ as well as fitted AR($p$) model parameters $\hat{\phi}_1$ and $\hat{\phi}_2$:
\begin{equation}
    \label{eq:7}
    \mathcal{L}(\tau) \overset{\text{def}}{=} \sum_{t=1}^{\lfloor T\tau \rfloor} (X_t - \hat{\phi}_1^\prime Z_t)^2 + \sum_{t=\lfloor T\tau \rfloor +1}^T(X_t - \hat{\phi}_2^\prime Z_t)^2,
\end{equation}
where $Z_t \overset{\text{def}}{=} (X_{t-1}, \cdots, X_{t-p})^\prime \in \mathbb{R}^p$. The \emph{near optimal} estimated change point $\lfloor T\hat{\tau} \rfloor$ is defined as:
\begin{equation}
    \label{eq:8}
    \hat{\tau} = \argmin_{\tau \in \lbrace 2/T, 3/T, \ldots,(T-1)/T \rbrace} \mathcal{L}(\tau).
\end{equation}

\subsection{Theoretical guarantee for the near optimal estimator}

In this section, we establish the theoretical results for the near optimal estimator $\lfloor T\hat{\tau} \rfloor$ of the change point parameter.
Before stating the main theorem, we state two conditions as follows:

\noindent 
{\bf Condition A} (On underlying model distributions): We assume that the underlying distribution in model \eqref{eq:3} satisfies the following three assumptions:
\begin{itemize}
    \item[(a)] The model error terms $\epsilon_t, \tilde{\epsilon}_t$'s are i.i.d. random variables from sub-Weibull($\gamma_2$) distribution for some parameter $\gamma_2$, with variance $\sigma^2$, and are independent of each other and independent of observation $X_{t-1}$. 
    \item[(b)] The processes $X_t$ before and after the change point are geometrically $\beta$-mixing. More specifically, there exist constants $c, \gamma_1 > 0$ such that for all $n \in \mathbb{N}$, $\beta(n) \leq 2\exp(-cn^{\gamma_1})$ for both processes. The beta mixing coefficient is defined as $\beta(n) = \sup_{t \in \mathbb{Z}} \beta(X_{-\infty:t}, X_{t+n:\infty})$ where $\beta(X, Y) = \frac{1}{2}\sup_{I,J}\sum_{i\in I}\sum_{j\in J}|\mathbb{P}(A_i \cap B_j) - \mathbb{P}(A_i)\mathbb{P}(B_j)|$. 
    \item[(c)] Denote the covariance matrices as $\Gamma_p^X(h) = \text{Cov}(X_t, X_{t+h})$ for $t, h \in \mathbb{Z}$. Then, the covariance matrix $\Gamma_p^X(0)$ has bounded eigenvalues, that is, there exist $\kappa_{\min}$ and $\kappa_{\min}$ such that $0 < \kappa_{\min} \leq \lambda_{\min}(\Gamma_p^X(0)) \leq \lambda_{\max}(\Gamma_p^X(0)) \leq \kappa_{\max} < +\infty$. Specifically, this implies that the condition number of $\Gamma_p^X(0)$ is bounded by $\nu = \kappa_{\max} / \kappa_{\min}$. 
\end{itemize}


\begin{remark}
\label{remark:1}
In condition (a), the sub-Weibull assumption significantly relaxes the sub-Gaussian distribution assumption in many detection results. As shown in \cite{wong2020lasso}, given the assumption (a), each random realization $X_t$ follows a sub-Weibull($\gamma_2$) distribution with zero mean and its Orlicz norm $\|X_t\|_{\psi_{\gamma_2}} \leq K_X < +\infty$. We refer to supplementary material for more details and proofs. Condition (b) is an essential assumption to control the overall temporal dependence. The last condition (c) is a natural assumption and provides the assumption for the covariance matrices $\Gamma^X$ of the series $X_t$, which controls the upper and lower bounds for the maximum/minimum eigenvalues of $\Gamma^X$ as well as the condition number. 
\end{remark}

\noindent
{\bf Condition B} (On the model parameters): 
\begin{itemize}
    \item[(a)] Assume that the change point $ \lfloor T\tau^\star \rfloor $ is sufficiently separated from the boundaries of $(0,1)$. In other words, there exists a positive sequence $l_T \to 0$, such that $\lfloor T\tau^\star \rfloor \wedge (T - \lfloor T\tau^\star \rfloor) \geq Tl_T \to +\infty$. Specifically, the sequence $l_T$ satisfies the condition $Tl_T \geq p^2\log p$. 
    \item[(b)] For some appropriately chosen small enough universal constant $c_{u1}>0$, we have the following assumption:
    \begin{equation*}
        c_u\sqrt{1+\nu^2}\frac{\sigma^2}{\kappa_{\min}\xi_2}\left( \frac{p^2\log(p\vee T)}{T^{\frac{1}{2}-b}\sqrt{l_T}}\right) \leq c_{u1},
    \end{equation*}
    where $0 < b < \frac{1}{2}$, and parameters $\sigma^2, \nu, \kappa_{\min}$ are defined in Condition A. Finally, $\xi_2 = \|\phi_1^\star - \phi_2^\star\|_2 = \|\eta^\star\|_2$ is recognized as the jump size and $\phi_1^\star$, $\phi_2^\star \in \mathbb{R}^p$ are first $p$ orders coefficients in model \eqref{eq:3}.
    \item[(c)] There exists a positive constant $v$ such that the spectral densities $f_1(\lambda)$ and $f_2(\lambda)$ satisfy the inequality $\sup_{\lambda \in [-\pi, \pi]}|f_1(\lambda) - f_2(\lambda)| \geq v > 0$.
\end{itemize}

\begin{remark}
Condition (a) describes the boundary conditions for the change point $ \lfloor T\tau^\star \rfloor $. It illustrates that the change point in the given time series cannot be close to the boundaries. Condition (b) demonstrates the connection between sample size and lag p, as well as imposing a lower bound on the model parameter differences before/after the change point \citep{kaul2021inference}. As stated, the selected lag $p$ increases by increasing sample size, and thus, is a function of  $T$, i.e. $p = p(T)$. However, for ease of exposition, we omit this dependence and denote the lag as $p$.

Condition (c) is necessary since it is assumed that the spectral densities of the time series model differ before and after the change point; otherwise, there is no need to perform detection. Note that such an assumption is typically written in terms of an autoregressive model parameter. For example, in univariate time series change point detection problems, often they propose the assumption on the difference between the model parameters, that is, for some $k$, $|\phi^1_k - \phi^2_k| \geq c_0 > 0$ holds for some constant $c_0$, where $\phi^1$ and $\phi^2$ are the model parameters \citep{timmer2003change, bai2003computation}. It is possible to make a connection between these two sets of assumptions. To that end, note that:
\begin{equation*}
    \begin{aligned}
        |f_1(\lambda) - f_2(\lambda)| &= \left| \frac{\sigma^2}{2\pi|\phi^\star_1(e^{i\lambda})|^2} - \frac{\sigma^2}{2\pi|\phi^\star_2(e^{i\lambda})|^2} \right| = \frac{\sigma^2}{2\pi}\left|  \frac{\left| |\phi^\star_2(e^{i\lambda})|^2 - |\phi^\star_1(e^{i\lambda})|^2 \right|}{|\phi^\star_1(e^{i\lambda})\phi^\star_2(e^{i\lambda})|^2}\right| \\
        &\leq \frac{\sigma^2}{2\pi}\frac{\left| \phi^\star_1(e^{i\lambda}) - \phi^\star_2(e^{i\lambda}) \right|\left( |\phi^\star_1(e^{i\lambda})| + |\phi^\star_2(e^{i\lambda})| \right)}{|\phi^\star_1(e^{i\lambda})\phi^\star_2(e^{i\lambda})|^2} \leq \frac{C_0\sigma^2}{2\pi}\left| \phi^\star_1(e^{i\lambda}) - \phi^\star_2(e^{i\lambda}) \right| \\
        &= \frac{C_0\sigma^2}{2\pi}\sum_{k=1}^\infty|(\phi_{1,k}^\star - \phi_{2,k}^\star)e^{ik\lambda}| \leq \frac{C_0 \sigma^2}{2\pi}\sum_{k=1}^\infty|\phi_{1,k}^\star - \phi_{2,k}^\star|.
    \end{aligned}
\end{equation*}

The first inequality holds because both $(\phi^\star_j(e^{i\lambda}))^{-1}$ and $|\phi^\star_j(e^{i\lambda})|$ for $j=1,2$ are bounded. Hence, $C_0>0$ is a positive constant large enough to indicate the upper bound for the previous quantities. Since the model parameters are assumed to be absolutely summable, we can easily derive that if $|f_1(\lambda) - f_2(\lambda)| \geq v$, then we have $\sum_{j=1}^\infty |\phi_{1,j}^\star - \phi_{2,j}^\star| \geq {2\pi v}/{C_0\sigma^2}$, which implies that it is sufficient for us to assume that the difference between the spectral densities is lower bounded to identify the change point.
\end{remark}

\begin{theorem}
\label{thm:1}
Suppose conditions A and B are satisfied. Then, there exists a large enough positive constant $K_0 > 0$ such that as $T \to \infty$:
\begin{equation*}
    \mathbb{P}\left( \Big|\lfloor T\hat{\tau} \rfloor - \lfloor T\tau^\star \rfloor \Big| \leq K_0\frac{p^2\log (p \vee T)}{\xi_2^2}\right) \to 1,
\end{equation*}
where $\xi_2 = \|\phi_1^\star - \phi_2^\star\|_2 = \|\eta^\star\|_2$ is recognized as the jump size and $\phi_1^\star$, $\phi_2^\star \in \mathbb{R}^p$ are first $p$ orders coefficients in model \eqref{eq:3}. Specifically, when $\xi_2 = \mathcal{O}(p)$, we have $\lfloor T\hat{\tau} \rfloor - \lfloor T\tau^\star \rfloor = \mathcal{O}(\log(p\vee T))$ with probability $1 - o(1)$.
\end{theorem}

Theorem \ref{thm:1} provides a near optimal convergence rate of the change point parameter $\tau^\star$ given the nuisance parameter estimators $\hat{\phi}_1$ and $\hat{\phi}_2$ obtained by the Yule-Walker equations in \eqref{eq:5}.

\section{Construction of An Optimal Estimator of \texorpdfstring{$\lfloor T\tau^\star \rfloor$}{}}\label{sec:optimal}
In this section, our aim is to construct the optimal estimator of the change point parameter $\lfloor T\tau^\star \rfloor$ while assuming that the \emph{near optimal} estimator $ \lfloor T\hat{\tau} \rfloor $ is available. Let $\eta^\star = \phi_1^\star - \phi_2^\star \in \mathbb{R}^p$, and define $\xi_2 = \|\eta^\star\|_2$ and $\xi_1 = \|\eta^\star\|_1$. Based on the near optimal estimator $\lfloor T\hat{\tau} \rfloor$, we can partition the time series into $\mathcal{T}(1, \lfloor T \hat{\tau} \rfloor )$ and $\mathcal{T}(\lfloor T \hat{\tau} \rfloor+1, T)$, and separately refit the AR($p$) coefficients $\hat{\phi}_1(\hat{\tau})$ and $\hat{\phi}_2(\hat{\tau})$ using the Yule-Walker equations as in \eqref{eq:5}. To simplify the notation, we omit $\hat{\tau}$ in the estimated coefficients and just denote them by $\hat{\phi}_1$ and $\hat{\phi}_2$.

Consider the squared loss for any $\tau \in (0, 1)$, and $\phi_1, \phi_2 \in \mathbb{R}^p$:
\begin{equation}\label{eq:9}
    Q(\tau; \phi_1, \phi_2) = \frac{1}{T-p+1}\left( \sum_{t=p}^{\lfloor T\tau \rfloor} (X_t - \phi_1^\prime Z_t)^2 + \sum_{t=\lfloor T\tau\rfloor +1}^{T}(X_t - \phi_2^\prime Z_t)^2 \right).
\end{equation}
Under this setup, we define the refitted least squares estimator as:
\begin{equation}
    \tilde{\tau} = \argmin_{\tau \in \lbrace 2/T, 3/T, \ldots,(T-1)/T \rbrace} Q(\tau; \hat{\phi}_1, \hat{\phi}_2).
\end{equation}
The following Algorithm 1 summarizes the procedure to obtain the optimal estimator of the change point parameter $\lfloor T\tau^\star \rfloor$. The main idea is to separate the estimation of the model parameters and the detection of the change points \citep{kaul2021inference,kaul2021inferencenew}. 
\begin{algorithm}[!ht]
    \label{algo:1}
    \caption{{\bf Optimal estimator of $\tau^\star$ detection procedure}}
    \ \KwInput{ Consider the time series data set $\{X_t\}^T_{t=1}$.}
    
    \ {\bf Step 1:} Obtain \emph{near optimal} estimator $\lfloor T\hat{\tau} \rfloor$ according to Section \ref{sec:near-optimal}.
    
    \ {\bf Step 2:} Refit AR($p$) coefficients $\hat{\phi}_1 = \hat{\phi}_1(\hat{\tau})$, and $\hat{\phi}_2 = \hat{\phi}_2(\hat{\tau})$.
    
    \ {\bf Step 3:} Update the estimator for the change point parameter as:
    \begin{equation*}
        \tilde{\tau} = \argmin_{\tau \in \lbrace 2/T, 3/T, \ldots,(T-1)/T \rbrace}Q(\tau; \hat{\phi}_1, \hat{\phi}_2). 
    \end{equation*}
    
    \ \KwOutput{ The final estimated change point parameter is  $ \lfloor T\tilde{\tau} \rfloor $.}
\end{algorithm}

\subsection{Theoretical results of convergence of the change point estimator}\label{sec:optimal-theory}
First, we start by defining a few notation that are essential in this section. For any generic AR($p$) coefficients $\phi_1, \phi_2 \in \mathbb{R}^p$, and $\tau \in (0,1)$, we define:
\begin{equation}
    \mathcal{U}(\tau; \phi_1, \phi_2) = Q(\tau; \phi_1, \phi_2) - Q(\tau^\star; \phi_1, \phi_2),
\end{equation}
where $\tau^\star \in (0,1)$ is the unknown change point parameter and $Q$ is the least squares loss defined in \eqref{eq:9}. For any non-negative sequences $u_T$ and $v_T$, with $0 \leq v_T \leq u_T$, we consider the collection of change point parameters:
\begin{equation}
    \mathcal{G}(u_T, v_T) = \left\{ \tau \in (0,1): Tv_T \leq \Big| \lfloor T\tau \rfloor - \lfloor T\tau^\star \rfloor \Big| < Tu_T \right\}. 
\end{equation}
Next, we present the following lemma that provides a uniform lower bound with respect to $\mathcal{U}(\tau; \hat{\phi}_1, \hat{\phi}_2)$ in the collection $\mathcal{G}(u_T, v_T)$. 
\begin{lemma}\label{lemma:4.1}
Suppose conditions A and B hold. Let $u_T$ and $v_T$ be any nonnegative sequences such that $(1/T) \leq v_T \leq u_T$. For any sufficiently small constant $0 < a< 1$, let $c_{a1} = 2^{2/{\gamma_2}}c_{a2}$ with $c_{a2} \geq \max\left\{ (-c_2\log a)^{1/2}, \sqrt{(1/a)} \right\}$, where the constant $c_2$ depends only on $\gamma_1$, $\gamma_2$, and $c$. Then, we have:
\begin{equation}
    \inf_{\tau\in \mathcal{G}(u_T, v_T)} \mathcal{U}(\tau; \hat{\phi}_1, \hat{\phi}_2) \geq \kappa_{\min}\xi_2^2\left[ v_T - c_{a3}\max\left\{ \left(\frac{u_T}{T}\right)^{\frac{1}{2}}, \frac{u_T}{T^{2b}} \right\} \right],
\end{equation}
with probability at least $1-3a-o(1)$.
\end{lemma}

We are now ready to state the localization error result for the new estimator.

\begin{theorem}\label{thm:4.1}
Suppose conditions A and B hold. Let the constants $c_{a1}, c_{a2}$, and $c_{a3}>0$ be the same as defined in Lemma \ref{lemma:4.1}. Then, the refitted least squares estimate $\lfloor T\tilde{\tau} \rfloor$ satisfies the following conclusions as $T\to +\infty$:
\begin{itemize}
    \item[(i)] When $\xi_2 \to 0$, we have
    \begin{equation*}
        \kappa_{\min}^2(K_\epsilon \vee K_X)^{-4}\xi_2^2\Big| \lfloor T\tilde{\tau} \rfloor - \lfloor T\tau^\star \rfloor \Big| \leq c_u^2c_{a1}^2,
    \end{equation*}
    with probability at least $1-3a-o(1)$.
    \item[(ii)] When $\xi_2 \not\to 0$, we have
    \begin{equation*}
        \Big| \lfloor T\tilde{\tau}\rfloor - \lfloor T\tau^\star \rfloor \big| \leq c_{a3}^2,
    \end{equation*}
    with probability at least $1-3a-o(1)$. Hence, we have $\Big| \lfloor T\tilde{\tau} \rfloor - \lfloor T\tau^\star \rfloor \Big| = \mathcal{O}_p(1)$.
\end{itemize}
\end{theorem}

Theorem~\ref{thm:4.1} quantifies the localization error to estimate the change point, and its rate is improved compared to the result of Theorem~\ref{thm:1} by removing terms related to the estimation error of autoregressive parameters, i.e. $p^2\log (p \vee T)$. The new rate only depends reciprocally on the defined jump size $\xi_2$, which is considered the optimal localization error rate \citep{basseville1993detection,kaul2021inferencenew}.

\subsection{Asymptotic distribution of the change point estimator}
Before we continue to work on the asymptotic distribution of the change point estimator, the following additional assumptions are required. 

\noindent
{\bf Condition C} (On nuisance parameter estimators): For a fixed positive integer $p$, let $p$-dimensional vectors $\hat{\phi}_1$ and $\hat{\phi}_2$ be the estimators of Wold representation coefficient vectors $\phi_1^\star$ and $\phi_2^\star$, respectively. Then, the following bound should hold:
\begin{equation}
    \|\hat{\phi}_1 - \phi_1^\star\|_2 \vee \|\hat{\phi}_2 - \phi_2^\star\|_2 \leq c_u\sqrt{1+\nu^2}\frac{\sigma^2}{\kappa_{\min}}\sqrt{\frac{p^2\log (p\vee T)}{Tl_T}},
\end{equation}
with probability at least $1 - o(1)$, where $l_T$ is a sequence defined in Condition B(a), and $c_u>0$ is some universal large enough constant. Moreover, the following condition on the tail parameters of the Wold representation is required. Suppose $p$ is the fixed order of autoregressive model, then the tail lags of the true time series with respect to the Wold representation satisfies:
\begin{equation*}
    \sum_{k=p+1}^{\infty}(1+k)^r|\phi_{j,k}^\star| \leq \sqrt{\frac{p\log (p\vee T)}{T}},\quad j=1,2,
\end{equation*}
where $r \geq 0$ is a constant.

\noindent
{\bf Condition D} (On the asymptotic distribution): (a) Given the covariance matrices $\Sigma_Z^1 = \mathbb{E}(Z_tZ_t^\prime)$ for $t=1,2,\dots, \lfloor T\tau^\star \rfloor$ and $\Sigma_Z^2 = \mathbb{E}(Z_tZ_t^\prime)$ for $t=\lfloor T\tau^\star\rfloor+1, \dots, T$, we have the following limits:
\begin{equation*}
    \xi_2^{-2}(\eta^{\star^T}\Sigma_Z^1 \eta^\star) \to \sigma_1^2\quad \text{and}\quad \xi_2^{-2}(\eta^{\star^T}\Sigma_Z^2 \eta^\star) \to \sigma_2^2,
\end{equation*}
where $0 < \sigma_1^2, \sigma_2^2 < \infty$. 

(b) For error terms $\epsilon_t, \tilde{\epsilon}_t$, $t=1,2,\dots, T$, we assume that
\begin{equation*}
    \begin{aligned}
        &\xi_2^{-2}\text{Var}\left(\epsilon_t \eta^{\star^\prime} Z_t\right) \to \sigma_1^{\star^2},\quad \text{for }t=1, \dots, \lfloor T\tau^\star \rfloor, \\
        &\xi_2^{-2}\text{Var}\left(\tilde{\epsilon}_t \eta^{\star^\prime} Z_t\right) \to \sigma_2^{\star^2},\quad \text{for }t=\lfloor T\tau^\star\rfloor + 1,\dots, T,
    \end{aligned}
\end{equation*}
where $0 < \sigma_1^{\star^2}, \sigma_2^{\star^2} < +\infty$.

\begin{remark}
Condition C provides the rate of the nuisance parameter estimators and an upper bound rate for the tail parameters in the Wold representation. Also, the finite quantities $\sigma_1^2$ and $\sigma_2^2$ in condition D serve as the variance parameters of the limiting process. Similar conditions can be found in \cite{kaul2021inference}, \cite{wang2021consistent}, and \cite{kaul2021inferencenew}.
\end{remark}

Finally, we consider the following process:
\begin{equation}\label{eq:asymp-dist}
    Z(r) = 
    \begin{cases}
    2W_1(-r) + r, &\quad \text{if } r < 0, \\
    0, &\quad \text{if } r = 0, \\
    \frac{2\sigma_2^\star}{\sigma_1^\star}W_2(r) - \frac{\sigma_2^2}{\sigma_1^2}r, &\quad \text{if } r > 0,
    \end{cases}
\end{equation}
where $W_1(r)$ and $W_2(r)$ are two independent Brownian motions defined on $[0, +\infty)$, and $0 < \sigma_1, \sigma_2, \sigma_1^\star, \sigma_2^\star < +\infty$ are parameters defined in Condition D which determine the variance as well as the drift of the process $Z(r)$. The density of this limiting distribution is available in closed form in \cite{bai1997estimation}.

\begin{theorem}\label{thm:4.2}
Suppose Conditions A--D hold. Further, assume that $\xi_2 \to 0$ while satisfying
\begin{equation*}
    \frac{1}{\xi_2^2}\frac{p^2\log(p\vee T)}{Tl_T} = o(1).
\end{equation*}
Then, the change point estimator $\tilde{\tau}$ has the following asymptotic distribution:
\begin{equation}
    T\sigma_1^4\sigma_1^{\star^{-2}}\xi_2^2(\tilde{\tau} - \tau^\star) \overset{d}{\to} \argmax_{r\in \mathbb{R}}Z(r),
\end{equation}
where $Z(r)$ is defined in \eqref{eq:asymp-dist}.
\end{theorem}

Theorem~\ref{thm:4.2} allows the construction of a confidence interval for the location of the change point, especially because the density of $Z(r)$ is available in closed form \citep{bai1997estimation}.

\section{Simulation Studies} \label{sec:simulations}

In this section, we evaluate the empirical performance of our methodology for a variety of scenarios. The metrics used for the evaluations include absolute bias (AB): $E[|\lfloor T\tau \rfloor - \lfloor T\tau^\star \rfloor|]$ and root mean squared error (RMSE): $\sqrt{E[(\lfloor T\tau \rfloor - \lfloor T\tau^\star \rfloor)^2]}$, where $\lfloor T\tau \rfloor$ is a generic estimator for the location of the change point. Both are computed numerically over 100 replications. In addition, for inference, we report the empirical coverage compared to the nominal ones (e.g. $90\%, 95\%, 99\%$) for the confidence intervals. We consider the following simulation scenarios inspired by \cite{wang2021consistent}:

\begin{itemize}
\item Scenario (I): before change point: MA(1) process $X_t = \epsilon_t + \theta \epsilon_{t-1}$ where $\theta \in  \{-0.9,0.9\}$; after change point: non-linear process $X_t= \phi |X_{t-1}|+ \epsilon_t$ where $\phi \in \{-0.5,0.5\}$.
\item Scenario (II): before change point: MA(1) process $X_t = \epsilon_t + \theta \epsilon_{t-1}$ where $\theta = -0.9$; after change point: AR(1) process $X_t=\phi X_{t-1}+\epsilon_t$ where $\phi=0.5$.
\item Scenario (III): before change point: AR(3) process $X_t=0.9 X_{t-1}-0.5 X_{t-2}+0.3 X_{t-3}+\epsilon_t$; after change point: AR(1) process $X_t=\phi X_{t-1}+\epsilon_t$ where $\phi= -0.9$.
\item Scenario (IV): before change point: MA(1) process $X_t = \epsilon_t + \theta \epsilon_{t-1}$ where $\theta = -0.9$; after change point: AR(3) process $X_t=0.9 X_{t-1}-0.5 X_{t-2}+0.3 X_{t-3}+\epsilon_t$.
\item Scenario (V): before change point: AR(3) process $X_t=0.9 X_{t-1}-0.5 X_{t-2}+0.3 X_{t-3}+\epsilon_t$; after change point: non-linear process $X_t= \phi |X_{t-1}|+ \epsilon_t$, where $\phi \in \{-0.5,0.5\}$.
\end{itemize}


For all cases, we generate the error term $\{ \epsilon_t \}$ as normally distributed from $N(0,\sigma^2)$, where $\sigma \in \{ 0.1,1\}$. We present the results for two sample sizes $T=500,1000$. We consider multiple true change point locations as $ \lfloor T \tau^\star \rfloor \in \{\lfloor T/3 \rfloor, \lfloor T/2 \rfloor, \lfloor 2T/3 \rfloor, \lfloor 4T/5 \rfloor \}$. 
After generating each time series, we used \texttt{AIC} criterion to select the lag $p$ for the time series before and after the change point following the suggestions in  \cite{wang2021consistent}.
The plots of the true spectral densities of the simulation scenarios before and after the change points are shown in Figure \ref{Fig.Spectrum}. 
Further, we present the results for $T=500$ in Tables \ref{Tab.settingI} -- \ref{Tab.settingV}.

The results show that when the jump size is relatively large, we can easily detect the change point, leading to a smaller AB and RMSE as well as more accurate confidence coverage compared to nominal levels. The most complicated setting is scenario V in which the spectral densities before and after the change point have a rather similar pattern (bottom panels in Figure \ref{Fig.Spectrum}) which, as expected, results in relatively larger detection error as seen in say AB in Table~\ref{Tab.settingV}. This also makes the confidence coverages to be slightly lower than the nominal levels (due to bias from estimating the change point location properly).

In Figure \ref{Fig.CP_CaseI}, we show some of the coverage probability (CP) results when the sample size is $T=1000$ for Scenario I. These results confirm that the coverage probabilities are closely aligned with the nominal levels. 
Additional simulation results for both sample sizes 500 and 1000 are summarized in the Supplementary Material (Section~D).

\begin{table}[ht]
\footnotesize
\caption{Performance of the model for scenario I with $T=500$. } \label{Tab.settingI}
\centering
\setlength{\tabcolsep}{1pt} 
\begin{tabular}{@{} c|cccccccc @{}}
\hline
Coefficient & Truth & AB ($\lfloor T\hat{\tau} \rfloor$) & AB ($\lfloor T\tilde{\tau} \rfloor$) & RMSE ($\lfloor T\hat{\tau} \rfloor$) & RMSE ($\lfloor T\tilde{\tau} \rfloor$) & 90$\%$ CP 
($\lfloor T\tilde{\tau} \rfloor$) & 95$\%$ CP ($\lfloor T\tilde{\tau} \rfloor$) & 99$\%$ CP ($\lfloor T\tilde{\tau} \rfloor$) \\
\hline\hline
& $\lfloor T/3 \rfloor$ & 4.380 & 4.000 & 6.927 & 6.442 & 0.930 & 0.950 & 0.980 \\
$\theta=-0.9$ & $\lfloor T/2 \rfloor$ & 4.030 & 3.770 & 6.684 & 6.199 & 0.960 & 0.980 & 0.980 \\
$\phi=-0.5$ & $\lfloor 2T/3 \rfloor$ &  3.630 & 3.330 & 5.529 & 5.037 & 0.950 & 0.980 & 0.990 \\
& $\lfloor 4T/5 \rfloor$ &  3.130 & 2.340 &  5.871 & 3.842 & 0.980 & 0.990 & 1.000 \\
\hline\hline
& $\lfloor T/3 \rfloor$ & 8.930 & 8.560 & 19.172 & 19.141 & 0.970 & 0.980 & 0.980 \\
$\theta=0.9$ & $\lfloor T/2 \rfloor$ & 6.810 & 6.240 & 11.822 & 11.426 & 0.960 & 0.970 & 0.990 \\
$\phi=-0.5$ & $\lfloor 2T/3 \rfloor$ & 7.940 & 7.350 & 14.313 & 14.111 & 0.930 & 0.950 & 0.970 \\
& $\lfloor 4T/5 \rfloor$ & 7.160 & 6.850 & 11.760 & 11.173 & 0.910 & 0.930 & 0.960 \\
\hline\hline
& $\lfloor T/3 \rfloor$ & 4.210 & 3.820 & 6.821 & 6.279 & 0.940 & 0.960 & 0.970 \\
$\theta=-0.9$ & $\lfloor T/2 \rfloor$ & 4.740 & 4.220 & 8.021 & 7.268 & 0.950 & 0.970 & 0.970 \\
$\phi=0.5$ & $\lfloor 2T/3 \rfloor$ & 4.230 & 3.930 & 7.560 & 7.197 & 0.930 & 0.960 & 0.980 \\
& $\lfloor 4T/5 \rfloor$ &  4.300 & 3.260 & 9.103 & 5.389 & 0.940 & 0.960 & 0.990 \\
\hline\hline
& $\lfloor T/3 \rfloor$ & 10.210 & 9.860 &  22.120 & 21.882 & 0.940 & 0.960 & 0.970 \\
$\theta=0.9$ & $\lfloor T/2 \rfloor$ & 7.050 & 6.530 & 11.630 & 10.766 & 0.970 & 0.980 & 1.000 \\
$\phi=0.5$ & $\lfloor 2T/3 \rfloor$ & 7.940 & 6.630 & 14.291 & 11.404 & 0.940 & 0.960 & 0.990 \\
& $\lfloor 4T/5 \rfloor$ &  9.480 & 7.820 & 15.726 & 12.820 & 0.910 & 0.960 & 0.980 \\
\hline
\end{tabular}
\end{table}

\begin{table}[ht]
\footnotesize
\caption{Performance of the model for scenario II with $T=500$.} \label{Tab.settingII}
\centering
\setlength{\tabcolsep}{1pt} 
\begin{tabular}{@{} c|cccccccc @{}}
\hline
Coefficient & Truth & AB ($\lfloor T\hat{\tau} \rfloor$) & AB ($ \lfloor T\tilde{\tau} \rfloor$) & RMSE ($\lfloor T\hat{\tau} \rfloor$) & RMSE ($\lfloor T\tilde{\tau} \rfloor$) & 90$\%$ CP 
($ \lfloor T\tilde{\tau} \rfloor$) & 95$\%$ CP ($\lfloor T \tilde{\tau} \rfloor$) & 99$\%$ CP ($\lfloor T\tilde{\tau} \rfloor$) \\\hline\hline
& $\lfloor T/3 \rfloor$ & 2.680 & 2.450 & 4.693 & 4.403 & 0.890 & 0.900 & 0.940 \\
$\theta=-0.9$ & $\lfloor T/2 \rfloor$ & 2.380 & 2.480 & 3.945 & 4.268 & 0.900 & 0.920 & 0.940 \\
$\phi=0.5$ & $\lfloor 2T/3 \rfloor$ & 2.470 & 2.130 & 3.759 & 3.318 & 0.910 & 0.930 & 0.980 \\
& $\lfloor 4T/5 \rfloor$ & 2.030 & 1.970 & 3.872 & 3.678 & 0.900 & 0.950 & 0.970 \\
\hline
\end{tabular}
\end{table}

\begin{table}[ht]
\footnotesize
\caption{Performance of the model for scenario III with $T=500$.} \label{Tab.settingIII}
\centering
\setlength{\tabcolsep}{1pt} 
\begin{tabular}{@{} c|cccccccc @{}}
\hline
Coefficient & Truth & AB ($\lfloor T\hat{\tau} \rfloor$) & AB ($\lfloor T\tilde{\tau} \rfloor$) & RMSE ($\lfloor T\hat{\tau} \rfloor$) & RMSE ($\lfloor T\tilde{\tau} \rfloor$) & 90$\%$ CP 
($\lfloor T\tilde{\tau} \rfloor$) & 95$\%$ CP ($\lfloor T\tilde{\tau} \rfloor$) & 99$\%$ CP ($\lfloor T\tilde{\tau} \rfloor$) \\\hline\hline
& $\lfloor T/3 \rfloor$ & 0.950 & 0.950 & 1.473 & 1.578 & 0.930 & 0.940 & 0.950 \\
$\phi=-0.9$ & $\lfloor T/2 \rfloor$ & 1.440 & 1.420 & 2.437 & 2.433 & 0.840 & 0.850 & 0.880 \\
& $\lfloor 2T/3 \rfloor$ & 1.110 & 1.200 & 1.694 & 1.871 & 0.900 & 0.910 & 0.930 \\
& $\lfloor 4T/5 \rfloor$ &  1.870 & 1.630 & 3.372 & 2.655 & 0.780 & 0.810 & 0.860 \\
\hline
\end{tabular}
\end{table}

\begin{table}[ht]
\footnotesize
\caption{Performance of the model for scenario IV with $T=500$.} \label{Tab.settingIV}
\centering
\setlength{\tabcolsep}{1pt} 
\begin{tabular}{@{} c|cccccccc @{}}
\hline
Coefficient & Truth & AB ($\lfloor T\hat{\tau} \rfloor$) & AB ($\lfloor T\tilde{\tau} \rfloor$) & RMSE ($\lfloor T\hat{\tau} \rfloor$) & RMSE ($\lfloor T \tilde{\tau} \rfloor$) & 90$\%$ CP 
($\lfloor T\tilde{\tau} \rfloor$) & 95$\%$ CP ($\lfloor T\tilde{\tau} \rfloor$) & 99$\%$ CP ($\lfloor T\tilde{\tau} \rfloor$) \\\hline\hline
& $\lfloor T/3 \rfloor$ & 1.380 & 1.440 & 2.691 & 2.728 & 0.910 & 0.910 & 0.920 \\
$\theta=-0.9$ & $\lfloor T/2 \rfloor$ & 1.740 & 1.760 & 2.768 & 2.775 & 0.810 & 0.830 & 0.900 \\
& $\lfloor 2T/3 \rfloor$ & 1.810 & 1.680 & 2.968 & 2.612 & 0.830 & 0.870 & 0.910 \\
& $\lfloor 4T/5 \rfloor$ & 1.420 & 1.410 & 2.276 & 2.296 & 0.870 & 0.890 & 0.930 \\
\hline
\end{tabular}
\end{table}

\begin{table}[ht]
\footnotesize
\caption{Performance of the model for scenario V with $T=500$.} \label{Tab.settingV}
\centering
\setlength{\tabcolsep}{1pt} 
\begin{tabular}{@{} c|cccccccc @{}}
\hline
Coefficient & Truth & AB ($\lfloor T\hat{\tau} \rfloor$) & AB ($\lfloor T\tilde{\tau} \rfloor$) & RMSE ($\lfloor T\hat{\tau} \rfloor$) & RMSE ($\lfloor T\tilde{\tau} \rfloor$) & 90$\%$ CP 
($\lfloor T\tilde{\tau} \rfloor$) & 95$\%$ CP ($\lfloor T \tilde{\tau} \rfloor$) & 99$\%$ CP ($\lfloor T\tilde{\tau} \rfloor$) \\\hline\hline
& $\lfloor T/3 \rfloor$ & 16.690 & 17.640 & 38.977 & 39.789 & 0.890 & 0.920 & 0.960 \\
$\phi=-0.5$ & $\lfloor T/2 \rfloor$ & 10.840 & 9.490 & 20.429 & 18.639 & 0.930 & 0.960 & 0.970 \\
& $\lfloor 2T/3 \rfloor$ & 8.260 & 8.390 & 13.854 & 14.798 & 0.910 & 0.950 & 0.990  \\
& $\lfloor 4T/5 \rfloor$ & 8.500 & 7.550 & 14.310 & 11.985 & 0.890 & 0.920 & 0.960 \\
\hline\hline
& $\lfloor T/3 \rfloor$ & 15.100 & 14.070 & 30.370 & 29.970 & 0.890 & 0.900 & 0.950 \\
$\phi=0.5$ & $\lfloor T/2 \rfloor$ & 9.470 & 9.000 & 17.434 & 17.214 & 0.960 & 0.970 & 0.980 \\
& $\lfloor 2T/3 \rfloor$ & 8.860 & 8.340 &  15.677 & 15.867 & 0.900 & 0.930 & 0.960 \\
& $\lfloor 4T/5 \rfloor$ & 12.140 & 10.280 & 18.489 & 15.405 & 0.810 & 0.880 & 0.940 \\
\hline
\end{tabular}
\end{table}

\begin{figure}[ht]
\centering
\begin{tabular}{ cc }
\includegraphics[width=0.5\textwidth]{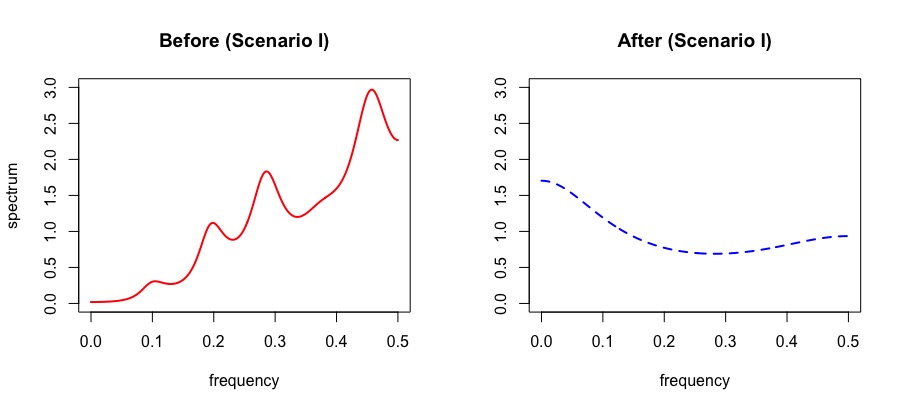} &
\includegraphics[width=0.5\textwidth]{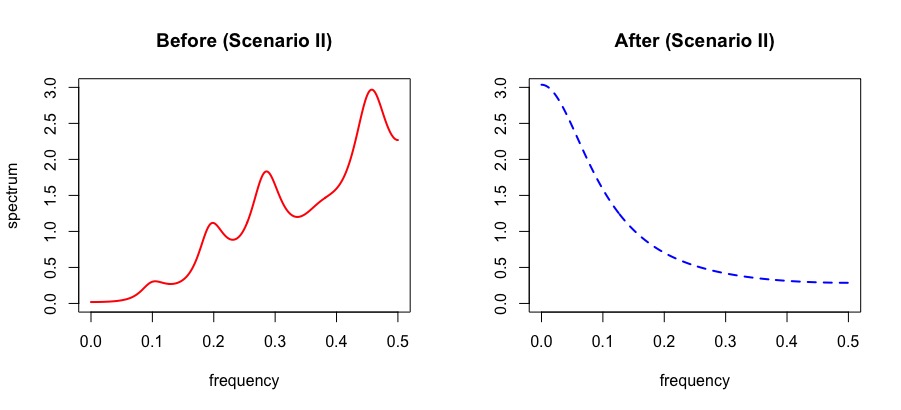} \\
\includegraphics[width=0.5\textwidth]{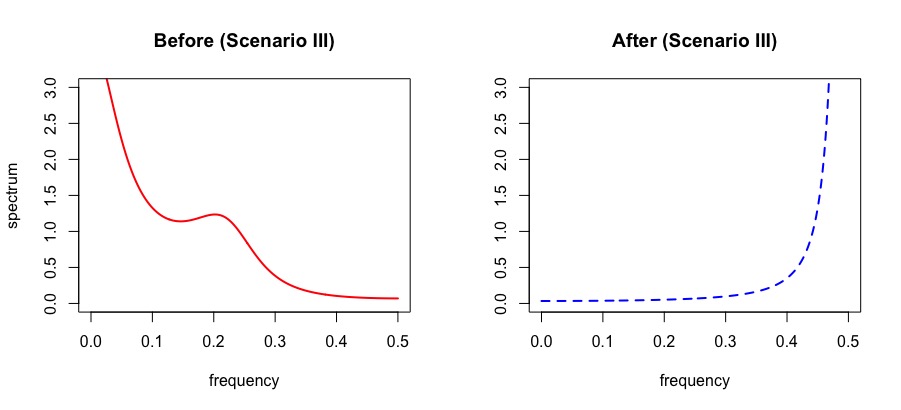} &  
\includegraphics[width=0.5\textwidth]{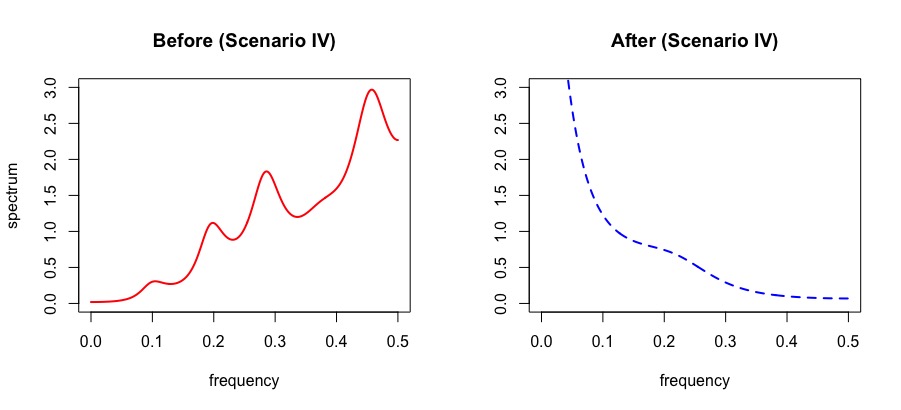} \\
\includegraphics[width=0.5\textwidth]{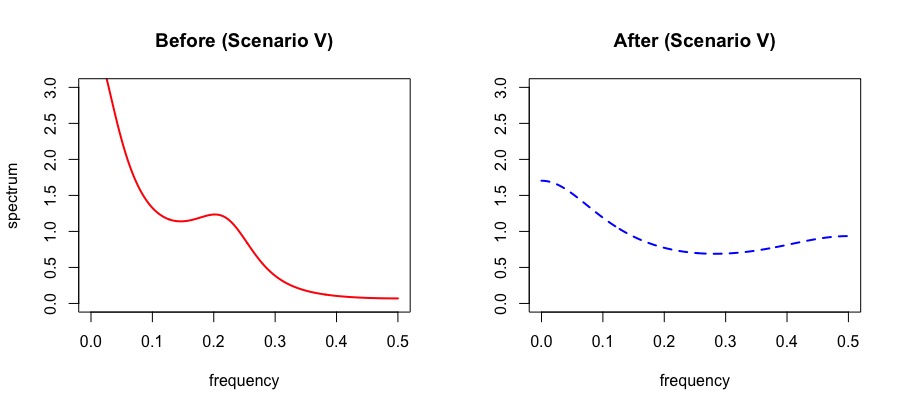}
\end{tabular}
\caption{Spectral densities of selected simulation scenarios for before and after change points. The values of coefficients for all the scenarios are as follows: Scenario (I) $\theta=-0.9$ and $\phi=-0.5$, Scenario (II) $\theta=-0.9$ and $\phi=0.5$, Scenario (III) $\phi=-0.9$, Scenario (IV) $\theta=-0.9$, and Scenario (V) $\phi=-0.5$.}
\label{Fig.Spectrum}
\end{figure}

\begin{figure}[ht]
\centering
\begin{tabular}{ cc }
\includegraphics[width=0.4\textwidth]{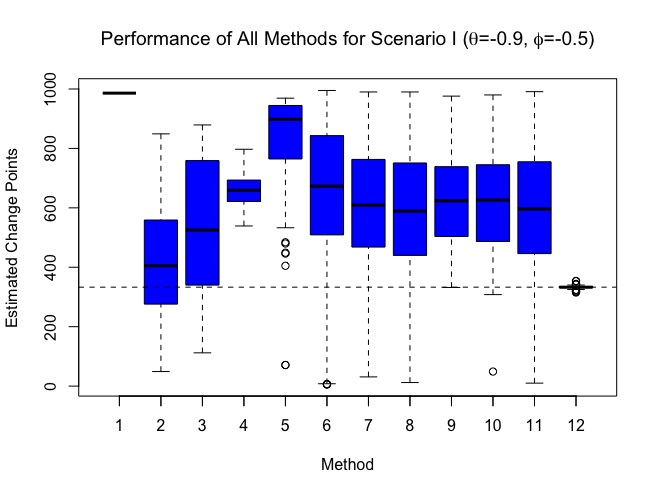} &
\includegraphics[width=0.4\textwidth]{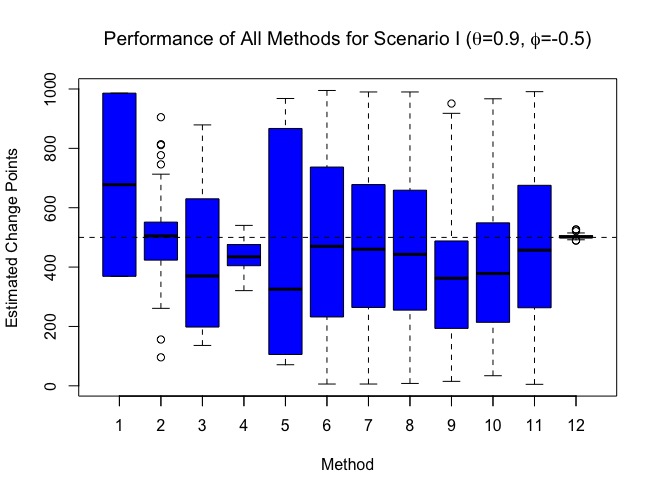} \\
\includegraphics[width=0.4\textwidth]{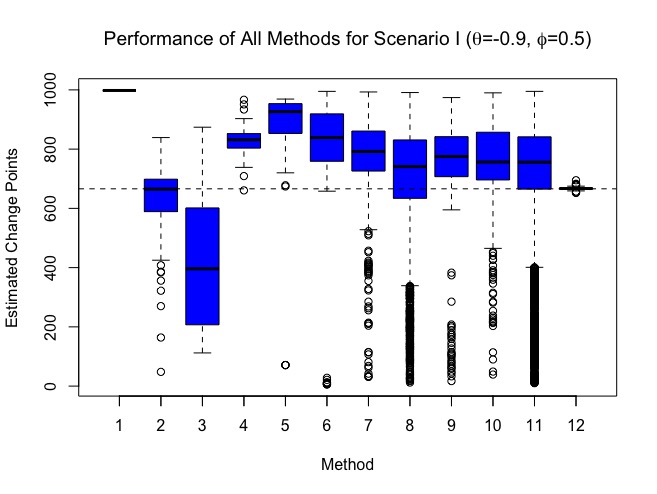} &
\includegraphics[width=0.4\textwidth]{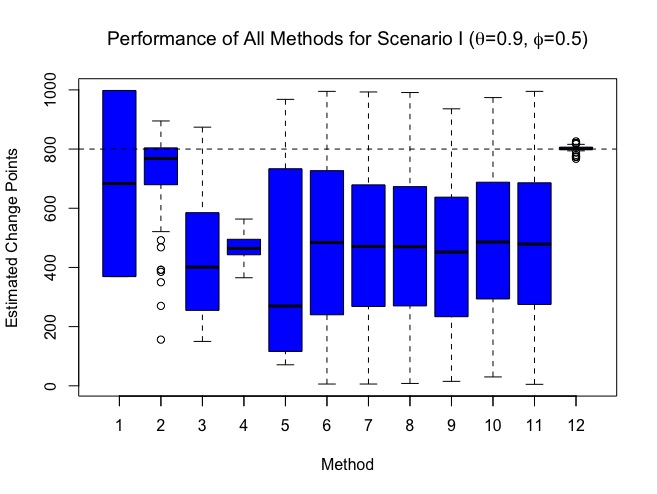}
\end{tabular}
\caption{Distributions of estimated change points for all the methods, where the ``true" values are specified with the dashed lines per case. Note that number 12 is our proposed method. Additionally, methods 1 and 3 have significant number of missing values. From top-left panel to bottom-right panel, the missing percentages are: (1) PELT $(99,98,99,98)\%$ and (3) ARC $(78,73,84,67)\%$.}
\label{Fig.comparing_alters}
\end{figure}

\begin{figure}[ht]
\centering
\begin{tabular}{ cc }
\includegraphics[width=0.5\textwidth]{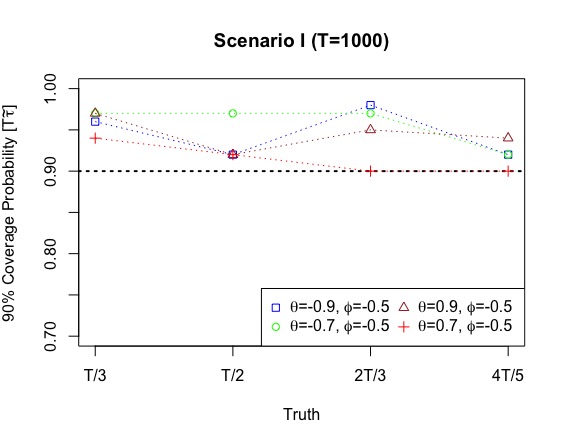} &
\includegraphics[width=0.5\textwidth]{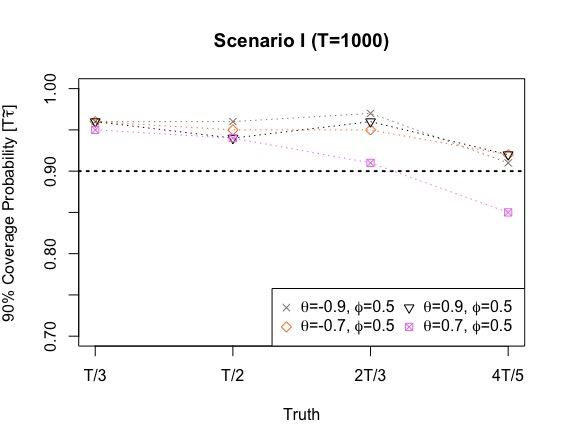} \\
\includegraphics[width=0.5\textwidth]{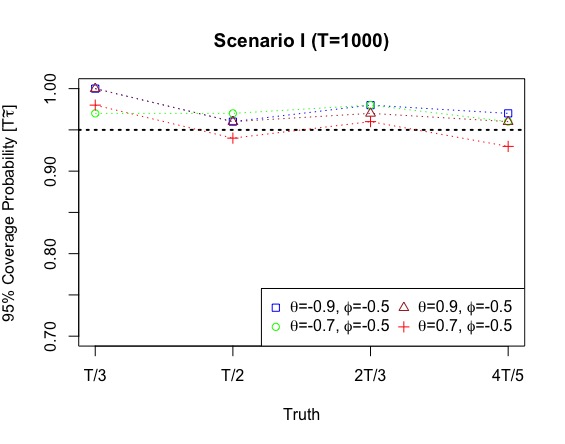} &
\includegraphics[width=0.5\textwidth]{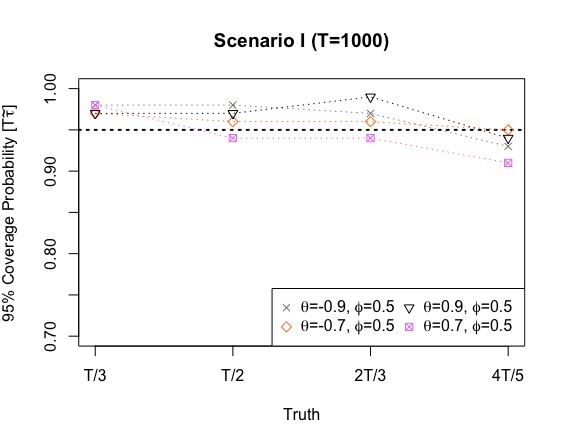} \\
\includegraphics[width=0.5\textwidth]{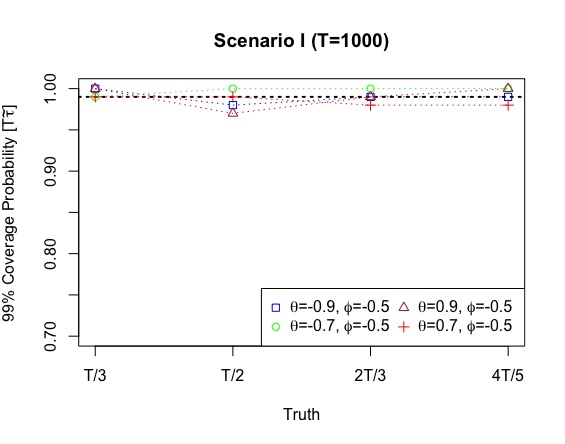} &
\includegraphics[width=0.5\textwidth]{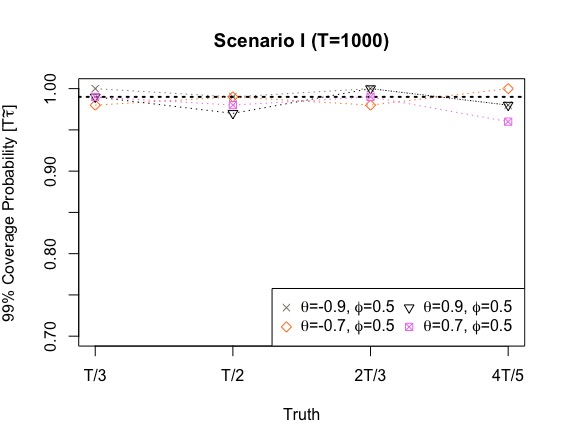} 
\end{tabular}
\caption{($90\%$, $95\%$, $99\%$) coverage probability in Scenario I with $T=1000$.}
\label{Fig.CP_CaseI}
\end{figure}

\subsection{Comparisons with some existing methods} \label{sec:sim_alters}

In this section, we report empirical comparisons with some alternative methods given in the literature and divide the section into two parts based on estimation and inference.

\noindent \textbf{Estimation Comparisons.} To make point estimation comparisons, we consider the following methods:
(1) pruned exact linear time (PELT) algorithm proposed by \cite{killick2012optimal};
(2) volatility change point estimator for diffusion processes based on the least squares (SDE) method proposed by \cite{iacus2008simulation};
(3) adversarially robust univariate mean change point detection (ARC) method proposed by \cite{li2021adversarially};
(4) dynamic programming method for the detection of univariate mean change points (DP-UNIVAR) proposed by \cite{wang2020univariate};
(5) local refinement for the univariate mean change point detection (LOCAL) method \citep{changepoints_package}; (6) standard binary segmentation for the univariate mean change point detection (STD) method \citep{changepoints_package}; 
(7) wild binary segmentation for the univariate mean change point detection (WBS UNI) method   \citep{changepoints_package};
(8) wild binary segmentation for univariate nonparametric change point detection (WBS UNI-NONPAR) proposed by \cite{madrid2021optimal};
(9) detection of univariate mean change points based on the standard binary segmentation with the tuning parameter (UNI TUNING1) method \citep{changepoints_package};
(10) detection of univariate mean change points based on wild binary segmentation with the tuning parameter (UNI TUNING2) method \citep{changepoints_package};
(11) robust wild binary segmentation for the univariate mean change point (ROB-WBS) method \citep{fearnhead2019changepoint}. Our method is denoted as method (12) in this comparison.

The PELT method is developed in the \texttt{changepoint} package in \texttt{R} written by \cite{killick2014changepoint}. Method 2 comes from the \texttt{sde} package developed by \cite{iacus2016package}. The rest of methods are in the \texttt{changepoints} package written by \cite{changepoints_package}. Note that all of these methods only provide point estimates for the change points, and they are not suitable for making an inference such as constructing confidence intervals unlike our method.

For the purpose of comparison, we consider the first scenario ``Scenario I" described earlier in this section. We assume $T=1000$ and the error terms follow $N(0,1)$. Then, we estimate the change point based on all the methods.  In Figure \ref{Fig.comparing_alters}, we show the distribution of estimated change points based on all methods using boxplots for some of the selected ``true" values. Based on these results, our method could perform significantly better than the alternatives.


\vspace{0.25cm} 

\noindent \textbf{Inference Comparisons.} For the inference comparisons, we use the recent method developed by \cite{fryzlewicz2024narrowest}. Specifically, we use the Narrowest Significance Pursuit (NSP) algorithm from the ``\texttt{nsp}" package in \texttt{R} developed by the author. We picked 2 cases from Scenario I and displayed our comparison results in Table \ref{Tab.CI_alters}. 
 
The NSP method provides multiple confidence intervals, and we take their unions (i.e. the sum of their lengths) per iteration to report the empirical average of the lengths of the confidence intervals in Table \ref{Tab.CI_alters}. Similarly, for coverage, the true value should be within at least one of the confidence intervals per iteration to record it as captured by the confidence interval. From Table \ref{Tab.CI_alters}, we observe that our method significantly outperforms the NSP method.

\begin{table}[ht]
\footnotesize
\caption{Comparison of coverage probabilities and average of confidence interval lengths for two methods. The results for the lengths are in the parentheses.} \label{Tab.CI_alters}
\centering
\setlength{\tabcolsep}{6pt} 
\begin{tabular}{@{} c|ccccc @{}} 
\hline
Scenario I & Truth & method & $90\%$ & $95\%$ & $99\%$   \\ \hline\hline
$(\theta=0.9, \phi=-0.5)$ & $T/2$ & NSP & 0.370 & 0.370 & 0.329   \\
& & & (279.130) & (275.920) & (191.210)\\
&& our method & 0.920 & 0.960 & 0.970 \\
& & & (50.220) & (60.420) & (85.850) \\
\hline
$(\theta=0.9, \phi=0.5)$ & $T/2$ & NSP & 0.273 & 0.343 & 0.271 \\
& & & (273.790) & (243.440) & (189.990) \\
&& our method & 0.940 & 0.970 & 0.970  \\ 
& & & (50.860) & (61.200) & (86.750) \\
\hline
\end{tabular}
\end{table}


\section{Real Data Applications} \label{sec:application}

In this section, we present two real-world applications to further assess the effectiveness of our methodology. The first involves an EEG dataset for epileptic seizure detection, while the second focuses on surveillance video analysis for human action recognition.


\subsection{Electroencephalogram (EEG) data} \label{sec:EEG}




To illustrate the nonstationarity inherent in some EEG signals, we consider electroencephalogram (EEG) data recorded from a patient diagnosed with left temporal lobe epilepsy during an epileptic seizure \citep{ombao2005slex}. The data consist of signals collected from 18 scalp locations, sampled over a duration of 228 seconds. Figure \ref{Fig.channels} displays the EEG recordings from three selected channels—P3, T3, and T5—positioned in the left temporal lobe, where the seizure activity was localized. Based on clinical assessment, the seizure is estimated to have occurred around t $\approx$ 85 s. The EEG recordings exhibit time-varying spectral characteristics, with notable changes in both magnitude and volatility occurring near the seizure onset. These features highlight the dynamic nature of the underlying brain activity and motivate the need for methodologies capable of capturing such temporal heterogeneity. In particular, data from channel P3 have been previously utilized by \cite{davis2006structural} and \cite{chan2014group} to identify structural breaks in time series, while all three channels have served as a case study in \cite{safikhani2018joint}.


In our analysis, we concentrate on the time interval $(1,110)$ seconds, which encompasses the primary seizure activity. We apply our algorithm to estimate the change points $ \lfloor T \tilde{\tau} \rfloor 
$ for the three channels of interest—P3, T3, and T5. The estimated change points are reported in Table \ref{Tab.EEG_seg1}, along with multiple confidence intervals computed at varying significance levels. The estimated values closely align with the true change points, and all corresponding confidence intervals successfully capture the true locations, demonstrating the strong inferential accuracy of the proposed method.

Figure \ref{Fig.channels} depicts the time series plots of these three channels while the confidence intervals are visualized as shaded areas. The bottom-right panel in Figure \ref{Fig.channels} describes the estimated spectral densities before and after the estimated change points, where before is estimated using the interval $(1, \lfloor T \tilde{\tau} \rfloor)$, and for after, we used additional data points (162,228) for description. The selection of the latter interval is justified using the results in \cite{safikhani2018joint} in which they report no detected change points after $t=162$. Looking at estimated spectral densities, we realize that there are significant changes before and after the occurrence of the seizure. We also used the NSP method to construct confidence intervals for all 3 channels and display the results in the Supplementary Material.
Additionally, we estimate the change points based on the competing methods listed in Section~\ref{sec:sim_alters}, and the results are given in the Supplementary Material (Section E). Among them, the SDE method performs better than the rest. Most of the competing methods detected multiple change points in the channels and the estimated change point locations are similar except the PELT method for Channel T5; see more details in Section~E.1 in the Supplementary Material. In addition, the NSP method provided multiple time intervals as confidence sets for the location of the true change point, while the overall length of these intervals when combined is much larger than the length of the confidence interval provided by the proposed method.

\begin{table}[ht]
\footnotesize
\caption{ Location of estimated change points in the EEG data set with their CIs for the segment $(1,110)$ based on our method. True change point is $\lfloor T \tau \rfloor = 85$.} \label{Tab.EEG_seg1}
\centering
\setlength{\tabcolsep}{6pt} 
\begin{tabular}{@{} ccccccc @{}}
\hline
Channel & value ($\lfloor T \tilde{\tau} \rfloor$) & $70\%$ CI & $80\%$ CI & $90\%$ CI & $95\%$ CI & $99\%$ CI   \\ \hline\hline
P3 & 90 & (64,96) & (63,100) & (60,107) & (58,115) & (52,136) \\
T3 & 86 & (85,87) & (85,87) & (85,91) & (85,98) & (85,118)  \\
T5 & 89 & (37,94) & (34,98) & (29,106) & (24,115) & (12,139)  \\
\hline
\end{tabular}
\end{table}

\begin{figure}[ht]
\centering
\begin{tabular}{ cc }
\includegraphics[width=0.5\textwidth]{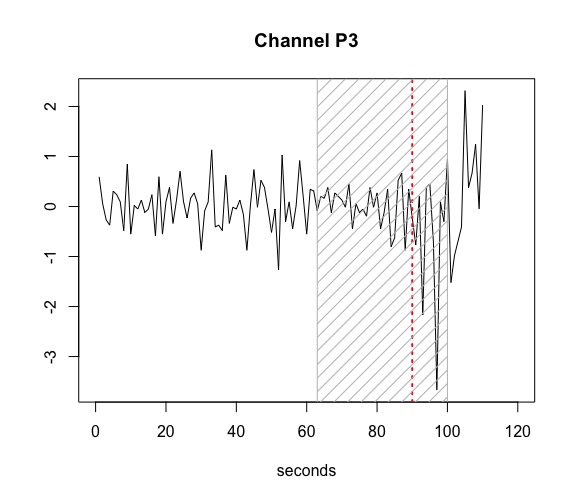} &
\includegraphics[width=0.5\textwidth]{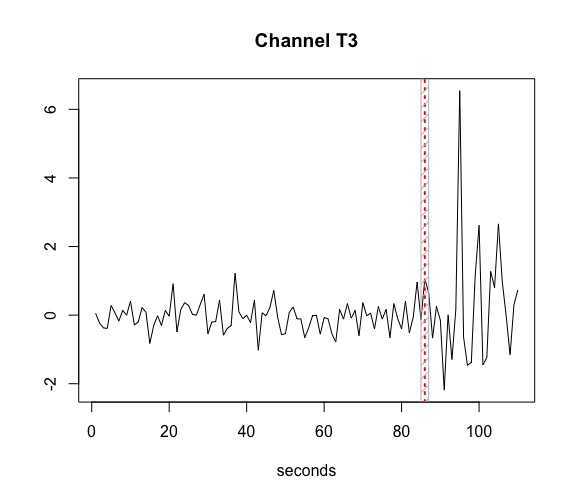} \\
\includegraphics[width=0.5\textwidth]{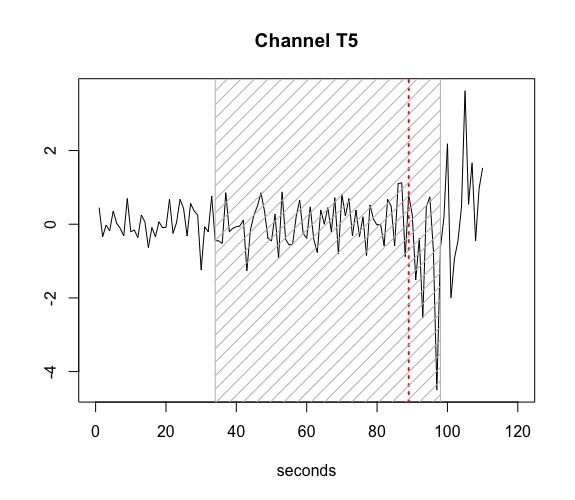} &
\includegraphics[width=0.5\textwidth]{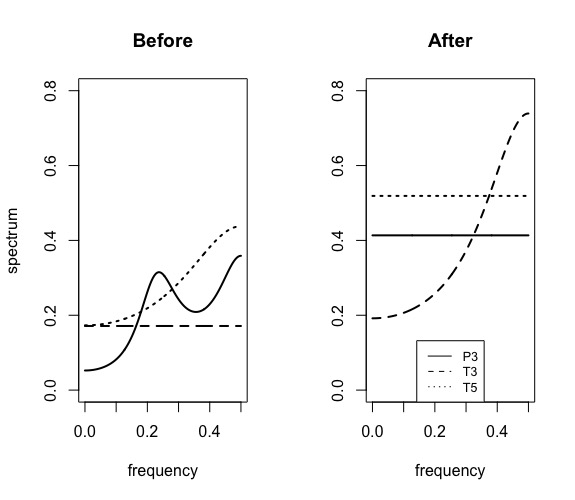}
\end{tabular}
\caption{EEG data set over 110 seconds with their estimated change points using our method. The $\% 80$ CIs are overlaid on the plots through shaded parts. In the last panel, we describe the spectral densities before and after the estimated change points for all the channels. Note that we used additional points (162,228) for the description of spectral density of after, and before is simply related to the interval of (1,$\lfloor T \tilde{\tau} \rfloor$).}
\label{Fig.channels}
\end{figure}

\subsection{Surveillance video data set} \label{sec:image}


In this section, we evaluate the performance of our methodology on a distinct data set consisting of surveillance video frames. Specifically, we use a sequence of 828 images obtained from the CAVIAR project\footnote{\url{http://homepages.inf.ed.ac.uk/rbf/CAVIARDATA1/}}. Each frame captures various human activities in different contexts, such as walking alone, meeting others, and entering or exiting a room. The original image resolution is $384 \times 288$ pixels, which we downsample to $24 \times 32$ pixels to facilitate computational efficiency. We focus our analysis on two distinct action sequences with known ground truth locations \citep{bai2020multiple}: (1) $ \lfloor T \tau^\star \rfloor=116$ (first person walks out of the lobby); (2) $\lfloor T \tau^\star \rfloor=174$ (second person walks into the lobby).


The results are presented in Figures E.1 and E.2 of the Supplementary Material, which illustrate that the pixels corresponding to the action regions are accurately identified. Additionally, the estimated change points and their associated confidence intervals are summarized in Tables \ref{Tab.Image1} and \ref{Tab.Image2}. Overall, the proposed method demonstrates strong performance. For comparison, we also apply several alternative change point detection methods to evaluate the robustness and effectiveness of our approach. Unlike SDE and LOCAL methods, all the other methods estimated multiple change points for two actions. For the first action, the ROB-WBS method performs better than the rest, and for the second action, DP-UNIVAR and STD outperform the rest. We refer to Section~E.2 of the Supplementary Material for further details and comparisons.

\begin{table}[ht]
\footnotesize
\caption{Location of estimated change points from the video data set with their CIs for the ``\textit{first person walks out of the lobby}". Note that the true value is $116$. } \label{Tab.Image1}
\centering
\setlength{\tabcolsep}{6pt} 
\begin{tabular}{@{} ccccccc @{}}
\hline
pixel & value ($\lfloor T \tilde{\tau} \rfloor$) & $70\%$ CI & $80\%$ CI & $90\%$ CI & $95\%$ CI & $99\%$ CI   \\ \hline\hline
700 & 114 & (113,120) & (113,121) & (113,125) & (113,129) & (113,139) \\
702 & 115 & (114,121) & (114,123) & (114,127) & (114,131) & (114,142) \\
731 & 109 & (108,123) & (108,127) & (108,136) & (108,147) & (108,173) \\
762 & 106 & (105,138) & (105,150) & (105,172) & (105,181) & (105,181) \\
764 & 118 & (117,169) & (117,169) & (117,169) & (117,169) & (117,169) \\
\hline
\end{tabular}
\end{table}

\clearpage

\begin{table}[ht]
\footnotesize
\caption{Location of estimated change points from the video data set with their CIs for the ``\textit{second person walks into the lobby}". Note that the true value is 174. } \label{Tab.Image2}
\centering
\setlength{\tabcolsep}{6pt} 
\begin{tabular}{@{} ccccccc @{}}
\hline
pixel & value ($\lfloor T \tilde{\tau} \rfloor$) & $70\%$ CI & $80\%$ CI & $90\%$ CI & $95\%$ CI & $99\%$ CI   \\ \hline\hline
48 & 178 & (176,188) & (176,191) & (176,198) & (175,206) & (175,226) \\
78 & 179 & (175,184) & (175,186) & (175,191) & (175,196) & (174,209) \\
110 & 181 & (168,187) & (165,189) & (160,194) & (156,198) & (145,211) \\
174 & 171 & (170,172) & (170,172) & (170,172) & (170,172) & (170,333) \\
209 & 174 & (170,185) & (169,190) & (169,199) & (168,209) & (168,234) \\
241 & 166 & (157,176) & (156,181) & (156,191) & (155,202) & (155,230) \\
\hline
\end{tabular}
\end{table}



\section{Discussion and Future Work} \label{sec:discussion}

In this paper, we proposed a nonparametric approach for detecting change points in the spectral density of time series, motivated by the need for precise seizure detection in EEG data. By leveraging the Wold decomposition, we modeled time series as autoregressive processes with infinite lags, truncating, and estimating them around potential change points. Our detection method, based on an initial estimator followed by an optimal rate procedure, enables accurate localization with quantifiable error bounds. The asymptotic distribution of the estimator further allows the construction of confidence intervals, which enhances reliability. Through extensive numerical experiments, we demonstrated the superior performance of the method over existing approaches and illustrated its practical effectiveness in EEG-based seizure detection and event detection in video data.  

Future work includes extending the method to detect multiple change points using techniques such as fused lasso regularization or rolling window schemes, which could further improve adaptability in dynamic environments. Additionally, the framework can be extended to multivariate time series by considering infinite-lag multivariate autoregressive models. Understanding the impact of dimensionality and quantifying its effect on localization error remain an important direction for further research. These extensions will broaden the applicability of the proposed approach, making it more versatile in complex high-dimensional settings.


\bigskip
\begin{center}
{\large\bf SUPPLEMENTARY MATERIAL}
\end{center}

\begin{description}
\item[Supplementary Material:] The file includes the proof of the main result and additional numerical results. (pdf)

\item[Code:] The file includes the R codes for the implementation of the proposed detection method. (zip)
\end{description}

\bibliographystyle{chicago}  

\bibliography{Bibliography}

\newpage
\clearpage

\def\spacingset#1{\renewcommand{\baselinestretch}%
{#1}\small\normalsize} \spacingset{1}

\begin{center}

\LARGE{Supplementary Material for \hspace{.2cm}\\ ``Optimal Change Point Detection and Inference in the Spectral Density of General Time Series Models"}

\end{center}

\spacingset{1.9} 

\vspace{1cm}

\setcounter{secnumdepth}{0}

This supplementary material is structured as follows. In Section A, we provide detailed proofs of the theoretical results of Section 3 of the manuscript. In Section B, we provide the proofs of the theoretical results of Section 4 of the manuscript. Section C contains auxiliary lemmas and their proofs. In Section D, we list additional simulation results, and in Section E, we provide additional numerical results from the applications part of the manuscript. Note that based on the manuscript and to simplify the notation, we refer to both error terms for before and after the change point as $\epsilon_t$ if there are no ambiguities. 

\clearpage

\section{A \quad Proofs of Theoretical Results in Section 3} \label{App_A}
\renewcommand{\theequation}{A.\arabic{equation}}
\setcounter{equation}{0}
\renewcommand{\thelemma}{A.\arabic{lemma}}
\renewcommand{\thedefinition}{A.\arabic{definition}}

\begin{proof}[Proof of Theorem 3.1.]

The main idea is to use the basic inequality:
\begin{equation*}
    \mathcal{L}(\hat{\tau}) \leq \mathcal{L}(\tau^\star),
\end{equation*}
to derive the error bound of estimated change point. Therefore, we need to find the upper bound of $\mathcal{L}(\tau^\star)$ as well as the lower bound of $\mathcal{L}(\hat{\tau})$.

Without loss of generality, we assume that $\tau > \tau^\star$, then for the interval $\left[1, \lfloor T\tau \rfloor \right]$, notice that the data in $\left[ \lfloor T\tau^\star \rfloor, \lfloor T\tau \rfloor \right]$ is generated from $f_2(\lambda)$ while we need to use it to fit an AR process combined with data generated by $f_1(\lambda)$. Therefore, we consider the fitting procedure demonstrated in (5)-(6):

For the data $\mathcal{T}(1,\lfloor T\tau \rfloor) = \{X_1, \dots, X_{\lfloor T\tau^\star \rfloor}, X_{\lfloor T\tau^\star \rfloor +1}, \dots, X_{\lfloor T\tau \rfloor}\}$, we first define the sample covariances $\hat{\gamma}_{1,k}^{(1)}$ and $\hat{\gamma}_{1,k}^{(2)}$ as follows:

\begin{equation}
    \begin{aligned}
        \hat{\gamma}_{1,k}^{(1)} &= \frac{1}{\lfloor T\tau^\star \rfloor }\sum_{t=1}^{\lfloor T\tau^\star \rfloor-|k|}(X_t - \Bar{X})(X_{t+|k|} - \Bar{X}),\quad t=1,2,\dots, \lfloor T\tau^\star \rfloor, \\
        \hat{\gamma}_{1,k}^{(2)} &= \frac{1}{\lfloor T\tau \rfloor - \lfloor T\tau^\star \rfloor}\sum_{t=\lfloor T\tau^\star \rfloor-|k|+1}^{\lfloor T\tau \rfloor -|k|}(X_t - \Bar{X})(X_{t+|k|} - \Bar{X}),\quad t=\lfloor T\tau^\star \rfloor+1, \dots, \lfloor T\tau \rfloor,
    \end{aligned}
\end{equation}
where $k=1,2,\dots, p$ indicates the lag-$k$ autocovariances. Recall the definition of sample covariance $\hat{\gamma}_{1,k}$ based on data $\mathcal{T}(1,\lfloor T\tau \rfloor)$, we can easily derive that:
\begin{equation}
    \hat{\gamma}_{1,k} = \frac{\tau^\star}{\tau}\hat{\gamma}_{1,k}^{(1)} + \frac{\tau-\tau^\star}{\tau}\hat{\gamma}_{1,k}^{(2)},
\end{equation}
furthermore, we obtain that

\begin{equation}
    \hat{\Gamma}_1 = \frac{\tau^\star}{\tau}\hat{\Gamma}_1^{(1)} + \frac{\tau-\tau^\star}{\tau}\hat{\Gamma}_1^{(2)},\quad \hat{\gamma}_1 = \frac{\tau^\star}{\tau}\hat{\gamma}_1^{(1)} + \frac{\tau-\tau^\star}{\tau}\hat{\gamma}_1^{(2)},
\end{equation}
where $\hat{\gamma}_1^{(1)}$ and $\hat{\gamma}_1^{(2)}$ are vectors of stacking $\hat{\gamma}_{1,k}^{(1)}$ and $\hat{\gamma}_{1,k}^{(2)}$, respectively, as well as the matrices $\hat{\Gamma}_1^{(1)} = (\hat{\gamma}_{1,|i-j|}^{(1)})_{i,j=1}^p$ and $\hat{\Gamma}_1^{(2)} = (\hat{\gamma}_{2,|i-j|}^{(2)})_{i,j=1}^p$. Then, by using Yule-Walker equations, we have the fitted AR($p$) parameter vector for the interval $\left[1, \lfloor T\tau \rfloor \right]$ in the following:
\begin{equation} \label{eq.phi1_hat}
    \begin{aligned}
        \hat{\phi}_1 &= \left(\frac{\tau^\star}{\tau}\hat{\Gamma}_1^{(1)} + \frac{\tau-\tau^\star}{\tau}\hat{\Gamma}_1^{(2)}\right)^{-1}\left(\frac{\tau^\star}{\tau}\hat{\gamma}_1^{(1)} + \frac{\tau-\tau^\star}{\tau}\hat{\gamma}_1^{(2)}\right) \\
        &= \left(\mathbf{I}_p + \frac{\tau-\tau^\star}{\tau^\star}\hat{\Gamma}_1^{(2)}\hat{\Gamma}_1^{(1)^{-1}}\right)^{-1}\left(\hat{\phi}_1^{(1)} + \frac{\tau-\tau^\star}{\tau^\star}\hat{\Gamma}_1^{(1)^{-1}}\hat{\gamma}_1^{(2)}\right),
    \end{aligned}
\end{equation}
where $\hat{\phi}_1^{(1)}$ is the fitted AR($p$) model parameter vector on the interval $\mathcal{T}(1, \lfloor T\tau^\star \rfloor)$. 

Then, by directly using the Searle's identity (see \citealp{searle2017matrix} Page 151), we have: 
\begin{equation}
    \hat{\phi}_1 = \left(\mathbf{I}_p - \frac{\tau-\tau^\star}{\tau^\star}\hat{\Gamma}_1^{(2)}\left(\mathbf{I}_p + \hat{\Gamma}_1^{(1)^{-1}}\hat{\Gamma}_1^{(2)} \right)^{-1}\hat{\Gamma}_1^{(1)^{-1}}\right)\left(\hat{\phi}_1^{(1)} + \frac{\tau-\tau^\star}{\tau^\star}\hat{\Gamma}_1^{(1)^{-1}}\hat{\gamma}_1^{(2)}\right),
\end{equation}
therefore, we obtain that
\begin{equation*}
    \begin{aligned}
        \hat{\phi}_1 &= \hat{\phi}_1^{(1)} + \left(\mathbf{I}_p - \frac{\tau-\tau^\star}{\tau^\star}\hat{\Gamma}_1^{(2)}\left(\mathbf{I}_p + \hat{\Gamma}_1^{(1)^{-1}}\hat{\Gamma}_1^{(2)} \right)^{-1}\hat{\Gamma}_1^{(1)^{-1}}\right)\frac{\tau-\tau^\star}{\tau^\star}\hat{\Gamma}_1^{(1)^{-1}}\hat{\gamma}_1^{(2)} \\
        &- \frac{\tau-\tau^\star}{\tau^\star}\hat{\Gamma}_1^{(2)}\left(\mathbf{I}_p + \hat{\Gamma}_1^{(1)^{-1}}\hat{\Gamma}_1^{(2)} \right)^{-1}\hat{\Gamma}_1^{(1)^{-1}}\hat{\phi}_1^{(1)},
    \end{aligned}
\end{equation*}
by using the similar equation, we can also derive that:
\begin{equation*}
    \begin{aligned}
        \hat{\phi}_1 &= \hat{\phi}_1^{(2)} + \left(\mathbf{I}_p - \frac{\tau^\star}{\tau-\tau^\star}\hat{\Gamma}_1^{(1)}\left(\mathbf{I}_p + \hat{\Gamma}_1^{(2)^{-1}}\hat{\Gamma}_1^{(1)}\right)^{-1}\hat{\Gamma}_1^{(2)^{-1}}\right)\frac{\tau^\star}{\tau-\tau^\star}\hat{\Gamma}_1^{(2)^{-1}}\hat{\gamma}_1^{(1)} \\
        &- \frac{\tau^\star}{\tau-\tau^\star}\hat{\Gamma}_1^{(1)}\left(\mathbf{I}_p + \hat{\Gamma}_1^{(2)^{-1}}\hat{\Gamma}_1^{(1)}\right)^{-1}\hat{\Gamma}_1^{(2)^{-1}}\hat{\phi}_1^{(2)}.
    \end{aligned}
\end{equation*}

Based on the detection algorithm, we consider the introduced objective function $\mathcal{L}(\tau)$ and it can be written as follows: 
\begin{equation}
    \label{eq:13}
    \mathcal{L}(\tau) = \sum_{t=1}^{\lfloor T\tau \rfloor} (X_t - \hat{\phi}_1^\prime Z_t)_2^2 + \sum_{t=\lfloor T\tau \rfloor+1}^T(X_t - \hat{\phi}_2^\prime Z_t)_2^2 \overset{\text{def}}{=} I_1 + I_2.
\end{equation}


Note that the estimated AR($p$) coefficients $\hat{\phi}_1$ and $\hat{\phi}_2$ are functions of $\tau$. By using the result from Corollary 2 in \cite{wang2021consistent} as well as the Condition A(c), we first derive the lower bound for $I_2$:
\begin{align}
    \label{eq:20}
    I_2 & = \sum_{t=\lfloor T\tau \rfloor+1}^T(X_t - \hat{\phi}_2^\prime Z_t)^2 
        \geq \sum_{t=\lfloor T\tau\rfloor +1}^T \epsilon_t^2 + \sum_{t=\lfloor T\tau \rfloor +1}^T \left((\hat{\phi}_2 - \phi^\star_2)^\prime Z_t \right)^2 - 2\left|\sum_{t=\lfloor T\tau \rfloor+1}^T \epsilon_t(\hat{\phi}_2 - \phi^\star_2)^\prime Z_t \right| \nonumber \\
        &\geq \sum_{t=\lfloor T\tau \rfloor+1}^T \epsilon_t^2 + c_2(T-\lfloor T\tau \rfloor)\|\hat{\phi}_2 - \phi^\star_2\|_2^2 - 2c_2^\prime \left\|\sum_{t=\lfloor T\tau \rfloor+1}^T Z_t \epsilon_t \right\|_\infty\|\hat{\phi}_2 - \phi^\star_2\|_1 \nonumber \\
        &\geq \sum_{t=\lfloor T\tau \rfloor+1}^T \epsilon_t^2 + c_2(T- \lfloor T\tau \rfloor)\|\hat{\phi}_2 - \phi^\star_2\|_2\left(\|\hat{\phi}_2 - \phi^\star_2\|_2 - \frac{2c_2^\prime\sqrt{p}}{c_2}\left\| \frac{1}{T-\lfloor T\tau \rfloor}\sum_{t=\lfloor T\tau \rfloor+1}^T Z_t \epsilon_t \right\|_\infty\right) \nonumber \\
        &\geq \sum_{t=\lfloor T\tau \rfloor+1}^T \epsilon_t^2 + c_2(T-\lfloor T\tau \rfloor)\|\hat{\phi}_2 - \phi^\star_2\|_2\left(\|\hat{\phi}_2 - \phi^\star_2\|_2 - \frac{2c_2^\prime}{c_2}\sqrt{\frac{p^2\log p}{T-\lfloor T\tau \rfloor}} - c_3\sqrt{p}\sum_{k=p+1}^\infty(1+k)^r|\phi^\star_{2,k}|\right) \nonumber  \\
        &\geq \sum_{t=\lfloor T\tau \rfloor+1}^T\epsilon_t^2 + c_2(T-\lfloor T\tau \rfloor)\|\hat{\phi}_2 - \phi^\star_2\|_2\left(\|\hat{\phi}_2 - \phi^\star_2\|_2 - \frac{2c_2^\prime}{c_2}\sqrt{\frac{p^2\log p}{T-\lfloor T\tau \rfloor}} - c_3\sqrt{p}\sqrt{\frac{p\log p}{T-\lfloor T\tau \rfloor}}\right) \nonumber \\
        &\geq \sum_{t=\lfloor T\tau \rfloor+1}^T \epsilon_t^2 - C_2p^2\log p,
\end{align}

where the $C$ terms are generic constants which may differ at each appearance.

Similarly, to calculate the lower bound for $I_1$, we have:
\begin{align}\label{eq:21}
    I_1 &= \sum_{t=1}^{\lfloor T\tau^\star \rfloor}(X_t - \hat{\phi}_1^\prime Z_t)^2 + \sum_{t=\lfloor T\tau^\star \rfloor+1}^{\lfloor T\tau \rfloor} (X_t - \hat{\phi}_1^\prime Z_t)^2 \nonumber \\
    &\geq \sum_{t=1}^{\lfloor T\tau \rfloor} \epsilon_t^2 + \sum_{t=1}^{\lfloor T\tau^\star \rfloor}((\hat{\phi}_1 - \phi^\star_1)^\prime Z_t)^2 + \sum_{t=\lfloor T\tau^\star \rfloor+1}^{\lfloor T\tau \rfloor}((\hat{\phi}_1 - \phi^\star_2)^\prime Z_t)^2 - 2\left| \sum_{t=1}^{\lfloor T\tau^\star \rfloor}\epsilon_t(\hat{\phi}_1 - \phi^\star_1)^\prime Z_t \right| \nonumber \\
    &- 2\left| \sum_{t=\lfloor T \tau^\star \rfloor+1}^{\lfloor T\tau\rfloor} \epsilon_t(\hat{\phi}_1 - \phi^\star_2)^\prime Z_t \right| \nonumber \\
    &\geq \sum_{t=1}^{\lfloor T\tau \rfloor} \epsilon_t^2 + c_1\lfloor T\tau^\star \rfloor\|\hat{\phi}_1 - \phi^\star_1\|_2^2 + c_3(\lfloor T\tau\rfloor -\lfloor T\tau^\star \rfloor)\|\hat{\phi}_1 - \phi^\star_2\|_2^2 - 2c_1^\prime\left\| \sum_{t=1}^{\lfloor T\tau^\star \rfloor}Z_t\epsilon_t \right\|_\infty\|\hat{\phi}_1 - \phi^\star_1\|_1 \nonumber \\
    &- 2c_3^\prime\left\| \sum_{t=\lfloor T\tau^\star \rfloor+1}^{\lfloor T\tau\rfloor} Z_t \epsilon_t \right\|_\infty\|\hat{\phi}_1 - \phi^\star_2\|_1 \nonumber \\
    &\geq \sum_{t=1}^{\lfloor T\tau \rfloor} \epsilon_t^2 + c_1\lfloor T\tau^\star\rfloor\|\hat{\phi}_1 - \phi^\star_1\|_2\left(\|\hat{\phi}_1 - \phi^\star_1\|_2 - \frac{2c_1^\prime}{c_1}\sqrt{\frac{p^2\log p}{\lfloor T\tau^\star \rfloor}} - c_2\sqrt{p}\sum_{k=p+1}^\infty(1+k)^r|\phi^\star_{1,k}|\right) \nonumber \\
    &+ c_3(\lfloor T\tau \rfloor - \lfloor T\tau^\star \rfloor)\|\hat{\phi}_1 - \phi^\star_2\|_2\left(\|\hat{\phi}_1 - \phi^\star_2\|_2 - \frac{2c_3^\prime}{c_3}\sqrt{\frac{p^2\log p}{\lfloor T\tau \rfloor - \lfloor T\tau^\star \rfloor}}\right) \nonumber \\
    &\geq \sum_{t=1}^{\lfloor T\tau \rfloor} \epsilon_t^2 + c_1\lfloor T\tau^\star \rfloor\|\hat{\phi}_1 - \phi^\star_1\|_2\left(\|\hat{\phi}_1 - \phi^\star_1\|_2 - \frac{2c_1^\prime}{c_1}\sqrt{\frac{p^2\log p}{\lfloor T\tau^\star \rfloor}} - c_2\sqrt{\frac{p^2\log p}{\lfloor T\tau^\star \rfloor}}\right) \nonumber \\
    &+ c_3(\lfloor T\tau \rfloor - \lfloor T\tau^\star \rfloor)\|\hat{\phi}_1 - \phi^\star_2\|_2\left(\|\hat{\phi}_1 - \phi^\star_2\|_2 - \frac{2c_3^\prime}{c_3}\sqrt{\frac{p^2\log p}{\lfloor T\tau \rfloor - \lfloor T\tau^\star \rfloor}}\right).
\end{align}

Since we have $\|\phi_1^\star - \phi_2^\star\|_2 = \xi_2 > 0$, then either $\|\hat{\phi}_1 - \phi^\star_1\|_2 \geq \xi_2/4 > 0$ or $\|\hat{\phi}_1 - \phi^\star_2\|_2 \geq \xi_2/4 > 0$ holds. Therefore, by substituting the corresponding quantities in \eqref{eq:21}, we finally derive that
\begin{equation}\label{eq:22}
    I_1 \geq \sum_{t=1}^{\lfloor T\tau \rfloor} \epsilon_t^2 + C_1\Big|\lfloor T\tau \rfloor - \lfloor T\tau^\star \rfloor\Big|\xi_2^2 - C_3p^2\log p.
\end{equation}

Combining the results from \eqref{eq:20} and \eqref{eq:22}, we have the lower bound of $\mathcal{L}(\hat{\tau})$:
\begin{equation}
    \mathcal{L}(\hat{\tau}) \geq \sum_{t=1}^T\epsilon_t^2 + C_1\Big|\lfloor T\hat{\tau} \rfloor - \lfloor T\tau^\star \rfloor\Big|\xi_2^2 - (C_2+C_3)p^2\log p.
\end{equation}

On the other hand, we are going to calculate the upper bound for the $\mathcal{L}(\tau^\star)$. To see the result, we analogously use the procedure above. Thus, we have:
\begin{equation*}
    \mathcal{L}(\tau^\star) = \sum_{t=1}^{\lfloor T\tau^\star \rfloor}(X_t - \hat{\phi}_1^\prime Z_t)^2 + \sum_{t=\lfloor T\tau^\star \rfloor+1}^T (X_t - \hat{\phi}_2^\prime Z_t)^2 \overset{\text{def}}{=} J_1+J_2.
\end{equation*}
Then we have:
\begin{align}
    J_1 &= \sum_{t=1}^{\lfloor T\tau^\star \rfloor}\epsilon_t^2 + 2\left| \sum_{t=1}^{\lfloor T\tau^\star \rfloor}\epsilon_t(\hat{\phi}_1 - \phi^\star_1)^\prime Z_t \right| + \sum_{t=1}^{\lfloor T\tau^\star \rfloor}((\hat{\phi}_1 - \phi^\star_1)^\prime Z_t)^2 \nonumber \\
    &\leq \sum_{t=1}^{\lfloor T\tau^\star \rfloor}\epsilon_t^2 + 2\left\| \sum_{t=1}^{\lfloor T\tau^\star \rfloor}Z_t\epsilon_t \right\|_\infty\|\hat{\phi}_1 - \phi^\star_1\|_1 + c_1\lfloor T\tau^\star \rfloor \|\hat{\phi}_1 - \phi^\star_1\|_2^2 \nonumber \\
    &\leq \sum_{t=1}^{\lfloor T\tau^\star \rfloor}\epsilon_t^2 + 2c_1^\prime\lfloor T\tau^\star \rfloor\sqrt{\frac{p\log p}{\lfloor T\tau^\star \rfloor}}\sqrt{\frac{p^2\log p}{\lfloor T\tau^\star \rfloor}} + c_1p^2\log p \nonumber \\
    &\leq \sum_{t=1}^{\lfloor T\tau^\star \rfloor}\epsilon_t^2 + 2c_1^\prime p^{\frac{3}{2}}\log p + c_1 p^2\log p \leq \sum_{t=1}^{\lfloor T\tau^\star \rfloor}\epsilon_t^2 + K_1p^2\log p,
\end{align}
where $K_1$ is a constant. 


Then, we obtain the upper bound for $J_2$ by using the similar procedure:
\begin{equation}
    J_2 \leq \sum_{t=\lfloor T\tau^\star \rfloor+1}^T\epsilon_t^2 + K_2p^2\log p,
\end{equation}
where $K_2$ is a constant.
Hence, according to the basic inequality together with (\ref{eq.phi1_hat}-\ref{eq:13}), 
we have:
\begin{equation}
    \sum_{t=1}^T\epsilon_t^2 + C_1\Big|\lfloor T\hat{\tau} \rfloor - \lfloor T\tau^\star \rfloor \Big|\xi^2 - (C_2+C_3)p^2\log p \leq \sum_{t=1}^T \epsilon_t^2 + (K_1+K_2)p^2\log p,
\end{equation}
which leads to the final result by choosing $K_0 = (C_2+C_3+K_1+K_2)/C_1$.

\end{proof}

\section{B \quad Proofs of Theoretical Results in Section 4} \label{App_B}
\renewcommand{\theequation}{B.\arabic{equation}}
\setcounter{equation}{0}
\renewcommand{\thelemma}{B.\arabic{lemma}}
\renewcommand{\thedefinition}{B.\arabic{definition}}

\begin{proof}[Proof of Lemma 4.1.]
Denote that $\hat{\eta} = \hat{\phi}_1 - \hat{\phi}_2$, and for any fixed $\tau \geq \tau^\star$, that:
\begin{equation}
    \begin{aligned}
        &\mathcal{U}(\tau; \hat{\phi}_1, \hat{\phi}_2) = Q(\tau; \hat{\phi}_1, \hat{\phi}_2) - Q(\tau^\star; \hat{\phi}_1, \hat{\phi}_2) \\
        &= \frac{1}{T-p+1}\left\{ \sum_{t=p}^{\lfloor T\tau \rfloor}(X_t - \hat{\phi}_1^\prime Z_t)^2 + \sum_{t=\lfloor T\tau \rfloor + 1}^{T} (X_t - \hat{\phi}_2^\prime Z_t)^2 \right\} \\
        &- \frac{1}{T-p+1}\left\{ \sum_{t=p}^{\lfloor T\tau^\star\rfloor}(X_t - \hat{\phi}_1^\prime Z_t)^2 + \sum_{t=\lfloor T\tau^\star\rfloor + 1}^T(X_t - \hat{\phi}_2^\prime Z_t)^2 \right\} \\
        &= \frac{1}{T-p+1}\left\{ \sum_{t=\lfloor T\tau^\star \rfloor + 1}^{\lfloor T\tau \rfloor}(X_t - \hat{\phi}_1^\prime Z_t)^2 - \sum_{t=\lfloor T\tau^\star \rfloor+1}^{\lfloor T\tau \rfloor}(X_t - \hat{\phi}_2^\prime Z_t)^2 \right\} \\
        &= \frac{1}{T-p+1}\sum_{t=\lfloor T\tau^\star\rfloor + 1}^{\lfloor T\tau \rfloor}\left\{ (\hat{\eta}^\prime Z_t)^2 - 2\epsilon_t\hat{\eta}^\prime Z_t + 2(\hat{\phi}_2 - \phi_2^\star)^\prime Z_tZ_t^\prime \hat{\eta} \right\}.
    \end{aligned}
\end{equation}


This algebraic rearrangement leads to the following result:
\begin{equation}\label{eq:partitioningJs}
    \begin{aligned}
        \inf_{\substack{\tau \in \mathcal{G}(u_T, v_T) \\ \tau \geq \tau^\star}}\mathcal{U}(\tau; \hat{\phi}_1, \hat{\phi}_2) 
        &\geq \inf_{\substack{\tau \in \mathcal{G}(u_T, v_T) \\ \tau \geq \tau^\star}}\frac{1}{T-p+1}\sum_{t=\lfloor T\tau^\star\rfloor+1}^{\lfloor T\tau\rfloor}(\hat{\eta}^\prime Z_t)^2 - \sup_{\substack{\tau\in \mathcal{G}(u_T, v_T) \\ \tau\geq \tau^\star}}\frac{2}{T-p+1}\left|\sum_{t=\lfloor T\tau^\star\rfloor+1}^{\lfloor T\tau\rfloor}\epsilon_t\hat{\eta}^\prime Z_t\right| \\
        &- \sup_{\substack{\tau \in \mathcal{G}(u_T, v_T) \\ \tau \geq \tau^\star}}\frac{2}{T-p+1}\left| \sum_{t=\lfloor T\tau^\star\rfloor+1}^{\lfloor T\tau \rfloor} (\hat{\phi}_2 - \phi_2^\star)^\prime Z_tZ_t^\prime \hat{\eta} \right| \overset{\text{def}}{=} J_1 - J_2 - J_3.
    \end{aligned}
\end{equation}

In the following proof, we define $T'=T-p+1$. However, for simplicity in writing the proof, we use $T$ and $T'$ interchangeably. By applying the result of Lemma \ref{lemma:c5}(i) and Lemma \ref{lemma:c6}, in particular, we have:
\begin{equation}
    \begin{aligned}
        J_1 &\geq \kappa_{\min} \xi_2^2\left( v_T - c_{a1}\sigma^2\sqrt{\frac{u_T}{T}} - c_u\frac{u_T}{\xi_2}p\sqrt{\log (p\vee T)}\|\hat{\eta} - \eta^\star\|_2 \right) \\
        &\geq \kappa_{\min}\xi_2^2\left( v_T - c_{a1}\sigma^2\sqrt{\frac{u_T}{T}} - c_u\frac{u_T}{T^b} \right).
    \end{aligned}
\end{equation}
Then for $J_2$, using the result of Lemma \ref{lemma:c2}, Lemma \ref{lemma:c5}(ii) and recall that $K := 2^{2/\gamma_2}(K_X + K_\epsilon)^2$ we have:
\begin{equation}
    \begin{aligned}
        J_2 &\leq c_{a1}2^{\frac{2}{\gamma_2}+1}K_\epsilon K_X\xi_2\sqrt{\frac{u_T}{T}} + c_uK\sqrt{\frac{u_T\log(p\vee T)}{T}}\|\hat{\eta} - \eta^\star\|_1 \\
        &\leq c_{a1}2^{\frac{2}{\gamma_2}+1}K_\epsilon K_X\xi_2\sqrt{\frac{u_T}{T}} + c_uc_{a1}K\xi_2\sqrt{\frac{u_T}{T}} \\
        &\leq \frac{c_{a1}}{2}\left(2^{\frac{2}{\gamma_2}}4K_{\epsilon}K_X\right)\xi_2\sqrt{\frac{u_T}{T}} + c_uc_{a1}K\xi_2\sqrt{\frac{u_T}{T}}
        \\
        &\leq \frac{c_{a1}}{2}K\xi_2\sqrt{\frac{u_T}{T}} + c_uc_{a1}K\xi_2\sqrt{\frac{u_T}{T}} = \left(c_u + \frac{1}{2}\right)c_{a1}K\xi_2\sqrt{\frac{u_T}{T}},
    \end{aligned}
\end{equation}


with probability at least $1 - a - o(1)$. Similarly, for the term $J_3$, according to Condition B(b), Condition C, Lemma \ref{lemma:c5}(iii) and Lemma \ref{lemma:c6}, we have:
\begin{equation}
    \begin{aligned}
        \frac{J_3}{\xi_2^2} &\leq 2c_u(\sigma^2\vee \alpha)u_T\frac{1}{\xi_2}p\sqrt{\log(p\vee T)}\|\hat{\phi}_2 - \phi_2^\star\|_2\left(1 + \frac{1}{\xi_2}p\sqrt{\log(p\vee T)}\|\hat{\eta} - \eta^\star\|_2\right) \\
        &\leq 2c_u(\sigma^2 \vee \alpha)u_T\left(\frac{1}{\xi_2}p\sqrt{\log(p\vee T)}\|\hat{\phi}_2 - \phi_2^\star\|_2\right)\left(1 + \frac{c_{u1}}{T^b}\right) \\
        &\leq 2c_u(\sigma^2 \vee \alpha)u_T\left(\frac{1}{\xi_2}c_u\sqrt{1 + \nu^2}\frac{\sigma^2}{\kappa_{\min}}\frac{p^2\log(p\vee T)}{\sqrt{Tl_T}}\right)\left(1 + \frac{c_{u1}}{T^b}\right) \\
        &\leq 2c_u u_T(\sigma^2 \vee \alpha)\frac{c_{u1}}{T^b}(1 + \frac{c_{u1}}{T^b}),
    \end{aligned}
\end{equation}


with probability at least $1-a-o(1)$. Therefore, by combining all these three results together and substitute the bounds in \eqref{eq:partitioningJs} yields the bound uniformly over the set $\{\mathcal{G}(u_T, v_T), \tau \geq \tau^\star\}$, which completes the proof.
\end{proof}

According to the results of Lemma 4.1, we are now in a position to demonstrate the proof of Theorem 4.1.
\begin{proof}[Proof of Theorem 4.1.]
To prove this theorem, we show that the bound
\begin{equation}
    \Big| \lfloor T\tilde{\tau} \rfloor - \lfloor T\tau^\star \rfloor \Big| \leq c_{a3}^2,
\end{equation}
holds with probability at least $1-3a-o(1)$. Notice that part (i) of Theorem 4.1 is a direct consequence of the bound in part (ii) with respect to $\xi_2 \to 0$ for sufficiently large $T$. We begin by considering any $v_T > 0$ and applying Lemma 4.1 on the set $\mathcal{G}(1, v_T)$ to obtain that:
\begin{equation*}
    \inf_{\substack{\tau\in \mathcal{G}(1, v_T) \\ \tau\geq \tau^\star}}\mathcal{U}(\tau; \hat{\phi}_1, \hat{\phi}_2) \geq \kappa_{\min}\xi_2^2\left[v_T - c_{a3}\max\left\{ \left(\frac{1}{T}\right)^{\frac{1}{2}}, \frac{1}{T^{2b}} \right\}\right],
\end{equation*}
with probability at least $1 - o(1)$. Recall that we define $0 < b < 1/2$ and choose any $v_T > v_T^\star = c_{a3}/T^b$. Then we have $\inf_{\tau \in \mathcal{G}(1,v_T)}\mathcal{U}(\tau; \hat{\phi}_1, \hat{\phi}_2) > 0$, which implies that $\tilde{\tau} \not\in \mathcal{G}(1,v_T^\star)$, i.e., $|\lfloor T\tilde{\tau}\rfloor - \lfloor T\tau^\star\rfloor| \leq Tv_T^\star$ with probability at least $1-o(1)$. Next we set $u_T = v_T^\star$ and apply Lemma 4.1 again for any $v_T > 0$ to obtain,
\begin{equation*}
    \inf_{\tau \in \mathcal{G}(u_T, v_T)} \mathcal{U}(\tau; \hat{\phi}_1, \hat{\phi}_2) \geq \kappa_{\min}\xi_2^2\left[ v_T - c_{a3} \max \left\{\left(\frac{c_{a3}}{T^{1+b}}\right)^{\frac{1}{2}}, \frac{c_{a3}}{T^{2b}} \right\} \right],
\end{equation*}

then by selecting the sequence $v_T$ such that
\begin{equation*}
    v_T > v_T^\star = \max\left\{\frac{c_{a3}^{g_2}}{T^{u_2}}, \frac{c_{a3}^{h_2}}{T^{v_2}} \right\},
\end{equation*}
where we have
\begin{equation*}
    g_2 = g_1 + \frac{1}{2},\ h_2 = h_1 + 1,\ u_2 = \frac{1}{2} + \frac{u_1}{2},\ v_2 = b + v_1 \geq 2b,\ \text{with}\ u_1 = v_1 + b,\ g_1 = h_1 = 1,
\end{equation*}
then we have $\inf_{\tau \in \mathcal{G}(u_T, v_T)}\mathcal{U}(\tau;\hat{\phi}_1, \hat{\phi}_2) > 0$, with probability at least $1 - 3a - o(1)$ and the rate of convergence of $\tilde{\tau}$ has been sharpened at the second iteration of recursion in comparison to the first iteration. Continuing these recursions by resetting $u_T$ to the bound of the previous recursion and applying Lemma 4.1 then we obtain the $m$-th recursion:
\begin{equation*}
    \Big| \lfloor T\tilde{\tau}\rfloor - \lfloor T\tau^\star\rfloor \Big| \leq T\max\left\{\frac{c_{a3}^{g_m}}{T^{u_m}}, \frac{c_{a3}^{h_m}}{T^{v_m}} \right\},
\end{equation*}
where we define
\begin{equation*}
    g_m = \sum_{k=0}^{m-1}\frac{1}{2^k},\ u_m = \frac{1}{2} + \frac{u_{m-1}}{2} = \frac{b}{2^{m-1}} + \sum_{k=0}^{m-1}\frac{1}{2^k},\ v_m = b + v_{m-1} \geq mb,\ \text{with}\ u_1 = v_1 = b.
\end{equation*}
Therefore, for a large enough $m$, we have $\Big| \lfloor T\tilde{\tau}\rfloor - \lfloor T\tau^\star\rfloor \Big| \leq Tc_{a3}^{g_m} / T^{u_m}$ with probability at least $1-3a-o(1)$. Then continuing this recursion procedure for an infinite number of iterations leads to:
\begin{equation*}
    \Big| \lfloor T\tilde{\tau}\rfloor - \lfloor T\tau^\star\rfloor \Big| \leq T\frac{c_{a3}^2}{T} = c_{a3}^2,
\end{equation*}
with probability at least $1-3a-o(1)$. It finishes the proof of Theorem 4.1.
\end{proof}

\begin{proof}[Proof of Theorem 4.2.]
Note that under the assumption that $\xi_2 \to 0$, we have $\xi_2^2T(\tilde{\tau} - \tau^\star) = \mathcal{O}_p(1)$, which is proved in Theorem 4.1. Then it directly implies that $\tilde{\tau} = \tau^\star + \omega T^{-1}\xi_2^{-2}$ with any $\omega \in [-c_u, c_u]$ ($c_u$ is some constant). 

Define 
\begin{equation}\label{eq:argmax}
    \mathcal{C}(\tau, \phi_1, \phi_2) \overset{\text{def}}{=} -T \, \mathcal{U}(\tau, \phi_1, \phi_2),
\end{equation}
then we obtain that $\tilde{\tau} = \argmax_{\tau \in (0,1)}\mathcal{C}(\tau, \phi_1, \phi_2)$. Recall the argmax theorem introduced in \cite{wellner2013weak}, which requires verification of the following conditions.
\begin{itemize}
    \item[(i)] The quantity $\xi_2^2T(\tilde{\tau} - \tau^\star)$ is uniformly tight;
    \item[(ii)] The process $Z(r)$ satisfies appropriate regularity conditions;
    \item[(iii)] For any $\omega \in [-c_u, c_u] \subset \mathbb{R}$, we have $\mathcal{C}(\tau^\star + \omega T^{-1}\xi_2^{-2}, \hat{\phi}_1, \hat{\phi}_2) \overset{d}{\to} Z(r)$.
\end{itemize}

From the result of Theorem 4.1, we know that $\xi_2^2 T|\tilde{\tau} - \tau^\star| \leq C < +\infty$, where $C$ is some constant and then if we consider $\xi_2^2 T(\tilde{\tau} - \tau^\star)$ as a sequence with respect to the index of $T$, then we have a compact set $K = [-C, C] \subset \mathbb{R}$ and
\begin{equation*}
    \forall \epsilon > 0,\quad \mathbb{P}(\xi_2^2T(\tilde{\tau} - \tau^\star) \in K) \geq 1 - \epsilon, \quad (T \to +\infty)
\end{equation*}
which satisfies the definition of uniform tight. Next, part (ii) follows directly from the regularity condition of Brownian motion since the sample path of $r \mapsto Z(r)$ is negative shifted upper semi-continuous and posses a unique maximum point, which as a random map is tight in some indexing metric space. Hence, we only need to verify part (iii) to finish our proof. For this purpose, we separate our proof into the following steps.

The first step is to illustrate that
\begin{equation*}
    \sup_{\tau \in \mathcal{G}(c_{u}T^{-1}\xi_2^{-2}, 0)}\left| \mathcal{C}(\tau, \hat{\phi}_1, \hat{\phi}_2) - \mathcal{C}(\tau, \phi_1^\star, \phi_2^\star) \right| = o_p(1).
\end{equation*}

This result is proved in Lemma \ref{lemma:c7}. Therefore, we can derive that:
\begin{equation*}
    \sup_{\tau \in \mathcal{G}(c_{u}T^{-1}\xi_2^{-2}, 0)}T\left| \mathcal{U}(\hat{\tau}) - \mathcal{U}(\tau) \right| \leq \sup_{\tau \in \mathcal{G}(c_{u}T^{-1}\xi_2^{-2}, 0)}\left| \mathcal{C}(\tau, \hat{\phi}_1, \hat{\phi}_2) - \mathcal{C}(\tau, \phi_1^\star, \phi_2^\star) \right| = o_p(1).
\end{equation*}

The second step is to consider $\tilde{\tau} = \tau^\star + rT^{-1}\xi_2^{-2}$ with $r \in (0, c_{u}]$. Then by using the results of Lemma \ref{lemma:c10}, we have the limits:
\begin{equation}
    \label{eq:37}
    \sum_{t=\lfloor T\tau^\star \rfloor+1}^{\lfloor T\tilde{\tau} \rfloor} \|\eta^{\star^\prime}Z_t\|_2^2 \overset{P}{\to} r\sigma_2^2.
\end{equation}

Then by using the similar method we applied in the proof of Lemma \ref{lemma:c4}. Define $\zeta_t = \epsilon_t\eta^{\star^\prime}Z_t$ and $\zeta_t^\star = \xi_2^{-1}\zeta_t$. Here we are going to show that $S_T = \sum_{t=\lfloor T\tau^\star \rfloor+1}^T \zeta_t^\star$ is a martingale with respect to the noise terms $\epsilon_{\lfloor T\tau^\star \rfloor+1}, \dots, \epsilon_{T-1}$. We define a filtration $\{\mathcal{F}_T\}_{T\in \mathbb{N}}$ which is a $\sigma$-algebra generated by $(\epsilon_{\lfloor T\tau^\star \rfloor+1}, \dots, \epsilon_{T-1})$. Since $\epsilon_t$'s are independent random variables with mean zero, hence, we have:
\begin{equation*}
    \mathbb{E}(S_{T+1}|\mathcal{F}_T) = \mathbb{E}(S_T + \xi_2^{-1}\epsilon_{T+1}\eta^{\star^\prime}Z_t|\mathcal{F}_T) = S_T + \xi_2^{-1}\mathbb{E}(\epsilon_{T+1}|\mathcal{F}_T)\mathbb{E}(\eta^{\star^\prime}Z_t|\mathcal{F}_T) = S_T,
\end{equation*}
which indicates that $S_T$ is a martingale. From the Condition D(b), we obtain that $\text{Var}(\zeta_t^\star) = \xi_2^{-2}\text{Var}(\zeta_t) \to \sigma_2^{\star^2}$. Next we denote $Y_T = S_T - S_{T-1} = \xi_2^{-1}\epsilon_T\eta^{\star^\prime}Z_T$, $\sigma_T^2 = \mathbb{E}(Y_T^2|\mathcal{F}_{T1})$, $V_T^2 = \sum_{t=\lfloor T\tau^\star \rfloor+1}^T \sigma_t^2$, and $s_T^2 = \mathbb{E}V_T^2$. Then by using the martingale difference central limit theorem \citep{brown1971martingale}, we only need to verify that $V_T^2s_T^{-2} \overset{P}{\to} 1$ as $T \to \infty$. Note that 
\begin{equation}
    \label{eq:38}
    \begin{aligned}
        \sigma_T^2 &= \xi_2^{-2}\mathbb{E}\left( \epsilon_T^2\sum_{j=1}^p(\eta_j^{\star}Z_{T,j})^2 | \mathcal{F}_{T-1} \right) = \xi_2^{-2}\sum_{j=1}^p \mathbb{E}\left( (\epsilon_T\eta_j^{\star}Z_{T,j})^2 | \mathcal{F}_{T-1}\right) \\
        &= \xi_2^{-2}\sum_{j=1}^p \mathbb{E}(\epsilon_T^2 | \mathcal{F}_{T-1})\mathbb{E}((\eta_j^\star Z_{T,j}) | \mathcal{F}_{T-1}) = \xi_2^{-2}\sigma^2 \mathbb{E}\|\eta^{\star^\prime}Z_T\|_2^2,
    \end{aligned}
\end{equation}

and similarly we have:
\begin{equation}
    \label{eq:39}
    V_T^2 = \sum_{t=\lfloor T\tau^\star \rfloor+1}^T \sigma_t^2 = \xi_2^{-2}\sigma^2\sum_{t=\lfloor T\tau^\star \rfloor+1}^{T}\|\eta^{\star^\prime}Z_t\|_2^2,
\end{equation}
\begin{equation}
    \label{eq:40}
    s_n^2 = \mathbb{E}\left(\sum_{t=\lfloor T\tau^\star\rfloor+1}^T \mathbb{E}(Y_t^2|\mathcal{F}_{t-1})\right) = (T - \lfloor T\tau^\star\rfloor)\xi_2^{-2}\sigma^2\eta^\star\Sigma_Z^2\eta^{\star^\prime},
\end{equation}
then combining \eqref{eq:39} and \eqref{eq:40} implies that
\begin{equation}
    \frac{V_T^2}{s_T^2} = \frac{\sum_{t=\lfloor T\tau^\star \rfloor+1}^T \|\eta^{\star^\prime}Z_t\|_2^2}{(T-\lfloor T\tau^\star\rfloor)\eta^\star\Sigma_Z^2\eta^{\star^\prime}} \overset{P}{\to} 1, 
\end{equation}
as $T \to +\infty$, which holds based on ergodicity. Next, we focus on validating the Lindeberg condition. 
\begin{equation}
    \label{eq:41}
    \begin{aligned}
        &s_T^{-2}\sum_{t=\lfloor T\tau^\star\rfloor + 1}^n \mathbb{E}(Y_t^2\mathbb{I}(|Y_t| \geq \epsilon s_T)|\mathcal{F}_{T-1}) = s_T^{-4}\sum_{t=\lfloor T\tau^\star \rfloor+1}^T \mathbb{E}(s_n^2 Y_t^2 \mathbb{I}(|Y_t| \geq \epsilon s_T)|\mathcal{F}_{T-1}) \\
        \leq &\quad \epsilon^{-2}s_T^{-4}\sum_{t=\lfloor T\tau^\star \rfloor+1}^T \mathbb{E}(Y_t^4\mathbb{I}(|Y_t| \geq \epsilon s_T)|\mathcal{F}_{T-1}) \overset{(i)}{\leq} \epsilon^{-2}s_T^{-4}\sum_{t=\lfloor T\tau^\star \rfloor+1}^T \mathbb{E}(Y_t^4|\mathcal{F}_{T-1}) \\
        \leq &\quad \epsilon^{-2}s_T^{-4}\sum_{t=\lfloor T\tau^\star \rfloor+1}^T \xi_2^{-4}\mathbb{E}\left(\epsilon_t^4 \left(\sum_{j=1}^p \eta_j^\star Z_{t,j}\right)^4 | \mathcal{F}_{T-1}\right) \\
        \overset{(ii)}{\leq} &\quad \epsilon^{-2}s_T^{-4}\sum_{t=\lfloor T\tau^\star\rfloor+1}^T \xi_2^{-4}\mathbb{E}(\epsilon_t^4)\text{Tr}\left(\eta^{\star^\prime}Z_tZ_t^\prime \eta^\star\right)\text{Tr}\left(\eta^{\star^\prime}Z_tZ_t^\prime\eta^\star\right) \\
        {\leq} &\quad \frac{\sum_{t=\lfloor T\tau^\star\rfloor+1}^{T}\xi_2^{-4}\mathbb{E}(\epsilon_t^4)\left(\text{Tr}\left(\eta^{\star^\prime}Z_tZ_t^\prime\eta^\star\right)\right)^2}{\epsilon^2(T-\lfloor T\tau^\star\rfloor)\xi_2^{-4}\sigma^4\left(\text{Tr}\left(\eta^{\star^\prime}\Sigma_Z\eta^\star\right)\right)^2} \overset{P}{\to} 0.
    \end{aligned}
\end{equation}


Here the inequality (i) holds because of the definition of $Y_t$; (ii) holds because
\begin{equation*}
    \begin{aligned}
        &\mathbb{E}\left(\left(\sum_{j=1}^p \eta_j^\star Z_{t,j}\right)^4 | \mathcal{F}_{T-1}\right) \\
        &= \mathbb{E}\left(\left(\sum_{j=1}^p(\epsilon_t Z_{t,j}\eta_j^\star)^2\right)^2|\mathcal{F}_{T-1}\right) \\
        &= \sum_{j=1}^p\mathbb{E}(\epsilon_t^4)\mathbb{E}\left((Z_{t,j}\eta_j^\star)^4|\mathcal{F}_{T-1}\right) + (\mathbb{E}(\epsilon_t^2))^2\sum_{j,j'=1,2,\dots,p,j\neq j'}\mathbb{E}\left((Z_{t,j}\eta_j^\star)^2|\mathcal{F}_{T-1}\right)\mathbb{E}\left((Z_{t,j'}\eta_{j'}^\star)^2|\mathcal{F}_{T-1}\right) \\
        &\leq \mathbb{E}(\epsilon_t^4)\text{Tr}\left(\eta^{\star^\prime}Z_tZ_t^\prime \eta^\star\right)\text{Tr}\left(\eta^{\star^\prime}Z_tZ_t^\prime\eta^\star\right).
    \end{aligned}
\end{equation*}

Then for any $k \geq 1$, the moment function $\mathbb{E}(|\epsilon_t|^k) \leq (K_\epsilon k^{1/\gamma_2})^k$, which is derived by the property of sub-Weibull distribution, hence, $\mathbb{E}(\epsilon_t^4)$ is upper bounded. Therefore, by using the result of \eqref{eq:42}, we have validated the Lindeberg condition in accordance to Lemma 2 in \cite{brown1971martingale}. Then applying the martingale difference central limit theorem \cite{brown1971martingale} implies that:
\begin{equation}
    \label{eq:43}
    \sum_{t=\lfloor T\tau^\star\rfloor + 1}^{\lfloor T\tilde{\tau} \rfloor} \epsilon_t \eta^{\star^\prime}Z_t = \xi_2\sum_{t=\lfloor T\tau^\star\rfloor + 1}^{\lfloor T\tilde{\tau}\rfloor}\xi_2^{-1}\zeta_t = \xi_2\sum_{t=\lfloor T\tau^\star\rfloor+1}^{\lfloor T\tilde{\tau}\rfloor} \zeta_t^\star \overset{d}{\to} \sigma_2^\star W_2(r),
\end{equation}
as $\xi_2 \to 0$, where $W_2(r)$ is a Brownian motion defined on $[0, \infty)$. Thus, we obtain that
\begin{equation}
    \label{eq:44}
    \begin{aligned}
        \mathcal{C}(\tilde{\tau}; \phi_1^\star, \phi_2^\star) 
        &= -\sum_{t=\lfloor T\tau^\star\rfloor+1}^{\lfloor T\tilde{\tau}\rfloor}\left\{ (X_t - \phi_1^{\star^\prime} Z_t)^2 - (X_t - \phi_2^{\star^\prime} Z_t)^2 \right\} \\
        &= -\sum_{t=\lfloor T\tau^\star\rfloor+1}^{\lfloor T\tilde{\tau} \rfloor} \left(\eta^\star Z_tZ_t^\prime \eta^{\star^\prime} - 2\epsilon_t\eta^{\star^\prime}Z_t\right) \overset{d}{\to} 2\sigma_2^\star W_2(r) - \sigma_2^2 r. 
    \end{aligned}
\end{equation}
Symmetrically, for $r \in [-c_1, 0)$, it can be shown that $\mathcal{C}(\tilde{\tau}; \phi_1^\star, \phi_2^\star) \overset{d}{\to} 2\sigma_1^\star W_1(r) + \sigma_1^2r$, where $W_1(r)$ is another Brownian motion on $[0, \infty)$ independent from $W_2(r)$.


Symmetrically, we can show that for $r \in [-c_u, 0)$, it implies $\mathcal{C}(\tilde{\tau}, \phi_1^\star, \phi_2^\star) \overset{d}{\to} 2\sigma_1^\star W_1(r) + \sigma_1^2r$, where $W_1(r)$ is another Brownian motion defined on 
$[0,\infty)$ independent from $W_2(r)$. 
The last step is to verify the asymptotic distribution equals to the process $Z(r)$. We define the process
\begin{equation*}
    G(r) = 
    \begin{cases}
        2\sigma_1^\star W_1(r) + \sigma_1^2r,\quad &\text{if } r < 0, \\
        0,\quad &\text{if } r = 0, \\
        2\sigma_2^\star W_2(r) - \sigma_2^2r,\quad &\text{if } r > 0. 
    \end{cases}
\end{equation*}

Let $r' = (\sigma_1^4/\sigma_1^{\star^2})r$, then it implies that $W((\sigma_1^4/\sigma_1^{\star^2})r) \overset{d}{=} \sigma_1^2/\sigma_1^\star W(r)$. Consequently, we have
\begin{equation*}
    \begin{aligned}
        \argmax_r\left\{ 2\sigma_2^\star W_2(r) - \sigma_2^2 r \right\} &\overset{d}{=} \frac{\sigma_1^{\star^2}}{\sigma_1^4}\argmax_{r'}\left\{ \frac{2\sigma_1^\star\sigma_2^\star}{\sigma_1^2}W_2(r') - \frac{\sigma_1^{\star^2}\sigma_2^2}{\sigma_1^4}r' \right\} \\
        &\overset{d}{=} \frac{\sigma_1^{\star^2}}{\sigma_1^4}\argmax_{r'}\left\{ 2\sigma_1^\star\sigma_2^\star W_2(r') - \frac{\sigma_1^{\star^2}\sigma_2^2}{\sigma_1^2}r' \right\} \\
        &\overset{d}{=} \frac{\sigma_1^{\star^2}}{\sigma_1^4}\argmax_{r'}\left\{ \frac{2\sigma_2^\star}{\sigma_1^\star}W_2(r') - \frac{\sigma_2^2}{\sigma_1^2}r' \right\},
    \end{aligned}
\end{equation*}
and similarly we have
\begin{equation*}
    \argmax_r\left\{ 2\sigma_1^\star W_1(r) + \sigma_1^2r \right\} \overset{d}{=} \frac{\sigma_1^{\star^2}}{\sigma_1^4}\argmax_{r'}\left\{2W_1(r') + r' \right\}.
\end{equation*}


Combining these results together completes the proof of part (iii) and the final statement of Theorem 4.2. 
\end{proof}

\section{C \quad Auxiliary Lemmas and Their Proofs} \label{App_C}
\renewcommand{\theequation}{C.\arabic{equation}}
\setcounter{equation}{0}
\renewcommand{\thelemma}{C.\arabic{lemma}}
\renewcommand{\theorem}{C.\arabic{theorem}}
\renewcommand{\thedefinition}{C.\arabic{definition}}

\begin{lemma}\label{lemma:c1}
Under Condition A, we have the following results. (a) The random variable $\zeta_t = \epsilon_t{\eta^\star}^\prime Z_t$ is sub-Weibull with parameter $\gamma_2/2$ bounded by $K_{\zeta} = 2^{2/{\gamma_2}}K_\epsilon K_X \xi_2$ for each time point $t=1,2,\dots, T$. (b) the $k$-th moment of $\zeta_t$ is bounded: $\mathbb{E}(|\zeta_t|^k) \leq K_{\zeta}^k k^{\frac{2k}{\gamma_2}}$ for any $k \geq 1$. 
\end{lemma}
\begin{proof}[Proof of Lemma \ref{lemma:c1}]
For part (a), we apply Condition A that $\{\epsilon_t\}$ and $\{X_t\}$ are both sub-Weibull with parameter $\gamma_2$. Then the $p$-lagged vector $Z_t$ defined in (7) is also a sub-Weibull($\gamma_2$) distributed. Then $\zeta_t = \epsilon_t {\eta^\star}^\prime Z_t$ is real valued sub-Weibull distributed, and by applying general H\"{o}lder inequality, it can be proved by the following:
\begin{equation*}
    \begin{aligned}
        \|\zeta_t\|_k 
        &\leq \|\epsilon_t{\eta^\star}^\prime Z_t\|_k \leq \|\epsilon_t\|_{2k}\|{\eta^\star}^\prime Z_t\|_{2k} \leq \|\epsilon_t\|_{2k}(2k)^{\frac{1}{\gamma_2}}\|{\eta^\star}^\prime Z_t\|_{\psi_{\gamma_2}} \\
        &\leq \|\epsilon_t\|_{2k}(2k)^{\frac{1}{\gamma_2}}\|\eta^\star\|_2\|Z_t\|_{\psi_{\gamma_2}} \leq (2k)^{\frac{2}{\gamma_2}}K_\epsilon K_X\xi_2.
    \end{aligned}
\end{equation*}
Hence, we proved that $\zeta_t \sim$ sub-Weibull($\gamma_2/2$). For part (b), the $k$-th moment ($k\geq 1$) generating function is given by $\mathbb{E}(|\zeta_t|^k) \leq (K_\zeta k^{\frac{2}{\gamma_2}})^k = K_\zeta^k k^{\frac{2k}{\gamma_2}}$. 
\end{proof}

\begin{lemma}\label{lemma:c2}
Suppose Condition A holds and $u_T$, $v_T$ are any non-negative sequences with $(1/T) \leq v_T \leq u_T$. Let $K$ be determined by $K := 2^{2/{\gamma_2}}(K_X + K_\epsilon)^2$. Then for any $c_u > 0$ and $T \to +\infty$, we have:
\begin{equation*}
    \begin{aligned}
        &(i)\quad \sup_{\substack{\tau \in \mathcal{G}(u_T, v_T),\\ \tau \geq \tau^\star}}\frac{1}{T}\left\| \sum_{t=\lfloor T\tau^\star\rfloor + 1}^{\lfloor T\tau \rfloor}\epsilon_t Z_t \right\|_\infty \leq c_1 K\left( \frac{u_T\log(p \vee T)}{T} \right)^{\frac{1}{2}}, \\
        &(ii)\quad \sup_{\substack{\tau \in \mathcal{G}(u_T, v_T),\\ \tau \geq \tau^\star}}\frac{1}{T}\left| \sum_{t=\lfloor T\tau^\star\rfloor + 1}^{\lfloor T\tau \rfloor}\epsilon_t (\hat{\eta} - \eta^\star)^\prime Z_t \right| \leq c_1 K\left( \frac{u_T\log(p \vee T)}{T} \right)^{\frac{1}{2}}\|\hat{\eta} - \eta^\star\|_1,
    \end{aligned}
\end{equation*}
with probability at least $1 - 2\exp(-c_u\log (p\vee T)) = 1-o(1)$, where $c_1$ only depends on $c_{u1}$ and parameters $\gamma_1, \gamma_2$ and $c$ in the Condition A.
\end{lemma}

\begin{proof}[Proof of Lemma \ref{lemma:c2}]
To prove (i), we use a similar strategy as that of Proposition 7 in \cite{wong2020lasso}. Under Condition A, and the conclusions in Lemma \ref{lemma:c1}, $\{\epsilon_t\}$ and $\{Z_t\}$ are sub-Weibull($\gamma_2$) distributed. Moreover, we know that $\mathbb{E}(\epsilon_t) = 0$ and $\mathbb{E}(Z_{t,j}) = 0$ for each $j=1,2,\dots, p$. Note that 
\begin{equation*}
    \left\| \sum_{t=\lfloor T\tau^\star\rfloor + 1}^{\lfloor T\tau \rfloor}\epsilon_t Z_t \right\|_\infty = \max_{1\leq j \leq p}\left| \sum_{t=\lfloor T\tau^\star\rfloor + 1}^{\lfloor T\tau \rfloor}\epsilon_t Z_{t,j} \right|.
\end{equation*}
Therefore, we consider the following fact, for any $t \in \mathbb{R}^+$,
\begin{equation*}
    \begin{aligned}
        &\mathbb{P}\left( \frac{1}{Tu_T}\left| \sum_{t=\lfloor T\tau^\star \rfloor + 1}^{\lfloor T\tau \rfloor}\epsilon_t Z_{t,j} \right| > 3t \right) \\
        &\leq \mathbb{P}\left( \frac{1}{2Tu_T}\left| \sum_{t=\lfloor T\tau^\star \rfloor + 1}^{\lfloor T\tau \rfloor}\left\{(\epsilon_t + Z_{t,j})^2 - \mathbb{E}(\epsilon_t + Z_{t,j})^2\right\} \right| > t \right) \\
        &+ \mathbb{P}\left( \frac{1}{2Tu_T}\left| \sum_{t=\lfloor T\tau^\star \rfloor + 1}^{\lfloor T\tau \rfloor}\left\{ Z_{t,j}^2 - \mathbb{E}(Z_{t,j}^2) \right\} \right| > t \right) + \mathbb{P}\left( \frac{1}{2Tu_T}\left| \sum_{t=\lfloor T\tau^\star \rfloor + 1}^{\lfloor T\tau \rfloor}\left\{ \epsilon_t^2 - \mathbb{E}(\epsilon_t^2) \right\} \right| > t \right).
    \end{aligned}
\end{equation*}


Recall that $\|\epsilon_t\|_{\psi_{\gamma_2}} \leq K_\epsilon$ and $\|Z_{t,j}\|_{\psi_{\gamma_2}} \leq K_X$, as well as the triangle inequality for Orlicz norm: $\|\epsilon_t + Z_{t,j}\|_{\psi_{\gamma_2}} \leq \|\epsilon_t\|_{\psi_{\gamma_2}} + \|Z_{t,j}\|_{\psi_{\gamma_2}} \leq K_\epsilon + K_X$. According to the properties of the sub-Weibull distribution, we know that $\epsilon_t^2$, $Z_{t,j}^2$ and $\epsilon_t Z_{t,j}$ are all sub-Weibull($\gamma_2/2$) and upper bounded by $K := 2^{2/{\gamma_2}}(K_\epsilon + K_X)^2$.

Next, by applying Lemma 13 in \cite{wong2020lasso} to each term above and using union bound inequality over all possible $p$ coordinates in $Z_t$, we obtain that:
\begin{equation*}
    \mathbb{P}\left(\frac{1}{Tu_T}\left\| \sum_{t=\lfloor T\tau^\star \rfloor +1}^{\lfloor T\tau \rfloor} \epsilon_t Z_t \right\|_\infty > 3t\right) \leq 3pTu_T\exp\left(-\frac{(2tTu_T)^\gamma}{K^\gamma C_1}\right) + 3p\exp\left(-\frac{4t^2 Tu_T}{K^2 C_2}\right),
\end{equation*}


where $\gamma = (1/{\gamma_1} + 2/{\gamma_2})^{-1} < 1$, and $\gamma_1$ is for $\beta$-mixing assumption in Condition A(b). The constants $C_1, C_2>0$ only depend on $\gamma_1, \gamma_2$. 

Now, we take $t = K\max\left\{\sqrt{\frac{C_2\log(p\vee T)}{Tu_T}}, \frac{(C_1\log(p\vee T))^{1/\gamma}}{2Tu_T}\right\}$, then for sufficiently large $T$, we obtain that:
\begin{equation*}
    \mathbb{P}\left( \frac{1}{Tu_T}\left\| \sum_{t=\lfloor T\tau^\star \rfloor + 1}^{\lfloor T\tau \rfloor} \epsilon_t Z_t \right\|_\infty > 3K\sqrt{\frac{C_2\log(p\vee T)}{Tu_T}} \right) \leq 2\exp(-c_u^\prime\log(p\vee T)),
\end{equation*}
where $c_u^\prime > 0$ is some large enough universal constant. Hence, we have that
\begin{equation*}
    \mathbb{P}\left(\frac{1}{T}\left\| \sum_{t=\lfloor T\tau^\star\rfloor + 1}^{\lfloor T\tau \rfloor}\epsilon_t Z_t \right\|_\infty > c_1K\left(\frac{u_T\log(p\vee T)}{T}\right)^{\frac{1}{2}}\right) \leq 2\exp(-c_u^\prime\log(p\vee T)),
\end{equation*}
then for (ii), it can be implied directly by using H\"{o}lder inequality and result of (i).
\end{proof}

\begin{lemma}\label{lemma:c3}
Suppose Condition A holds and let $u_T, v_T$ be any non-negative sequence such that $(1/T) \leq v_T \leq u_T$. Then there exist some large enough constants $c_i>0$ such that
\begin{equation*}
    \begin{aligned}
        &(i)\quad \inf_{\substack{\tau \in \mathcal{G}(u_T, v_T) \\ \tau \geq \tau^\star}}\frac{1}{T}\sum_{t=\lfloor T\tau^\star \rfloor + 1}^{\lfloor T\tau \rfloor}\|\eta^{\star^\prime} Z_t\|_2^2 \geq v_T \kappa_{\min}\xi_2^2 - c_{a1}K_X^2\xi_2^2\sqrt{\frac{u_T}{T}}, \\
        &(ii)\quad \sup_{\substack{\tau \in \mathcal{G}(u_T, v_T) \\ \tau \geq \tau^\star}}\frac{1}{T}\sum_{t=\lfloor T\tau^\star \rfloor+1}^{\lfloor T\tau\rfloor}\|\eta^{\star^\prime}Z_t\|_2^2\leq u_T\kappa_{\max}\xi_2^2 + c_{a1}K_X^2\xi_2^2\sqrt{\frac{u_T}{T}},
    \end{aligned}
\end{equation*}
where we set $c_{a1} = c_{a2}2^{2/{\gamma_2}}$ and $c_{a2} \geq \sqrt{-C_2\log a}$ with any $0 < a < 1$.
\end{lemma}

\begin{proof}[Proof of Lemma \ref{lemma:c3}]
According to Lemma \ref{lemma:c1}, we know that both $\eta^{\star^\prime}Z_t$ and $Z_t$ are sub-Weibull($\gamma_2$) distributed. Then we use lines similar to those of Lemma 3 in \cite{safikhani2018joint} and Lemma 1 in \cite{bai2021multiple}. We rewrite the left hand side of both (i) and (ii) as:
\begin{equation*}
    \frac{1}{T}\sum_{t=\lfloor T\tau^\star \rfloor + 1}^{\lfloor T\tau \rfloor} \|\eta^{\star^\prime}Z_t\|_2^2 = \frac{1}{T}\sum_{t=\lfloor T\tau^\star \rfloor + 1}^{\lfloor T\tau \rfloor}\sum_{j=1}^p (\eta_j^\star Z_{t,j})^2,
\end{equation*}
where $\eta_j^\star$ and $Z_{t,j}$ are the $j$-th coordinate of coefficient vector $\eta^\star$ and random vector $Z_t$. Note that by using the property of sub-Weibull distribution (see Definition 2.3), we have:
\begin{equation*}
    \sum_{j=1}^p s_{t,j} = \sum_{j=1}^p \left\{ (\eta_j^\star Z_{t,j})^2 - \mathbb{E}(\eta_j^\star Z_{t,j})^2 \right\} \sim \text{sub-Weibull}(\gamma_2/2),
\end{equation*}
and $\mathbb{E}(s_{t,j}) = 0$. Then we let $S_\tau(\eta) = \sum_{t=\lfloor T\tau^\star \rfloor+1}^{\lfloor T\tau \rfloor}\left( \sum_{j=1}^p\left\{ (\eta_j^\star Z_{t,j})^2 - \mathbb{E}(\eta_j^\star Z_{t,j})^2\right\} \right)$, and by the definition of the sequence $u_T$, we know that $\lfloor T\tau \rfloor - \lfloor T\tau^\star \rfloor \leq Tu_T$. We apply the Lemma 13 in \cite{wong2020lasso} to $S_\tau(\eta)$, for any $t>0$:
\begin{equation}\label{eq:27}
    \mathbb{P}\left( \sup_{\substack{\tau \in \mathcal{G}(u_T, v_T) \\ \tau \geq \tau^\star}}\left| \frac{S_\tau(\eta)}{Tu_T} \right| > t \right) \leq Tu_T\exp\left(-\frac{(tTu_T)^\gamma}{K^\gamma C_1}\right) + \exp\left(-\frac{t^2Tu_T}{K^2C_2}\right),
\end{equation}
where $\gamma = (1/{\gamma_1} + 2/{\gamma_2})^{-1} < 1$. $C_1, C_2$ are some large enough constants that depend only on $\gamma_1$, $\gamma_2$, and $K_X$. 

Hence, we take $t = \frac{c_{a2}2^{2/{\gamma_2}}\xi_2^2K_X^2}{\sqrt{Tu_T}}$, we put it back to \eqref{eq:27} derive that:
\begin{equation*}
    \mathbb{P}\left(\sup_{\substack{\tau \in \mathcal{G}(u_T, v_T) \\ \tau\geq \tau^\star}}\left| \frac{S_\tau(\eta)}{Tu_T} \right| > \frac{c_{a2}2^{2/{\gamma_2}}\xi_2^2K_X^2}{\sqrt{Tu_T}}\right) \leq Tu_T\exp\left(-\frac{c_{a2}^\gamma(Tu_T)^{\gamma/2}}{C_1}\right) + \exp\left(-\frac{c_{a2}^2}{C_2}\right),
\end{equation*}
hence, as $T\to +\infty$, and we set $c_{a2} \geq \sqrt{-C_2\log a}$ for $0<a<1$, then we obtain that:
\begin{equation*}
    \mathbb{P}\left(\sup_{\substack{\tau \in \mathcal{G}(u_T, v_T) \\ \tau\geq \tau^\star}}\left| \frac{S_\tau(\eta)}{Tu_T} \right| > \frac{c_{a2}2^{2/{\gamma_2}}\xi_2^2K_X^2}{\sqrt{Tu_T}}\right) \leq \exp\left(-\frac{c_{a2}^2}{C_2}\right) \leq a.
\end{equation*}

Therefore, together with the Condition A(c), we have:
\begin{equation*}
    \begin{aligned}
        &\sup_{\substack{\tau \in \mathcal{G}(u_T, v_T) \\ \tau \geq \tau^\star}}\left| \frac{S_\tau(\eta)}{T} \right| \leq c_{a1}\xi_2^2K_X^2\sqrt{\frac{u_T}{T}}, \\
        &\sup_{\substack{\tau \in \mathcal{G}(u_T, v_T) \\ \tau \geq \tau^\star}}\frac{1}{T}\sum_{t=\lfloor T\tau^\star \rfloor +1}^{\lfloor T\tau \rfloor}\|\eta^{\star^\prime}Z_t\|_2^2 \leq \frac{Tu_T}{T}\kappa_{\max}\xi^2_2 = u_T\kappa_{\max}\xi_2^2.
    \end{aligned}
\end{equation*}
Then by substituting the result back to $s_{t,j}$, we obtain:
\begin{equation*}
    \sup_{\substack{\tau \in \mathcal{G}(u_T, v_T) \\ \tau \geq \tau^\star}}\frac{1}{T}\sum_{t=\lfloor T\tau^\star\rfloor + 1}^{\lfloor T\tau \rfloor}\|\eta^{\star^\prime}Z_t\|_2^2 \leq u_T\kappa_{\max} \xi_2^2 + c_{a1}K_X^2\xi_2^2\sqrt{\frac{u_T}{T}},
\end{equation*}
which is the inequality (ii), and we use the similar method and Condition A(c) to show the inequality (i) as well. 
\end{proof}

Before we state the next lemma, the following result in \cite{birnbaum1961some} is required and essential. Here we only provide the statement of the theorem and omit the detailed proof.
\begin{lemma}
    \label{thm:chebyshev}
    Let $X_1, X_2, \dots, X_n$ be random variables such that 
    \begin{equation*}
        \mathbb{E}(|X_k| \mid X_1, \dots, X_{k-1}) \geq \psi_k|X_{k-1}|,\quad a.e.,
    \end{equation*}
    where $\psi_k \geq 0$ for $k=2,3,\dots, n$. Then let $a_k>0$ and $b_k = \max\left\{ a_k, a_{k+1}\psi_{k+1}, \dots, a_n\prod_{i=k+1}^n \psi_i \right\}$, and let $X_0=0$, $b_{n+1}=0$. If $r\geq 1$ is satisfied that $\mathbb{E}(|X_k|^r) < +\infty$, for $k=1,2,\dots, n$, then we have:
    \begin{equation*}
        \mathbb{P}(\max_{1\leq k\leq n}a_k|X_k| \geq 1) \leq \sum_{k=1}^n b_k^r\Big(\mathbb{E}(|X_k|^r) - \psi_k^r\mathbb{E}(|X_{k-1}|^r)\Big).
    \end{equation*}
\end{lemma}

\begin{proof}[Proof of Lemma \ref{thm:chebyshev}]
See the proof of Theorem 2.1 in \cite{birnbaum1961some}.
\end{proof}


\vspace{0.5cm}

\begin{lemma}\label{lemma:c4}
Suppose Condition A is satisfied and let $u_T, v_T$ be any nonnegative sequence such that $(1/T) \leq v_T \leq u_T$. Then for any small enough $0 < a < 1$, there exist constants $c_{a1} = c_{a2}2^{2/{\gamma_2}}$ with $c_{a2} \geq \sqrt{1/a}$, then we have 
\begin{equation*}
    \sup_{\substack{\tau \in \mathcal{G}(u_T, v_T)\\ \tau \geq \tau^\star}}\frac{1}{T}\left| \sum_{t=\lfloor T\tau^\star \rfloor+1}^{\lfloor T\tau\rfloor}\epsilon_t\eta^{\star^\prime}Z_t \right| \leq c_{a1}2^{\frac{2}{\gamma_2}}K_\epsilon K_X\xi_2\sqrt{\frac{u_T}{T}},
\end{equation*}
with probability at least $1-a$. 
\end{lemma}


\begin{proof}[Proof of Lemma \ref{lemma:c4}]
According to the properties and definition of sub-Weibull distribution (see Definition 2.2) and applying the result of Lemma \ref{lemma:c1}, we know that $\zeta_t \overset{\text{def}}{=} \epsilon_t\eta^{\star^\prime}Z_t$ is sub-Weibull($\gamma_2/2$) distributed and bounded by $K_{\zeta} := 2^{2/\gamma_2}K_XK_\epsilon \xi_2$. Then we define $n = \lfloor T\tau \rfloor$,
\begin{equation*}
    S_n \overset{\text{def}}{=} \sum_{t=\lfloor T\tau^\star\rfloor+1}^{n}\zeta_t,
\end{equation*}
and we first show that $S_n$ is a martingale. Since the error terms $\epsilon_t$'s are independent zero mean random variables, then for filtration $\mathcal{F}_n$ generated by $(S_1, \dots, S_n)$, we have:
\begin{equation*}
    \mathbb{E}(S_{n+1}|\mathcal{F}_n) = \mathbb{E}(S_n + \epsilon_{n+1}\eta^{\star^\prime}Z_{n+1}|\mathcal{F}_n) = S_n + \mathbb{E}(\epsilon_{n+1}|\mathcal{F}_n)\mathbb{E}(\eta^{\star^\prime}Z_{n+1}|\mathcal{F}_n) = S_n,
\end{equation*}
hence, we show that $S_n$ is a martingale. Moreover, we notice that for any $i < j$, we have:
\begin{equation*}
    \mathbb{E}(\zeta_i\zeta_j) = \mathbb{E}(\mathbb{E}(\zeta_i\zeta_j|\mathcal{F}_i)) = \mathbb{E}(\epsilon_i\eta^{\star^\prime}Z_i\mathbb{E}(\epsilon_j|\mathcal{F}_i)\mathbb{E}(\eta^{\star^\prime}Z_j|\mathcal{F}_i)) = 0.
\end{equation*}
Then we apply Theorem \ref{thm:chebyshev} with $\psi_k \equiv 1$ and $r=2$ for all $k$ to $S_n$ together with the result of Lemma \ref{lemma:c1}. For any $d>0$ and let $a_k = b_k = 1/d$, we have:
\begin{equation*}
    \mathbb{P}\left(\sup_{\substack{\tau \in \mathcal{G}(u_T, v_T) \\ \tau\geq \tau^\star}}\left|\sum_{t=\lfloor T\tau^\star\rfloor+1}^{\lfloor T\tau\rfloor}\zeta_t \right| > d\right) \leq \frac{1}{d^2}\mathbb{E}\left( \left|\sum_{t=\lfloor T\tau^\star\rfloor+1}^{\lfloor T\tau\rfloor}\zeta_t \right|^2\right) \leq \frac{Tu_T}{d^2}\max_t\mathbb{E}(|\zeta_t|^2) \leq \frac{2^{\frac{4}{\gamma_2}}(Tu_T)K_\zeta^2}{d^2}.
\end{equation*}
Therefore, by choosing $d = c_{a2}2^{\frac{2}{\gamma_2}}K_\zeta\sqrt{Tu_T} = c_{a1}2^{\frac{2}{\gamma_2}}K_\epsilon K_X\xi_2\sqrt{Tu_T}$ with $c_{a2} \geq \sqrt{1/a}$, we obtain that,
\begin{equation*}
    \mathbb{P}\left(\sup_{\substack{\tau \in \mathcal{G}(u_T, v_T) \\ \tau\geq \tau^\star}}\frac{1}{T}\left|\sum_{t=\lfloor T\tau^\star\rfloor+1}^{\lfloor T\tau\rfloor}\zeta_t \right| > c_{a1}2^{\frac{2}{\gamma_2}}K_\epsilon K_X\xi_2\sqrt{\frac{u_T}{T}}\right) \leq a, 
\end{equation*}
which leads to the final result.
\end{proof}

Before we demonstrate the next lemma, the following auxiliary function $\Phi(\cdot, \cdot)$ is required. Suppose $\alpha$ and $\beta$ are any $p$-dimensional vectors, then we define:
\begin{equation}\label{eq:auxiliary-function}
    \Phi(\alpha, \beta) = \frac{1}{T}\sum_{t=\lfloor T\tau^\star\rfloor+1}^{\lfloor T\tau \rfloor}\alpha^\prime Z_tZ_t^\prime \beta.
\end{equation}
The following lemma presents the properties of $\Phi(\cdot, \cdot)$. 
\begin{lemma}
\label{lemma:c8}
Let $\Phi(\cdot, \cdot)$ be the auxiliary function defined in \eqref{eq:auxiliary-function} and assume that Conditions A, B, and C are satisfied. Let $u_T$ and $v_T$ be any non-negative sequences satisfying $(1/T) \leq v_T \leq u_T$. Then for any $0 < a < 1$, we choose $c_{a1} = 2^{2/\gamma_2}c_{a2}$, with $c_{a2} \geq (-C_2\log a)^{1/2}$, where the constant $C_2>0$ depends only on the sub-Weibull parameters $\gamma_1$ and $\gamma_2$. Suppose the sample size $T$ is sufficiently large, then we have:
\begin{equation*}
    \begin{aligned}
        &(i)\quad \inf_{\tau \in \mathcal{G}(u_T, v_T)}\Big| \Phi(\eta^\star, \eta^\star) \Big| \geq v_T\kappa_{\min}\xi_2^2 - c_{a1}K_X^2\xi_2^2\sqrt{\frac{u_T}{T}}; \\
        &(ii)\quad \sup_{\tau \in \mathcal{G}(u_T, v_T)}\Big| \Phi(\hat{\eta} - \eta^\star, \hat{\eta} - \eta^\star)\Big| \leq c_uT^\gamma u_T\kappa_{\max}(1+\nu^2)\frac{\sigma^4}{\kappa_{\min}^2}\left( \frac{p^2\log (p\vee T)}{Tl_T}\right),
    \end{aligned}
\end{equation*}
with probability at least $1-a$. Furthermore, when we have $u_T \geq c^2_{a1}K_X^4/(T\kappa^2_{\max})$, we obtain that:
\begin{equation*}
    \begin{aligned}
        &(iii)\quad \sup_{\tau \in \mathcal{G}(u_T, v_T)}\Big| \Phi(\eta^\star, \eta^\star) \Big| \leq 2u_T\kappa_{\max}\xi_2^2; \\
        &(iv)\quad \sup_{\tau \in \mathcal{G}(u_T, v_T)}\Big| \Phi(\hat{\eta} - \eta^\star, \eta^\star) \Big| \leq c'_u u_T\xi_2\sqrt{\kappa_{\max}(1+\nu^2)}\frac{\sigma^2}{\kappa_{\min}}\left( \frac{p^2\log(p\vee T)}{T^{1-\gamma}l_T} \right)^{\frac{1}{2}},
    \end{aligned}
\end{equation*}
with probability at least $1-a-o(1)$. 
\end{lemma}

\begin{proof}[Proof of Lemma \ref{lemma:c8}]
Parts (i) and (iii) of Lemma \ref{lemma:c8} are the straightforward application of Lemma \ref{lemma:c3}. For part (iii), when we choose $u_T \geq c_{a1}^2 K_X^4 / (T\kappa_{\max}^2)$, then it leads to
\begin{equation*}
    u_T\kappa_{\max}\xi_2^2 + c_{a1}K_X^2\xi_2^2\sqrt{\frac{u_T}{T}} = 2u_T\kappa_{\max}\xi_2^2,
\end{equation*}
which satisfies the results. 

To show part (ii), we first notice the fact that
\begin{equation*}
    \|\hat{\eta} - \eta^\star\|_2^2 
    \leq 2\Big(\|\hat{\phi}_1 - \phi_1^\star\|_2^2 + \|\hat{\phi}_2 - \phi_2^\star\|_2^2\Big) \leq 2c_u^2(1+\nu^2)\frac{\sigma^4}{\kappa^2_{\min}}\frac{p^2\log(p\vee T)}{Tl_T},
\end{equation*}
where the second inequality follows the condition C. 

Next, by using the definition of $\Phi(\cdot, \cdot)$ again, we obtain that:
\begin{equation*}
    \begin{aligned}
        \sup_{\tau \in \mathcal{G}(u_T, v_T)}\Big| \Phi(\hat{\eta} - \eta^\star, \hat{\eta} - \eta^\star)\Big| 
        &\leq T^\gamma u_T(\kappa_{\max} \vee \alpha)\|\hat{\eta} - \eta^\star\|_2^2 \\
        &\leq 2c_u^2T^\gamma u_T(\kappa_{\max} \vee \alpha)(1+\nu^2)\frac{\sigma^4}{\kappa^2_{\min}}\frac{p^2\log(p\vee T)}{Tl_T} \\
        &\leq c_u^\prime T^\gamma u_T\kappa_{\max}(1+\nu^2)\frac{\sigma^4}{\kappa^2_{\min}}\frac{p^2\log(p\vee T)}{Tl_T},
    \end{aligned}
\end{equation*}
where we assume $\kappa_{\max} \geq \alpha$, and $c_u^\prime > 0$ is a large enough constant. This upper bound holds with a probability of at least $1-o(1)$.

Part (iv) can be verified by using the results of part (ii) and (iii) together with the definition of auxiliary function $\Phi(\cdot, \cdot)$:
\begin{align*}
    \sup_{\tau\in \mathcal{G}(u_T, v_T)}\Big| \Phi(\hat{\eta} - \eta^\star, \eta^\star) \Big| 
    &\leq \sup_{\tau\in \mathcal{G}(u_T, v_T)}\left\{ \Phi(\hat{\eta} - \eta^\star, \hat{\eta} - \eta^\star) \right\}^{\frac{1}{2}}\sup_{\tau\in \mathcal{G}(u_T, v_T)}\left\{ \Phi(\eta^\star, \eta^\star) \right\}^{\frac{1}{2}} \\
    &\leq c_u\left\{T^\gamma u_T\kappa_{\max}(1+\nu^2)\frac{\sigma^4}{\kappa_{\min}^2}\left(\frac{p^2\log(p\vee T)}{Tl_T}\right)\right\}^{\frac{1}{2}}\sqrt{2u_T\kappa_{\max}\xi_2^2} \\
    &\leq c_u'u_T\xi_2\sqrt{\kappa_{\max}(1+\nu^2)}\frac{\sigma^2}{\kappa_{\min}}\left( \frac{p^2\log(p\vee T)}{T^{1-\gamma}l_T} \right)^{\frac{1}{2}},
\end{align*}
which completes the proof of part (iv). 
\end{proof}

\begin{lemma}\label{lemma:c5}
Suppose that the Condition A-C hold. Let $u_T$ and $v_T$ be any non-negative sequences such that $(1/T') \leq v_T \leq u_T$, where $T' = T-p+1$. For any $0 < a< 1$, we choose $c_{a1} = c_{a2}2^{\frac{2}{\gamma_2}}$ with $c_{a2} \geq \sqrt{1/a}$. Then for any sufficiently large $T$, we have:

\begin{itemize}
    \item[(i)] For the term $J_1$ in proof of Lemma 4.1, 
    \begin{equation*}
        \begin{aligned}
            &\inf_{\substack{\tau \in \mathcal{G}(u_T, v_T) \\ \tau \geq \tau^\star}}\frac{1}{T'}\sum_{t=\lfloor T\tau^\star\rfloor+1}^{\lfloor T\tau \rfloor}(\hat{\eta}^\prime Z_t)^2 \\
            &\geq \kappa_{\min} \xi_2^2\left[ v_T - c_{a1}\sigma^2\sqrt{\frac{u_T}{T'}} - c_u \frac{u_T}{\xi_2}p\sqrt{\log(p\vee T)}\|\hat{\eta} - \eta^\star\|_2 \right]
        \end{aligned}
    \end{equation*}
    
    \item[(ii)] For the term $J_2$ in proof of Lemma 4.1, 
    \begin{equation*}
        \begin{aligned}
            &\sup_{\substack{\tau\in \mathcal{G}(u_T, v_T) \\ \tau\geq \tau^\star}}\frac{1}{T'}\left|\sum_{t=\lfloor T\tau^\star\rfloor+1}^{\lfloor T\tau\rfloor}\epsilon_t\hat{\eta}^\prime Z_t\right| \\
            &\leq c_{a1}2^{\frac{2}{\gamma_2}}K_\epsilon K_X\xi_2\sqrt{\frac{u_T}{T'}} + c_u K \sqrt{\frac{u_T\log(p\vee T)}{T}}\|\hat{\eta} - \eta^\star\|_1,
        \end{aligned}
    \end{equation*}
    with probability at least $1-a - o(1)$.
    
    \item[(iii)] For the term $J_3$ in the proof of Lemma 4.1, 
    \begin{equation*}
        \begin{aligned}
            &\sup_{\substack{\tau \in \mathcal{G}(u_T, v_T) \\ \tau \geq \tau^\star}}\frac{1}{T'}\left| \sum_{t=\lfloor T\tau^\star\rfloor+1}^{\lfloor T\tau \rfloor} (\hat{\phi}_2 - \phi_2^\star)^\prime Z_tZ_t^\prime \hat{\eta} \right| \\
            &\leq 2c_u(\sigma^2\vee \alpha)u_T\xi_2p\sqrt{\log(p\vee T)}\|\hat{\phi}_2 - \phi_2^\star\|_2\left(1 + \frac{1}{\xi_2}p\sqrt{\log(p\vee T)}\|\hat{\eta} - \eta^\star\|_2\right)
        \end{aligned}
    \end{equation*}
\end{itemize}
\end{lemma}

\begin{proof}[Proof of Lemma \ref{lemma:c5}]
To prove (i), note that we have $\Phi(\hat{\eta}, \hat{\eta}) = \Phi(\eta^\star, \eta^\star) + 2\Phi(\hat{\eta} - \eta^\star, \eta^\star) + \Phi(\hat{\eta} - \eta^\star, \hat{\eta} - \eta^\star)$, then we have:
\begin{equation*}
    \begin{aligned}
        \inf_{\substack{\tau \in \mathcal{G}(u_T, v_T)\\ \tau\geq \tau^\star}}\Phi(\hat{\eta}, \hat{\eta}) &\geq \inf_{\substack{\tau \in \mathcal{G}(u_T, v_T) \\ \tau \geq \tau^\star}}\Phi(\eta^\star, \eta^\star) - 2\sup_{\substack{\tau\in \mathcal{G}(u_T, v_T) \\\tau \geq \tau^\star}}\left| \Phi(\hat{\eta} - \eta^\star, \eta^\star)\right| \\
        &\geq v_T\kappa_{\min}\xi_2^2 - c_{a1}\sigma^2\xi_2^2\kappa_{\min}\sqrt{\frac{u_T}{T'}} - \kappa_{\min}u_T\xi_2p\sqrt{\log(p\vee T)}\|\hat{\eta} - \eta^\star\|_2.
    \end{aligned}
\end{equation*}


To prove (ii), it's easily to observe that by using triangle inequality, we have:
\begin{equation*}
    \sup_{\substack{\tau\in \mathcal{G}(u_T, v_T) \\ \tau\geq \tau^\star}}\frac{1}{T'}\left|\sum_{t=\lfloor T\tau^\star\rfloor+1}^{\lfloor T\tau\rfloor}\epsilon_t\hat{\eta}^\prime Z_t\right| \leq \sup_{\substack{\tau \in \mathcal{G}(u_T, v_T)\\ \tau \geq \tau^\star}}\frac{1}{T'}\left\{\left| \sum_{t=\lfloor T\tau^\star\rfloor+1}^{\lfloor T\tau\rfloor}\epsilon_t(\hat{\eta} - \eta^\star)^\prime Z_t \right| + \left|\sum_{t=\lfloor T\tau^\star\rfloor+1}^{\lfloor T\tau\rfloor}\epsilon_t\eta^{\star^\prime}Z_t \right| \right\}.
\end{equation*}
Then by applying the results of Lemma \ref{lemma:c2}(ii) and Lemma \ref{lemma:c4}, we can directly obtain the result.

To illustrate (iii), we use the pre-defined function $\Phi(\cdot, \cdot)$ as follows:
\begin{equation*}
    \begin{aligned}
        \sup_{\substack{\tau \in \mathcal{G}(u_T, v_T) \\ \tau \geq \tau^\star}}\Phi(\hat{\phi}_2 - \phi_2^\star, \hat{\eta} - \eta^\star) &\leq c_u(\sigma^2\vee \alpha)u_Tp^2\log(p\vee T)\|\hat{\phi}_2 - \phi_2^\star\|_2\|\hat{\eta} - \eta^\star\|_2, \\
        \sup_{\substack{\tau \in \mathcal{G}(u_T, v_T) \\ \tau \geq \tau^\star}}\left| \Phi(\hat{\phi}_2 - \phi_2^\star, \eta^\star) \right| &\leq c_u(\sigma^2 \vee \alpha)u_T\xi_2p\sqrt{\log (p\vee T)}\|\hat{\phi}_2 - \phi_2^\star\|_2,
    \end{aligned}
\end{equation*}
with probability at least $1 - a - o(1)$. Then by using the fact that $\Phi(\hat{\phi}_2 - \phi_2^\star, \hat{\eta}) \leq \left| \Phi(\hat{\phi}_2 - \phi_2^\star, \hat{\eta} - \eta^\star) \right| + \left| \Phi(\hat{\phi}_2 - \phi_2^\star, \eta^\star) \right|$, we have:
\begin{equation*}
    \begin{aligned}
        \sup_{\substack{\tau\in \mathcal{G}(u_T, v_T)\\ \tau \geq \tau^\star}}\left| \Phi(\hat{\phi}_2 - \phi_2^\star, \hat{\eta}) \right| 
        &\leq 2 c_u(\sigma^2 \vee \alpha)u_Tp^2\log(p\vee T)\|\hat{\phi}_2 - \phi_2^\star\|_2\|\hat{\eta} - \eta^\star\|_2 \\
        &+ 2 c_u(\sigma^2\vee \alpha)u_T\xi_2p\sqrt{\log(p\vee T)}\|\hat{\phi}_2 - \phi_2^\star\|_2 \\
        &\leq 2c_u(\sigma^2 \vee \alpha)u_T\xi_2p\sqrt{\log(p\vee T)}\|\hat{\phi}_2 - \phi_2^\star \|_2\left(1 + \frac{1}{\xi_2}p\sqrt{\log(p\vee T)}\|\hat{\eta} - \eta^\star\|_2\right),
    \end{aligned}
\end{equation*}
with probability at least $1-a-o(1)$. 
\end{proof}

\begin{lemma}
\label{lemma:c6}
Suppose Conditions B and C hold, then we have:
\begin{equation*}
    \begin{aligned}
        &(i)\quad \|\hat{\eta} - \eta^\star\|_2^2 \leq c_u(1+ \nu^2)\frac{\sigma^4}{\kappa_{\min}^2}\frac{p^2\log (p\vee T)}{Tl_T}, \\
        &(ii)\quad \|\hat{\eta} - \eta^\star\|_1 \leq c_u\sqrt{1+\nu^2}\frac{\sigma^2}{\kappa_{\min}}\sqrt{\frac{p^2\log (p\vee T)}{Tl_T}}, \\ 
        &(iii)\quad \frac{1}{\xi_2}\left(p^2\log(p\vee T)\|\hat{\eta} - \eta^\star\|_2^2\right)^{\frac{1}{2}} \leq \frac{c_{u1}}{T^b} = o(1), \\
        &(iv)\quad \frac{1}{\xi_2}\|\hat{\eta} - \eta^\star\|_1 \leq \frac{c_{u1}}{\sqrt{\log(p\vee T)}}.
    \end{aligned}
\end{equation*}
with probability at least $1 - o(1)$.

\end{lemma}

\begin{proof}[Proof of Lemma \ref{lemma:c6}]
According to the Condition C, part (i) can be obtained by:
\begin{equation*}
    \begin{aligned}
        \|\hat{\eta} - \eta^\star\|_2^2 &\leq 2\left( \|\hat{\phi}_1 - \phi_1^\star\|_2^2 + \|\hat{\phi}_2 - \phi_2^\star\|_2^2 \right) \\
        &\leq c_u(1+\nu^2)\frac{\sigma^4}{\kappa_{\min}^2}\frac{p^2\log (p\vee T)}{Tl_T},
    \end{aligned}
\end{equation*}
with probability at least $1 - o(1)$. Part (ii) can be obtained analogously according to Condition C as well as the result of Corollary 2 in \cite{wang2021consistent}. For part (iii), we directly use the result in part (i) combining with Condition B(b), 
\begin{equation*}
    \begin{aligned}
        \frac{1}{\xi_2}\left(p^2\log(p\vee T)\|\hat{\eta} - \eta^\star\|_2^2\right)^{\frac{1}{2}} 
        &\leq \frac{1}{\xi_2}\left(p^2\log (p\vee T)c_u^2(1+\nu^2)\frac{\sigma^4}{\kappa^2_{\min}}\frac{p^2\log (p\vee T)}{Tl_T}\right)^{\frac{1}{2}} \\
        &= \frac{1}{\xi_2}c_u\sqrt{1+\nu^2}\frac{\sigma^2}{\kappa_{\min}}\frac{p^2\log (p\vee T)}{\sqrt{Tl_T}} \leq \frac{c_{u1}}{T^b} = o(1),
    \end{aligned}
\end{equation*}


with probability at least $1 - o(1)$. The last part (iv) can be derived by using the same procedure. 
\end{proof}

\begin{lemma}\label{lemma:c7}
Let $\mathcal{C}(\tau, \phi_1, \phi_2)$ be as defined in \eqref{eq:argmax} and assume that Conditions A, B, and C are satisfied. Then for any universal constant $c_u > 0$, we have:

\begin{equation*}
    \sup_{\tau \in \mathcal{G}(c_uT^{-1}\xi_2^{-2}, 0)}\Big| \mathcal{C}(\tau, \hat{\phi}_1, \hat{\phi}_2) - \mathcal{C}(\tau, \phi_1^\star, \phi_2^\star) \Big| = o_p(1).
\end{equation*}
\end{lemma}

\begin{proof}[Proof of Lemma \ref{lemma:c7}]
For any $\tau \geq \tau^\star$, according to the definition of $\mathcal{C}(\tau, \hat{\phi}_1, \hat{\phi}_2)$, we firstly consider the following quantities:
\begin{equation}
    \label{eq:39}
    \begin{aligned}
        R_1 &= \sum_{t=\lfloor T\tau^\star \rfloor +1}^{\lfloor T\tau \rfloor}\left( \|\hat{\eta}^\prime Z_t\|_2^2 - 2\epsilon_t\hat{\eta}^\prime Z_t + 2(\hat{\phi}_2 - \phi_2^\star)^\prime Z_tZ_t^\prime \hat{\eta} \right) \\
        &\overset{\text{def}}{=} R_{11} - 2R_{12} + 2R_{13},
    \end{aligned}
\end{equation}
as well as 
\begin{equation}
    \label{eq:40}
    R_2 = \sum_{t=\lfloor T\tau^\star \rfloor+1}^{\lfloor T\tau \rfloor}\left( \|\eta^{\star^\prime}Z_t\|_2^2 - 2\epsilon_t\eta^{\star^\prime}Z_t \right) \overset{\text{def}}{=} R_{21} - 2R_{22}.
\end{equation}
Therefore, by directly using the algebraic rearrangements, we obtain that
\begin{equation}
    \mathcal{C}(\tau, \hat{\phi}_1, \hat{\phi}_2) - \mathcal{C}(\tau, \phi_1^\star, \phi_2^\star) = R_2 - R_1 = R_{21} - 2R_{22} - (R_{11} - 2R_{12} + 2R_{13}).
\end{equation}
Then we use the results of Lemma \ref{lemma:c8}. Consequently, we have 
\begin{equation*}
    \begin{aligned}
        \sup_{\tau \in \mathcal{G}(c_uT^{-1}\xi_2^{-2}, 0)}\Big| \mathcal{C}(\tau, \hat{\phi}_1, \hat{\phi}_2) - \mathcal{C}(\tau, \phi_1^\star, \phi_2^\star) \Big| &\leq \sup_{\tau \in \mathcal{G}(c_uT^{-1}\xi_2^{-1}, 0)}|R_{21} - R_{11}| \\
        &+ 2\sup_{\tau \in \mathcal{G}(c_uT^{-1}\xi_2^{-2}, 0)}|R_{22} - R_{12}| \\
        &+ 2\sup_{\tau \in \mathcal{G}(c_uT^{-1}\xi_2^{-2}, 0)}|R_{13}| = o_p(1).
    \end{aligned}
\end{equation*}
This completes the proof of this lemma.
\end{proof}

\begin{lemma}
\label{lemma:c9}
Suppose Conditions A, B, and C are satisfied. Furthermore, let $R_{11}, R_{12}, R_{13},$ $R_{21},$ and $R_{22}$ be defined based on Lemma \ref{lemma:c7}. Consider a universal constant $0 < c_u < +\infty$, then we have the following bounds
\begin{equation*}
    \begin{aligned}
        &(i)\quad  \sup_{\tau \in \mathcal{G}(c_uT^{-1}\xi_2^{-2}, 0)}|R_{21} - R_{11}| = o_p(1); \\
        &(ii)\quad \sup_{\tau \in \mathcal{G}(c_uT^{-1}\xi_2^{-2}, 0)}|R_{22} - R_{12}| = o_p(1); \\
        &(iii)\quad \sup_{\tau \in \mathcal{G}(c_uT^{-1}\xi_2^{-2}, 0)}|R_{13}| = o_p(1),
    \end{aligned}
\end{equation*}
where each holds with probability at least $1 - o(1)$.
\end{lemma}

\begin{proof}[Proof of Lemma \ref{lemma:c9}]
Let $\Phi(\cdot, \cdot)$ be the function that we defined in the proof of Lemma \ref{lemma:c5}. Then for part (i), we derive that
\begin{equation}\label{eq:42}
    \begin{aligned}
        &\sup_{\tau \in \mathcal{G}(c_uT^{-1}\xi_2^{-2}, 0)}|R_{21} - R_{11}| \\
        &= \sup_{\tau \in \mathcal{G}(c_uT^{-1}\xi_2^{-2}, 0)}\left|\sum_{t=\lfloor T\tau^\star \rfloor+1}^{\lfloor T\tau\rfloor} \|\eta^{\star^\prime}Z_t\|_2^2 - \sum_{t=\lfloor T\tau^\star \rfloor + 1}^{\lfloor T\tau \rfloor}\|\hat{\eta}^\prime Z_t\|_2^2 \right| \\
        &= \sup_{\tau \in \mathcal{G}(c_uT^{-1}\xi_2^{-2}, 0)}\left| \sum_{t=\lfloor T\tau^\star \rfloor + 1}^{\lfloor T\tau \rfloor}(\hat{\eta}^\prime - \eta^\star)^\prime Z_tZ_t^\prime (\hat{\eta} + \eta^\star) \right| \\
        &= \sup_{\tau \in \mathcal{G}(c_uT^{-1}\xi_2^{-2}, 0)}\left| T\Phi(\hat{\eta} - \eta^\star, \hat{\eta} - \eta^\star) + 2T\Phi(\hat{\eta} - \eta^\star, \eta^\star) \right|.
    \end{aligned}
\end{equation}
Next, we follow the similar procedure as we used in Lemma \ref{lemma:c6} part (iii), then it implies that
\begin{equation}\label{eq:43}
    \begin{aligned}
        \sup_{\tau \in \mathcal{G}(c_uT^{-1}\xi_2^{-2}, 0)} T\Phi(\hat{\eta} - \eta^\star, \hat{\eta} - \eta^\star) 
        &\leq c_u T^\gamma \xi_2^{-2}\frac{K_X^2 K_\epsilon^2}{\kappa^2_{\min}}\left( \frac{p^2\log(p \vee T)}{Tl_T} \right) \\
        &\leq \mathcal{O}\left(\frac{p^2\log(p\vee T)}{\xi_2^2T^{1-\gamma}l_T} \right) \leq \frac{c_{u1}^2}{T^{2b}} = o(1)
    \end{aligned}
\end{equation}
with probability at least $1 - o(1)$. The last equality holds because of Condition B(b). 

The upper bound of second term is obtained by the following procedures:
\begin{align*}
    &2T\sup_{\tau \in \mathcal{G}(u_T, 0)}|\Phi(\hat{\eta} - \eta^\star, \eta^\star)| \\
    &\leq 2T\sup_{\tau \in \mathcal{G}(u_T, 0)}\Phi(\hat{\eta} - \eta^\star, \hat{\eta} - \eta^\star)^{\frac{1}{2}}\sup_{\tau \in \mathcal{G}(c_uT^{-1}\xi_2^{-2}, 0)}\Phi(\eta^\star, \eta^\star)^{\frac{1}{2}} \\
    &\leq 2Tc_u\left(T^\gamma u_T\kappa_{\max}(1+\nu^2)\frac{\sigma^4}{\kappa_{\min}^2}\frac{p^2\log(p\vee T)}{Tl_T}\right)^{\frac{1}{2}}(u_T\kappa_{\max}\xi_2^2)^{\frac{1}{2}} \\
    &\leq 2c_u Tu_T \xi_2\sqrt{\kappa_{\max}(1+\nu^2)}\frac{\sigma^2}{\kappa_{\min}}\left(\frac{p^2\log(p\vee T)}{T^{1-\gamma}l_T}\right)^{\frac{1}{2}},
\end{align*}
with probability at least $1-a-o(1)$. Then by selecting $a \to 0$ as well as $u_T = c_{u1}T^{-1}\xi_2^{-2}$, it leads to

\begin{equation}\label{eq:44}
    \begin{aligned}
        &\sup_{\tau \in \mathcal{G}(c_uT^{-1}\xi_2^{-2}, 0)}2T|\Phi(\hat{\eta} - \eta^\star, \eta^\star)| \\
        &\leq 2c_u c_{u1}\xi_2^{-1}\sqrt{\kappa_{\max}(1+\nu^2)}\frac{\sigma^2}{\kappa_{\min}}\left( \frac{p^2\log(p\vee T)}{T^{1-\gamma}l_T} \right)^{\frac{1}{2}} \leq \frac{c_uc_{u1}}{T^b} = o(1),
    \end{aligned}
\end{equation}
with probability at least $1-o(1)$. The last inequality holds because of the Condition B as well. Therefore, bringing \eqref{eq:43} and \eqref{eq:44} back into \eqref{eq:42} leads to the final result of part (i). 

To show part (ii), note that we have
\begin{equation}
    \label{eq:45}
    \begin{aligned}
        \sup_{\tau\in \mathcal{G}(c_uT^{-1}\xi_2^{-2}, 0)}|R_{22} - R_{12}| 
        &= \sup_{\tau\in \mathcal{G}(c_uT^{-1}\xi_2^{-2}, 0)}\left| \sum_{t=\lfloor T\tau^\star \rfloor+1}^{\lfloor T\tau \rfloor}\epsilon_t(\hat{\eta} - \eta^\star)^\prime Z_t \right| \\
        &\leq c_{u1}K(\xi_2^{-2}p\log(p\vee T))^{\frac{1}{2}}\|\hat{\eta} - \eta^\star\|_1 \\
        &\leq c_{u1}c_uK(\xi_2^{-2}p\log(p\vee T))^{\frac{1}{2}}\sqrt{1+\nu^2}\frac{\sigma^2}{\kappa_{\min}}\sqrt{\frac{p^2\log(p\vee T)}{Tl_T}} \\
        &= c_{u1}c_u\sqrt{1+\nu^2}K\frac{\sigma^2}{\xi_2\kappa_{\min}}\frac{p^{\frac{3}{2}}\log(p\vee T)}{\sqrt{Tl_T}} = o(1),
    \end{aligned}
\end{equation}
with probability at least $1-o(1)$, where $K = 2^{2/\gamma_2}(K_X + K_\epsilon)^2$. The first inequality holds because of H\'{o}lder inequality, while the second inequality is satisfied due to the result in Lemma \ref{lemma:c6} part (ii), and the last equality holds because of the extra assumption in the statement. 

For part (iii), we firstly observe that, by using the definition of auxiliary function $\Phi(\cdot, \cdot)$, we have:
\begin{equation*}
    \begin{aligned}
        &\sup_{\tau\in \mathcal{G}(c_uT^{-1}\xi_2^{-2}, 0)}\left| \sum_{t=\lfloor T\tau^\star\rfloor + 1}^{\lfloor T\tau \rfloor}(\hat{\phi}_2 - \phi_2^\star)^\prime Z_tZ_t^\prime \hat{\eta} \right| \\ 
        &\leq T\sup_{\tau\in \mathcal{G}(c_uT^{-1}\xi_2^{-2}, 0)}\left\{ \Big|\Phi(\hat{\phi}_2 - \phi_2^\star, \hat{\eta} - \eta^\star) \Big| + \Big|\Phi(\hat{\phi}_2 - \phi_2^\star, \eta^\star) \Big| \right\}.
    \end{aligned}
\end{equation*}
Therefore, we will separately discuss the upper bounds of $\Phi(\hat{\phi}_2 - \phi_2^\star, \hat{\eta} - \eta^\star)$ and $\Phi(\hat{\phi}_2 - \phi_2^\star, \eta^\star)$.

For the term $\Phi(\hat{\phi}_2 - \phi_2^\star, \hat{\eta} - \eta^\star)$, applying Cauchy-Schwarz inequality together with parts (ii) and (iv) from Lemma \ref{lemma:c8} implies that:
\begin{equation}
    \label{eq:46}
    \begin{aligned}
        &\sup_{\tau \in \mathcal{G}(c_uT^{-1}\xi_2^{-2}, 0)}T\Big| \Phi(\hat{\phi}_2 - \phi_2^\star, \hat{\eta} - \eta^\star)\Big| \\
        &\leq T\sup_{\tau \in \mathcal{G}(c_uT^{-1}\xi_2^{-2}, 0)}\left\{\Phi(\hat{\phi}_2 - \phi_2^\star, \hat{\phi}_2 - \phi_2^\star)\right\}^{\frac{1}{2}}\sup_{\tau \in \mathcal{G}(c_uT^{-1}\xi_2^{-2}, 0)}\left\{\Phi(\hat{\eta} - \eta^\star, \hat{\eta} - \eta^\star) \right\}^{\frac{1}{2}} \\
        &\leq \left\{c_u\sqrt{1+\nu^2}\frac{\kappa_{\max}\sigma^2}{\xi^2_2\kappa_{\min}}\left(\frac{p^2\log(p\vee T)}{T^{1-\gamma}l_T}\right)^{\frac{1}{2}} \right\}\left\{c_u^\prime \sqrt{(1+\nu^2)}\frac{\sigma^2}{\kappa_{\min}}\left(\frac{p^2\log(p\vee T)}{T^{1-\gamma}l_T}\right)^{\frac{1}{2}}\right\} \\
        &\leq c_uc_u'\kappa_{\max}(1+\nu^2)\frac{\sigma^4}{\xi_2^2\kappa_{\min}^2}\frac{p^2\log(p\vee T)}{T^{1-\gamma}l_T} = \mathcal{O}\left(\frac{p^2\log(p\vee T)}{\xi^2_2T^{1-\gamma}l_T}\right) = \frac{c_u^2}{T^{2b}} = o(1),
    \end{aligned}
\end{equation}
which holds with probability at least $1-o(1)$. Then we similarly obtain that with probability at least $1-o(1)$, it satisfies:
\begin{equation}
    \label{eq:47}
    \begin{aligned}
        &\sup_{\tau \in \mathcal{G}(c_uT^{-1}\xi_2^{-2}, 0)}T\Big| \Phi(\hat{\phi}_2 - \phi_2^\star, \eta^\star) \Big| \\
        &\leq c_uc_{u1}\xi_2^{-1}\kappa_{\max}\sqrt{1+\nu^2}\frac{\sigma^2}{\kappa_{\min}}\left(\frac{p^2\log(p\vee T)}{T^{1-\gamma}l_T}\right)^{\frac{1}{2}} \leq \frac{c_{u2}}{T^b} = o(1).
    \end{aligned}
\end{equation}
Now we put \eqref{eq:46} and \eqref{eq:47} back into part (iii), and it yields the bound of $|R_{13}|$. 
\end{proof}

\begin{lemma}
\label{lemma:c10}
Suppose the jump size $\xi_2 \to 0$. Then we have:

\begin{equation*}
    \left| \sum_{t=\lfloor T\tau^\star\rfloor+1}^{\lfloor T\tau \rfloor} \left( \|\eta^{\star^\prime}Z_t\|_2^2 - \mathbb{E}\|\eta^{\star^\prime}Z_t\|_2^2 \right) \right| = o_p(1).
\end{equation*}
Additionally, for any constant $r>0$ and if $\xi_2^{-2}\mathbb{E}\|\eta^{\star^\prime}Z_t\|_2^2 \to \sigma^2$, then we have:
\begin{equation*}
    \sum_{t=\lfloor T\tau^\star\rfloor+1}^{\lfloor T\tau \rfloor}\|\eta^{\star^\prime}Z_t\|_2^2 \overset{P}{\to} r\sigma^2.
\end{equation*}
\end{lemma}

\begin{proof}[Proof of Lemma \ref{lemma:c10}]
Note that if we assume $\xi_2 \to 0$, then by assuming that $\tau \geq \tau^\star$, we can show that
\begin{equation}\label{eq:48}
    c_{u1}r\xi_2^{-2} \leq \lfloor T\tau^\star + r\xi_2^{-2}\rfloor - \lfloor T\tau^\star \rfloor \leq c_{u2}r\xi_2^{-2}.
\end{equation}
Then using the fact that $\eta^{\star^\prime}Z_{t}$ follows sub-Weibull($\gamma_2$) as well as the result in Lemma \ref{lemma:c1}, we are able to derive that $\|\eta^{\star^\prime}Z_{t}\|_2^2$ is sub-Weibull($\gamma_2/2$) distributed and upper bound by $K = 2^{2/\gamma_2}K_X^2$. 

For any $\tau \geq \tau^\star$, we can easily verify the inequality that as of $\xi_2 \to 0$,
\begin{equation*}
    T(\tau - \tau^\star) - 1 \leq \lfloor T\tau \rfloor - \lfloor T\tau^\star \rfloor \leq T(\tau - \tau^\star) + 1,
\end{equation*}
which directly implies \eqref{eq:48}. Now, we find the convergence of the summation. Note that $\mathbb{E}(\eta^{\star^\prime}Z_{t}) = 0$, $\mathbb{E}(\eta^{\star^\prime}Z_{t})^2 = \text{Var}(\eta^{\star^\prime}Z_{t}) = \eta^{\star^\prime}\Gamma_1^Z(0)\eta^\star$. We use a similar strategy as that of the Lemma \ref{lemma:c3}. Thus we rewrite the desired equation as follows and apply the result of Lemma 13 in \cite{wong2020lasso}:
\begin{equation}
    \label{eq:49}
    \begin{aligned}
        &\mathbb{P}\left(\left| \sum_{t=\lfloor T\tau^\star \rfloor + 1}^{\lfloor T\tau^\star  + r\xi_2^{-2}\rfloor} \left( \sum_{j=1}^p \left\{ (\eta_j^\star Z_{t,j})^2 - \mathbb{E}(\eta^\star_jZ_{t,j})^2 \right\}  \right) \right| > c_{u2}rd\right) \\
        \leq&\quad (c_{u2}r\xi_2^{-2})\exp\left(-\frac{(c_{u1}r\xi_2^{-2}d)^\gamma}{K^\gamma C_1}\right) + \exp\left(-\frac{c_{u1}r\xi_2^{-2}d^2}{K^2C_2}\right),
    \end{aligned}
\end{equation}
where $\gamma = (1/\gamma_1 + 2/\gamma_2)^{-1} < 1$ and the constants $C_1, C_2$ depend only on $\gamma_1, \gamma_2$ and $c_{u1}, c_{u2}$. Then by selecting $d$ to be any sequence that vanishing with slower rate than $\xi_2$, for example $b = \xi_2^{1-f}$ for any $0 < f< 1$, for sufficiently large $T$, we obtain that:
\begin{equation*}
    \left|\sum_{t=\lfloor T\tau^\star \rfloor +1}^{\lfloor T\tau^\star + r\xi_2^{-2}\rfloor} \left(\|\eta^{\star^\prime}Z_t\|_2^2 - \mathbb{E}(\|\eta^{\star^\prime}Z_t\|_2^2)\right) \right| = o_p(1).
\end{equation*}
This is the result that was desired in the first part. The second part of the lemma can be proved directly from the assumption that if $\xi_2^{-2}\mathbb{E}\|\eta^{\star^\prime}Z_t\|_2^2 \to \sigma^2$, then we obtain that
\begin{equation*}
    \left| \sum_{t=\lfloor T\tau^\star \rfloor+1}^{\lfloor T\tau^\star + r\xi_2^{-2} \rfloor} \mathbb{E}\|\eta^{\star^\prime}Z_t\|_2^2\right| = \left| r\xi_2^{-2}\mathbb{E}\|\eta^{\star^\prime}Z_t\|_2^2 \right| \to r\sigma^2.
\end{equation*}
\end{proof}

\begin{lemma}\label{lemma:12}
Assume there exist some constants $C, c>0$ such that with probability at least $1-C\exp(-plog p)$, the given AR($p$) process $\{X_t\}$ and the corresponding lagged $p$-dimensional vectors $\{Z_t\}$, which is introduced in (7) hold. For any time interval $1\leq l < u \leq T$ within the stationary segment, we have:
\begin{equation}
    \left\| \frac{1}{u-l}\sum_{t=l}^{u-1}Z_t\epsilon_t \right\|_\infty \leq c\sqrt{\frac{p\log (p\vee (u-l))}{u-l}}.
\end{equation}
\end{lemma}

\begin{proof}[Proof of Lemma~\ref{lemma:12}]
Since we know that $\epsilon_t \sim \mathcal{N}(0, \sigma^2)$, then by applying the Bernstein inequality on $\{Z_t\epsilon_t\}$, we obtain that
\begin{equation}
    \begin{aligned}
        \mathbb{P}\left(\frac{1}{u-l}\left\| \sum_{t=l}^{u-1}Z_t\epsilon_t \right\|_\infty > d\right) \leq 2\exp\left(-c_1(u-l)d^2\right),
    \end{aligned}
\end{equation}
where $c_1>0$ is some large enough constant. Then, by choosing $d = \sqrt{\frac{p\log (p \vee (u-l))}{u-l}}$, we have the probability that is $2\exp\left(-c_1p\log p\right)$, which achieves the final result.

\end{proof}

\section{D \quad Additional Simulation Results} \label{sec:AppD}
\renewcommand{\theequation}{D.\arabic{equation}}
\setcounter{equation}{0}
\renewcommand{\thelemma}{D.\arabic{lemma}}
\renewcommand{\thedefinition}{D.\arabic{definition}}
\setcounter{table}{0}
\renewcommand{\thetable}{D.\arabic{table}}
\setcounter{figure}{0}
\renewcommand{\thefigure}{D.\arabic{figure}}

In this section, we provide additional simulation results. Firstly, we consider more cases following the scenarios in the main manuscript, where the sample size is $T=500$ and the error term follows $N(0,0.01)$. The results are displayed in Tables \ref{Tab.settingI_more} -- \ref{Tab.settingIII_more}. We observe that the values of absolute bias and root MSE are relatively small, particularly for $\lfloor T\tilde{\tau} \rfloor$ compared to $\lfloor T\hat{\tau} \rfloor$. In addition, the coverage probabilities are close to the nominal levels.
Secondly, we assume that the sample size is $T=1000$ and the error term follows $N(0,1)$. The results are displayed in Figures \ref{Fig.AB_RMSE_CaseI_1000} -- \ref{Fig.CP_CaseV_1000}, and the performance of our estimators is quite good.

\begin{table}[ht]
\footnotesize
\caption{Performance of the model for Scenario I with $T=500$. } \label{Tab.settingI_more}
\centering
\setlength{\tabcolsep}{1pt} 
\begin{tabular}{@{} c|cccccccc @{}}
\hline
Coefficient & Truth & AB ($\lfloor T\hat{\tau} \rfloor$) & AB ($\lfloor T \tilde{\tau} \rfloor$) & RMSE ($\lfloor T\hat{\tau} \rfloor$) & RMSE ($\lfloor T \tilde{\tau} \rfloor$) & 90$\%$ CP 
($\lfloor T\tilde{\tau} \rfloor$) & 95$\%$ CP ($\lfloor T\tilde{\tau} \rfloor$) & 99$\%$ CP ($\lfloor T\tilde{\tau} \rfloor$) \\
\hline\hline
& $\lfloor T/3 \rfloor$ & 6.340 & 5.070 & 9.949 & 7.993 & 0.950 & 0.980 & 0.990 \\
$\theta=-0.7$ & $\lfloor T/2 \rfloor$ & 5.880 & 4.780 & 9.241 & 7.529 & 0.960 & 0.970 & 0.990 \\
$\phi=-0.5$ & $\lfloor 2T/3 \rfloor$ & 7.030 & 6.110 & 11.405 & 10.356 & 0.930 & 0.960 & 0.980 \\
& $\lfloor 4T/5 \rfloor$ & 5.880 & 5.100 & 11.735 & 10.913 & 0.950 & 0.960 & 0.970 \\
\hline\hline
& $\lfloor T/3 \rfloor$ & 32.960 & 33.040 & 70.385 & 70.844 & 0.860 & 0.870 & 0.900 \\
$\theta=0.7$ & $\lfloor T/2 \rfloor$ &  19.010 & 17.780 & 38.391 & 37.690 & 0.920 & 0.960 & 0.970 \\
$\phi=-0.5$ & $\lfloor 2T/3 \rfloor$ & 15.070 & 12.740 & 24.603 & 21.601 & 0.900 & 0.940 & 0.950 \\
& $\lfloor 4T/5 \rfloor$ &  22.660 & 19.820 & 41.759 & 39.308 & 0.820 & 0.840 & 0.880 \\
\hline\hline
& $\lfloor T/3 \rfloor$ & 6.590 & 5.790 & 11.004 & 9.875 & 0.930 & 0.960 & 0.990 \\
$\theta=-0.7$ & $\lfloor T/2 \rfloor$ & 7.150 & 6.340 & 10.642 & 9.463 & 0.900 & 0.950 & 0.980 \\
$\phi=0.5$ & $\lfloor 2T/3 \rfloor$ & 7.650 & 6.470 & 14.283 & 13.004 & 0.940 & 0.960 & 0.980 \\
& $\lfloor 4T/5 \rfloor$ & 6.880 & 5.580 & 14.451 & 11.039 & 0.930 & 0.950 & 0.980 \\
\hline\hline
& $\lfloor T/3 \rfloor$ & 32.150 & 30.750 & 62.569 & 62.576 & 0.820 & 0.880 & 0.910 \\
$\theta=0.7$ & $\lfloor T/2 \rfloor$ & 17.280 & 16.350 & 30.687 & 30.059 & 0.950 & 0.980 & 0.990 \\
$\phi=0.5$ & $\lfloor 2T/3 \rfloor$ & 11.890 & 10.660 & 20.663 & 18.765 & 0.930 & 0.950 & 0.970 \\
& $\lfloor 4T/5 \rfloor$ &  21.580 & 19.390 & 37.252 & 34.955 & 0.770 & 0.840 & 0.890 \\
\hline
\end{tabular}
\end{table}

\begin{table}[ht]
\footnotesize
\caption{Performance of the model for Scenario II with $T=500$.} \label{Tab.settingII_more1}
\centering
\setlength{\tabcolsep}{1pt} 
\begin{tabular}{@{} c|cccccccc @{}}
\hline
Coefficient & Truth & AB ($\lfloor T\hat{\tau} \rfloor$) & AB ($\lfloor T\tilde{\tau} \rfloor$) & RMSE ($\lfloor T\hat{\tau} \rfloor$) & RMSE ($\lfloor T\tilde{\tau} \rfloor$) & 90$\%$ CP 
($\lfloor T\tilde{\tau} \rfloor$) & 95$\%$ CP ($\lfloor T\tilde{\tau} \rfloor$) & 99$\%$ CP ($\lfloor T\tilde{\tau} \rfloor$) \\\hline\hline
& $\lfloor T/3 \rfloor$ & 5.080 & 4.920 & 8.117 & 7.947 & 0.920 & 0.950 & 0.970 \\
$\theta=-0.9$ & $\lfloor T/2 \rfloor$ & 4.880 & 4.700 & 8.393 & 8.019 & 0.950 & 0.970 & 0.970 \\
$\phi=0.1$ & $\lfloor 2T/3 \rfloor$ & 4.140 & 3.950 & 6.841 & 6.701 & 0.950 & 0.970 & 0.990 \\
& $\lfloor 4T/5 \rfloor$ & 3.680 & 3.030 & 7.139 & 5.696 & 0.970 & 0.980 & 0.990 \\
\hline\hline
& $\lfloor T/3 \rfloor$ & 4.040 & 3.510 & 6.589 & 5.709 & 0.880 & 0.910 & 0.960 \\
$\theta=-0.9$ & $\lfloor T/2 \rfloor$ & 3.780 & 3.620 & 6.585 & 6.153 & 0.920 & 0.940 & 0.960 \\
$\phi=0.3$ & $\lfloor 2T/3 \rfloor$ & 3.250 & 2.880 & 5.379 & 4.695 & 0.940 & 0.950 & 0.990 \\
& $\lfloor 4T/5 \rfloor$ & 2.880 & 2.690 & 5.769 & 5.474 & 0.950 & 0.980 & 0.990 \\
\hline\hline
& $\lfloor T/3 \rfloor$ & 1.980 & 1.790 & 3.597 & 3.369 & 0.890 & 0.890 & 0.910 \\
$\theta=-0.9$ & $\lfloor T/2 \rfloor$ & 1.760 & 1.810 & 2.926 & 2.972 & 0.860 & 0.880 & 0.930 \\
$\phi=0.7$ & $\lfloor 2T/3 \rfloor$ & 2.160 & 1.630 & 3.734 & 2.648 & 0.900 & 0.900 & 0.940 \\
& $\lfloor 4T/5 \rfloor$ & 1.510 & 1.500 & 2.696 & 2.702 & 0.920 & 0.950 & 0.970 \\
\hline\hline
& $\lfloor T/3 \rfloor$ & 1.030 & 1.010 & 1.775 & 1.752 & 0.930 & 0.930 & 0.940 \\
$\theta=-0.9$ & $\lfloor T/2 \rfloor$ & 0.990 & 0.940 & 1.565 & 1.562 & 0.940 & 0.950 & 0.960 \\
$\phi=0.9$ & $\lfloor 2T/3 \rfloor$ & 1.290 & 1,260 & 2.170 & 2.163 & 0.900 & 0.910 & 0.920 \\
& $\lfloor 4T/5 \rfloor$ & 1.200 & 1.200 & 1.897 & 1.860 & 0.890 & 0.890 & 0.930 \\
\hline
\end{tabular}
\end{table}

\begin{table}[ht]
\footnotesize
\caption{Performance of the model for Scenario II with $T=500$.} \label{Tab.settingII_more2}
\centering
\setlength{\tabcolsep}{1pt} 
\begin{tabular}{@{} c|cccccccc @{}}
\hline
Coefficient & Truth & AB ($\lfloor T\hat{\tau} \rfloor$) & AB ($\lfloor T\tilde{\tau} \rfloor$) & RMSE ($\lfloor T\hat{\tau} \rfloor$) & RMSE ($\lfloor T\tilde{\tau} \rfloor$) & 90$\%$ CP 
($\lfloor T\tilde{\tau} \rfloor$) & 95$\%$ CP ($\lfloor T\tilde{\tau} \rfloor$) & 99$\%$ CP ($\lfloor T\tilde{\tau} \rfloor$) \\\hline\hline
& $\lfloor T/3 \rfloor$ & 6.510 & 6.350 & 10.074 & 9.976 & 0.940 & 0.970 & 1.000 \\
$\theta=-0.9$ & $\lfloor T/2 \rfloor$ & 7.920 & 7.830 & 12.504 & 12.349 & 0.940 & 0.950 & 0.970 \\
$\phi=-0.1$ & $\lfloor 2T/3 \rfloor$ & 7.120 & 6.210 & 11.889 & 10.449 & 0.950 & 0.960 & 0.990 \\
& $\lfloor 4T/5 \rfloor$ &  6.780 & 6.200 & 13.324 & 12.636 & 0.920 & 0.940 & 0.950 \\
\hline\hline
& $\lfloor T/3 \rfloor$ & 9.630 & 9.050 & 14.999 & 13.873 & 0.930 & 0.980 & 0.990 \\
$\theta=-0.9$ & $\lfloor T/2 \rfloor$ & 11.660 & 10.350 & 17.311 & 14.855 & 0.930 & 0.940 & 0.980 \\
$\phi=-0.3$ & $\lfloor 2T/3 \rfloor$ & 9.850 & 9.490 & 15.989 & 15.391 & 0.950 & 0.960 & 0.990 \\
& $\lfloor 4T/5 \rfloor$ &  11.120 & 9.580 & 20.316 & 16.468 & 0.890 & 0.920 & 0.950 \\
\hline\hline
& $\lfloor T/3 \rfloor$ &  16.540 & 16.190 & 36.083 & 34.738 & 0.900 & 0.950 & 0.980 \\
$\theta=-0.9$ & $\lfloor T/2 \rfloor$ & 15.070 & 13.900 & 29.896 & 28.574 & 0.920 & 0.970 & 0.990 \\
$\phi=-0.5$ & $\lfloor 2T/3 \rfloor$ & 12.720 & 11.460 & 21.085 & 17.281 & 0.910 & 0.940 & 0.970 \\
& $\lfloor 4T/5 \rfloor$ &  14.370 & 12.440 & 25.380 & 21.172 & 0.870 & 0.900 & 0.940 \\
\hline
\end{tabular}
\end{table}

\begin{table}[ht]
\footnotesize
\caption{Performance of the model for Scenario II with $T=500$.} \label{Tab.settingII_more3}
\centering
\setlength{\tabcolsep}{1pt} 
\begin{tabular}{@{} c|cccccccc @{}}
\hline
Coefficient & Truth & AB ($\lfloor T\hat{\tau} \rfloor$) & AB ($\lfloor T\tilde{\tau} \rfloor$) & RMSE ($\lfloor T\hat{\tau} \rfloor$) & RMSE ($\lfloor T\tilde{\tau} \rfloor$) & 90$\%$ CP 
($\lfloor T\tilde{\tau} \rfloor$) & 95$\%$ CP ($\lfloor T \tilde{\tau} \rfloor$) & 99$\%$ CP ($\lfloor T\tilde{\tau} \rfloor$) \\\hline\hline
& $\lfloor T/3 \rfloor$ & 4.400 & 4.230 & 6.657 & 6.332 & 0.960 & 0.980 & 0.990 \\
$\theta=0.9$ & $\lfloor T/2 \rfloor$ &  4.500 & 3.800 & 6.850 & 6.037 & 0.910 & 0.960 & 1.000 \\
$\phi=-0.1$ & $\lfloor 2T/3 \rfloor$ & 4.300 & 3.740 & 6.597 & 5.753 & 0.930 & 0.960 & 0.980 \\
& $\lfloor 4T/5 \rfloor$ & 5.420 & 4.760 & 8.415 & 7.597 & 0.880 & 0.920 & 0.970 \\
\hline\hline
& $\lfloor T/3 \rfloor$ & 3.330 & 2.920 & 5.815 & 4.315 & 0.960 & 0.970 & 0.990 \\
$\theta=0.9$ & $\lfloor T/2 \rfloor$ & 3.040 & 2.810 & 4.652 & 4.208 & 0.920 & 0.960 & 0.990 \\
$\phi=-0.3$ & $\lfloor 2T/3 \rfloor$ & 3.140 & 2.560 & 5.667 & 4.200 & 0.940 & 0.950 & 0.970 \\
& $\lfloor 4T/5 \rfloor$ & 3.770 & 3.390 & 5.832 & 5.438 & 0.860 & 0.890 & 0.950 \\
\hline\hline
& $\lfloor N/3 \rfloor$ & 2.790 & 2.230 & 4.520 & 3.375 & 0.940 & 0.970 & 0.970 \\
$\theta=0.9$ & $\lfloor N/2 \rfloor$ & 2.330 & 2.270 & 3.754 & 3.618 & 0.910 & 0.920 & 0.970 \\
$\phi=-0.5$ & $\lfloor 2N/3 \rfloor$ & 2.810 & 2.180 & 5.397 & 3.868 & 0.900 & 0.930 & 0.940 \\
& $\lfloor 4N/5 \rfloor$ & 3.000 & 2.940 & 5.075 & 5.073 & 0.820 & 0.830 & 0.890 \\
\hline
\end{tabular}
\end{table}

\begin{table}[ht]
\footnotesize
\caption{Performance of the model for Scenario II with $T=500$.} \label{Tab.settingII_more4}
\centering
\setlength{\tabcolsep}{1pt} 
\begin{tabular}{@{} c|cccccccc @{}}
\hline
Coefficient & Truth & AB ($\lfloor T\hat{\tau} \rfloor$) & AB ($\lfloor T\tilde{\tau} \rfloor$) & RMSE ($\lfloor T\hat{\tau} \rfloor$) & RMSE ($\lfloor T\tilde{\tau} \rfloor$) & 90$\%$ CP 
($\lfloor T\tilde{\tau} \rfloor$) & 95$\%$ CP ($\lfloor T\tilde{\tau} \rfloor$) & 99$\%$ CP ($\lfloor T\tilde{\tau} \rfloor$) \\\hline\hline
& $\lfloor T/3 \rfloor$ & 7.440 & 7.400 & 13.129 & 13.148 & 0.980 & 0.990 & 0.990 \\
$\theta=0.9$ & $\lfloor T/2 \rfloor$ & 7.050 & 6.100 & 12.308 & 11.217 & 0.970 & 0.980 & 1.000 \\
$\phi=0.1$ & $\lfloor 2T/3 \rfloor$ & 5.910 & 5.140 & 9.692 & 8.445 & 0.950 & 0.960 & 0.990 \\
& $\lfloor 4T/5 \rfloor$ &  7.380 & 6.340 & 11.688 & 10.288 & 0.900 & 0.940 & 0.990 \\
\hline\hline
& $\lfloor T/3 \rfloor$ & 12.040 & 11.970 & 22.967 & 23.012 & 0.970 & 0.970 & 0.980 \\
$\theta=0.9$ & $\lfloor T/2 \rfloor$ & 10.090 & 9.550 & 17.545 & 17.122 & 0.960 & 0.990 & 1.000 \\
$\phi=0.3$ & $\lfloor 2T/3 \rfloor$ & 8.420 & 7.110 & 13.646 & 11.904 & 0.950 & 0.960 & 0.990 \\
& $\lfloor 4T/5 \rfloor$ &  10.100 & 9.100 &  16.666 & 15.381 & 0.930 & 0.940 & 0.960 \\
\hline\hline
& $\lfloor T/3 \rfloor$ & 16.030 & 16.140 & 31.987 & 32.021 & 0.920 & 0.950 & 0.980 \\
$\theta=0.9$ & $\lfloor T/2 \rfloor$ & 13.340 & 12.300 & 20.712 & 20.175 & 0.950 & 0.980 & 0.980 \\
$\phi=0.5$ & $\lfloor 2T/3 \rfloor$ &  12.600 & 11.350 &  22.094 & 20.954 & 0.900 & 0.930 & 0.950  \\
& $\lfloor 4T/5 \rfloor$ & 14.930 & 13.980 & 29.411 & 28.339 & 0.880 & 0.890 & 0.930 \\
\hline
\end{tabular}
\end{table}

\begin{table}[ht]
\footnotesize
\caption{Performance of the model for Scenario II with $T=500$.} \label{Tab.settingII_more5}
\centering
\setlength{\tabcolsep}{1pt} 
\begin{tabular}{@{} c|cccccccc @{}}
\hline
Coefficient & Truth & AB ($\lfloor T\hat{\tau} \rfloor$) & AB ($\lfloor T\tilde{\tau} \rfloor$) & RMSE ($\lfloor T\hat{\tau} \rfloor$) & RMSE ($\lfloor T\tilde{\tau} \rfloor$) & 90$\%$ CP 
($\lfloor T\tilde{\tau} \rfloor$) & 95$\%$ CP ($\lfloor T\tilde{\tau} \rfloor$) & 99$\%$ CP ($\lfloor T\tilde{\tau} \rfloor$) \\\hline\hline
& $\lfloor T/3 \rfloor$ & 7.920 & 7.460 &  14.526 & 13.471 & 0.900 & 0.940 & 0.990 \\
$\theta=0.7$ & $\lfloor T/2 \rfloor$ &  7.170 & 6.950 & 12.158 & 12.052 & 0.930 & 0.970 & 0.980 \\
$\phi=-0.1$ & $\lfloor 2T/3 \rfloor$ & 7.170 & 6.130 & 11.876 & 10.527 & 0.910 & 0.940 & 0.970 \\
& $\lfloor 4T/5 \rfloor$ & 9.230 & 7.820 & 15.842 & 14.153 & 0.830 & 0.920 & 0.950 \\
\hline\hline
& $\lfloor T/3 \rfloor$ & 4.550 & 4.000 & 7.079 & 5.762 & 0.900 & 0.970 & 0.990 \\
$\theta=0.7$ & $\lfloor T/2 \rfloor$ & 4.640 & 3.800 & 6.562 & 5.694 & 0.920 & 0.960 & 1.000 \\
$\phi=-0.3$ & $\lfloor 2T/3 \rfloor$ & 4.270 & 3.930 & 6.890 & 6.070 & 0.930 & 0.950 & 0.970 \\
& $\lfloor 4T/5 \rfloor$ & 5.020 & 4.160 & 7.645 & 6.430 & 0.910 & 0.930 & 0.960 \\
\hline\hline
& $\lfloor T/3 \rfloor$ & 3.640 & 2.980 & 6.104 & 4.506 & 0.930 & 0.950 & 0.970 \\
$\theta=0.7$ & $\lfloor T/2 \rfloor$ &  3.130 & 2.600 & 4.742 & 4.040 & 0.920 & 0.940 & 0.990 \\
$\phi=-0.5$ & $\lfloor 2T/3 \rfloor$ & 3.190 & 2.950 & 5.527 & 4.738 & 0.900 & 0.940 & 0.970 \\
& $\lfloor 4T/5 \rfloor$ & 3.850 & 3.650 & 6.190 & 5.986 & 0.810 & 0.870 & 0.940 \\
\hline
\end{tabular}
\end{table}

\begin{table}[ht]
\footnotesize
\caption{Performance of the model for Scenario II with $T=500$.} \label{Tab.settingII_more6}
\centering
\setlength{\tabcolsep}{1pt} 
\begin{tabular}{@{} c|cccccccc @{}}
\hline
Coefficient & Truth & AB ($\lfloor T\hat{\tau} \rfloor$) & AB ($\lfloor T\tilde{\tau} \rfloor$) & RMSE ($\lfloor T\hat{\tau} \rfloor$) & RMSE ($\lfloor T\tilde{\tau} \rfloor$) & 90$\%$ CP 
($\lfloor T\tilde{\tau} \rfloor$) & 95$\%$ CP ($\lfloor T\tilde{\tau} \rfloor$) & 99$\%$ CP ($\lfloor T\tilde{\tau} \rfloor$) \\\hline\hline
& $\lfloor T/3 \rfloor$ & 6.410 & 5.830 & 9.413 & 8.566 & 0.930 & 0.970 & 1.000 \\
$\theta=-0.7$ & $\lfloor T/2 \rfloor$ & 7.220 & 6.790 & 10.928 & 10.294 & 0.940 & 0.960 & 0.980 \\
$\phi=0.1$ & $\lfloor 2T/3 \rfloor$ & 8.810 & 7.840 & 14.367 & 12.939 & 0.910 & 0.920 & 0.970 \\
& $\lfloor 4T/5 \rfloor$ & 7.080 & 6.710 & 13.360 & 13.159 & 0.910 & 0.930 & 0.950 \\
\hline\hline
& $\lfloor T/3 \rfloor$ & 5.140 & 4.130 & 7.996 & 6.172 & 0.930 & 0.960 & 0.990 \\
$\theta=-0.7$ & $\lfloor T/2 \rfloor$ &  5.680 & 4.650 & 9.438 & 7.876 & 0.940 & 0.970 & 0.980 \\
$\phi=0.3$ & $\lfloor 2T/3 \rfloor$ &  4.990 & 3.800 &  8.272 & 6.263 & 0.920 & 0.960 & 1.000 \\
& $\lfloor 4T/5 \rfloor$ & 4.240 & 3.880 & 8.070 & 7.593 & 0.930 & 0.970 & 0.980 \\
\hline\hline
& $\lfloor T/3 \rfloor$ & 3.790 & 3.040 & 6.002 & 5.079 & 0.880 & 0.900 & 0.960 \\
$\theta=-0.7$ & $\lfloor T/2 \rfloor$ &  3.170 & 3.580 & 5.160 & 6.263 & 0.890 & 0.920 & 0.970 \\
$\phi=0.5$ & $\lfloor 2T/3 \rfloor$ &  3.450 & 2.700 & 5.709 & 4.701 & 0.930 & 0.940 & 0.980 \\
& $\lfloor 4T/5 \rfloor$ & 2.360 & 2.170 & 4.273 & 3.931 & 0.960 & 0.970 & 1.000 \\
\hline
\end{tabular}
\end{table}

\begin{table}[ht]
\footnotesize
\caption{Performance of the model for Scenario III with $T=500$.} \label{Tab.settingIII_more}
\centering
\setlength{\tabcolsep}{1pt} 
\begin{tabular}{@{} c|cccccccc @{}}
\hline
Coefficient & Truth & AB ($\lfloor T\hat{\tau} \rfloor$) & AB ($\lfloor T\tilde{\tau} \rfloor$) & RMSE ($\lfloor T\hat{\tau} \rfloor$) & RMSE ($\lfloor T\tilde{\tau} \rfloor$) & 90$\%$ CP 
($\lfloor T\tilde{\tau} \rfloor$) & 95$\%$ CP ($\lfloor T\tilde{\tau} \rfloor$) & 99$\%$ CP ($\lfloor T\tilde{\tau} \rfloor$) \\\hline\hline
& $\lfloor T/3 \rfloor$ & 5.500 & 4.730 & 8.953 & 7.434 & 0.870 & 0.930 & 0.980 \\
$\phi=-0.1$ & $\lfloor T/2 \rfloor$ & 4.900 & 4.220 & 7.362 & 6.568 & 0.920 & 0.930 & 1.000 \\
& $\lfloor 2T/3 \rfloor$ & 4.730 & 3.880 & 7.838 & 6.384 & 0.910 & 0.940 & 0.970 \\
& $\lfloor 4T/5 \rfloor$ & 5.840 & 4.610 & 9.305 & 7.148 & 0.840 & 0.910 & 0.970 \\
\hline\hline
& $\lfloor T/3 \rfloor$ & 3.950 & 3.780 & 6.706 & 5.951 & 0.900 & 0.900 & 0.950 \\
$\phi=-0.3$ & $\lfloor T/2 \rfloor$ & 3.700 & 3.280 & 5.628 & 5.259 & 0.900 & 0.940 & 0.980 \\
& $\lfloor 2T/3 \rfloor$ & 2.780 & 2.530 & 4.324 & 3.711 & 0.920 & 0.940 & 0.990 \\
& $\lfloor 4T/5 \rfloor$ & 4.310 & 3.700 & 6.520 & 5.810 & 0.840 & 0.870 & 0.960 \\
\hline\hline
& $\lfloor T/3 \rfloor$ & 2.370 & 2.190 & 3.632 & 3.384 & 0.940 & 0.960 & 0.980 \\
$\phi=-0.5$ & $\lfloor T/2 \rfloor$ & 2.460 & 2.340 & 3.876 & 3.723 & 0.890 & 0.920 & 0.960 \\
& $\lfloor 2T/3 \rfloor$ & 1.860 & 1.980 & 2.821 & 2.946 & 0.900 & 0.950 & 0.970 \\
& $\lfloor 4T/5 \rfloor$ & 3.000 & 2.790 & 4.790 & 4.647 & 0.820 & 0.850 & 0.940 \\
\hline\hline
& $\lfloor T/3 \rfloor$ & 1.810 & 1.570 & 3.183 & 2.610 & 0.900 & 0.930 & 0.950 \\
$\phi=-0.7$ & $\lfloor T/2 \rfloor$ & 1.970 & 1.990 & 3.551 & 3.557 & 0.860 & 0.880 & 0.920 \\
& $\lfloor 2T/3 \rfloor$ & 1.830 & 1.700 & 3.132 &  2.877 & 0.870 & 0.890 & 0.930 \\
& $\lfloor 4T/5 \rfloor$ & 2.140 & 1.950 & 3.608 & 3.477 & 0.850 & 0.870 & 0.930 \\
\hline
\end{tabular}
\end{table}

\begin{figure}[ht]
\centering
\begin{tabular}{ cc }
\includegraphics[width=0.4\textwidth]{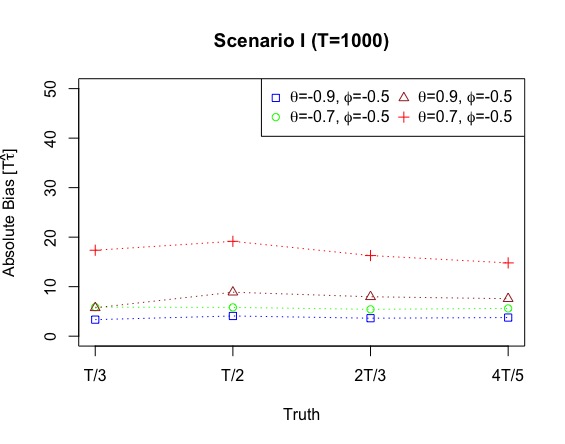} &
\includegraphics[width=0.4\textwidth]{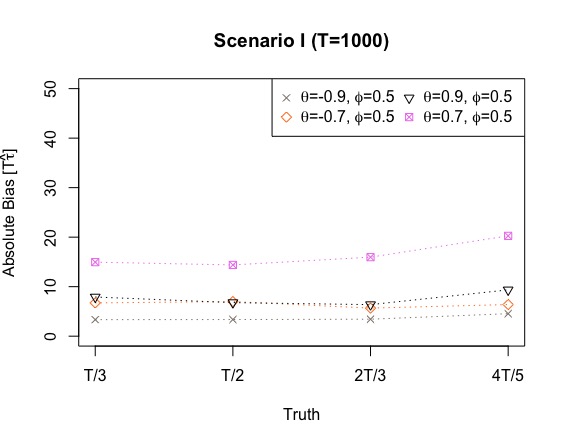} \\
\includegraphics[width=0.4\textwidth]{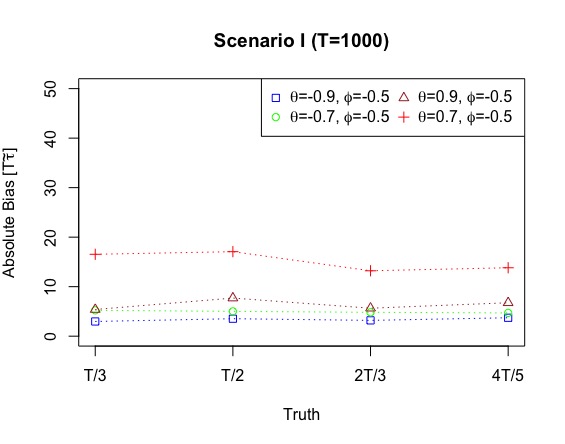} &
\includegraphics[width=0.4\textwidth]{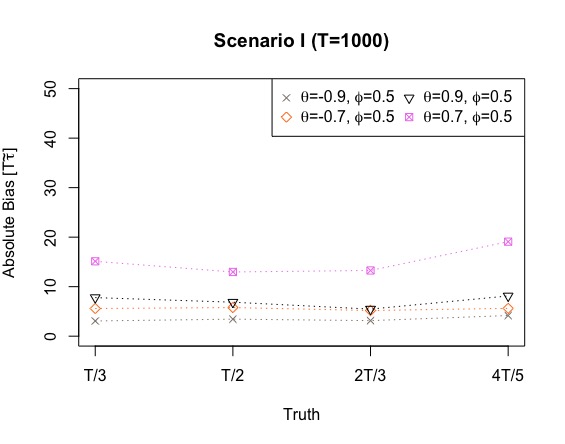} \\
\includegraphics[width=0.4\textwidth]{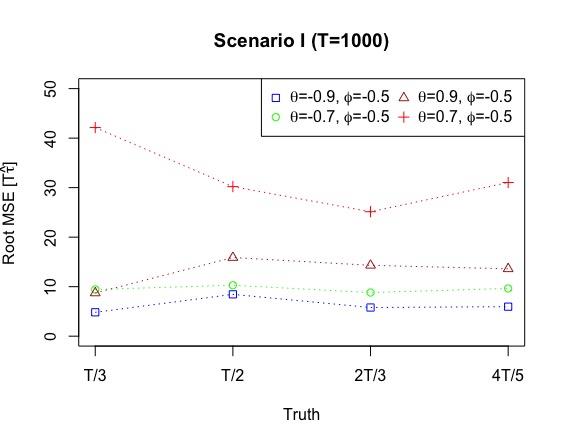} &
\includegraphics[width=0.4\textwidth]{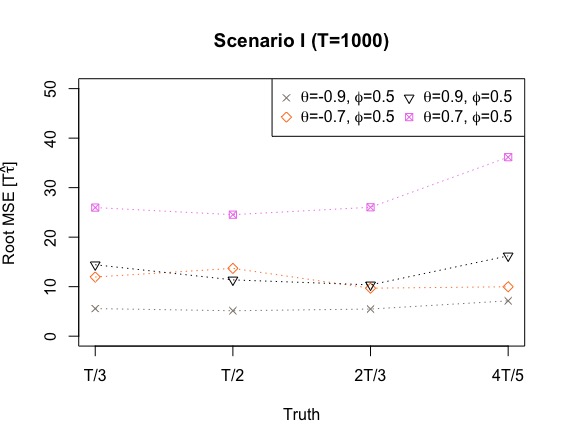} \\
\includegraphics[width=0.4\textwidth]{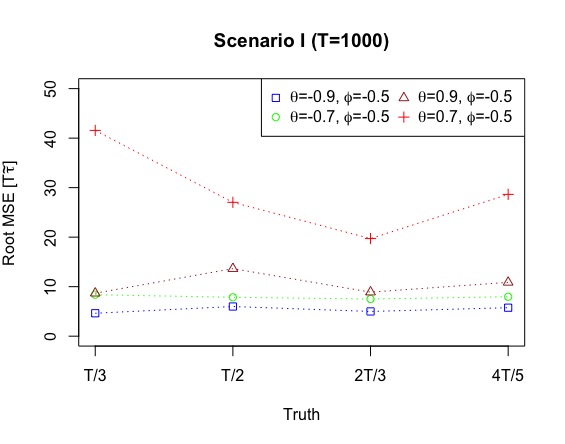} &
\includegraphics[width=0.4\textwidth]{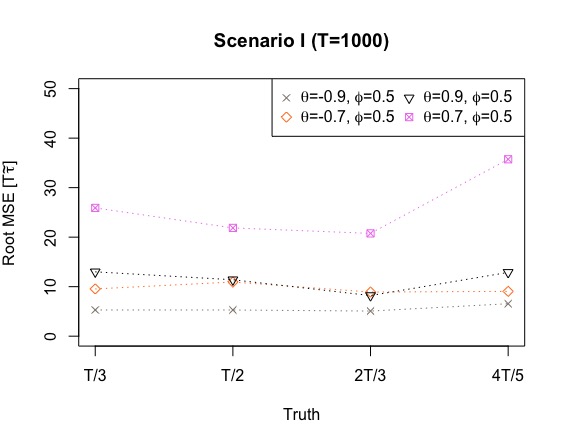}
\end{tabular}
\caption{AB and RMSE of $\lfloor T\hat{\tau} \rfloor$ and $\lfloor T\tilde{\tau} \rfloor$ for Scenario I with $T=1000$.}
\label{Fig.AB_RMSE_CaseI_1000}
\end{figure}

\begin{figure}[ht]
\centering
\begin{tabular}{ cc }
\includegraphics[width=0.4\textwidth]{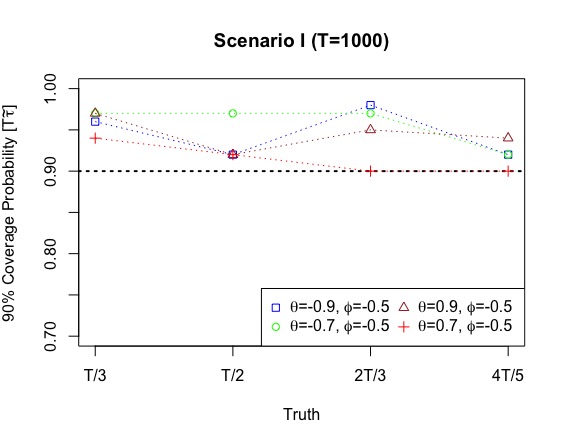} &
\includegraphics[width=0.4\textwidth]{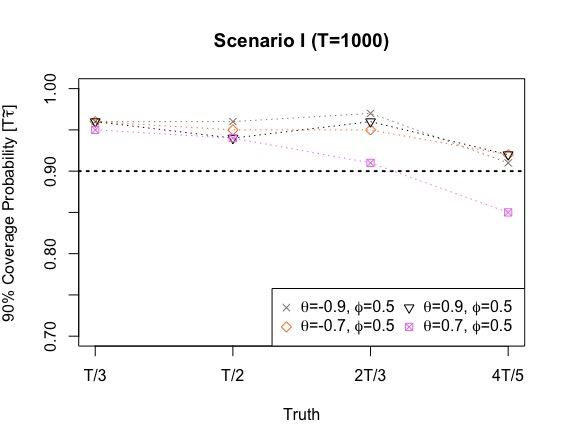} \\
\includegraphics[width=0.4\textwidth]{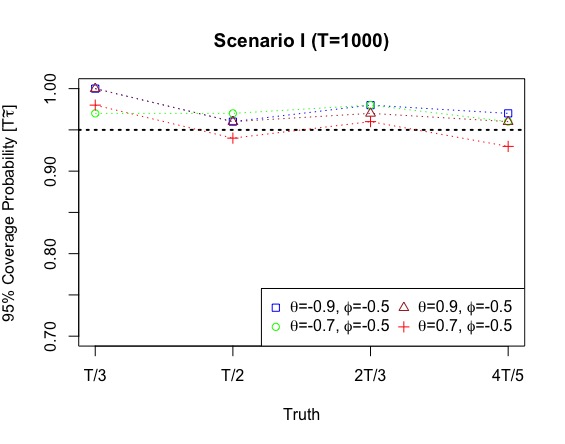} &
\includegraphics[width=0.4\textwidth]{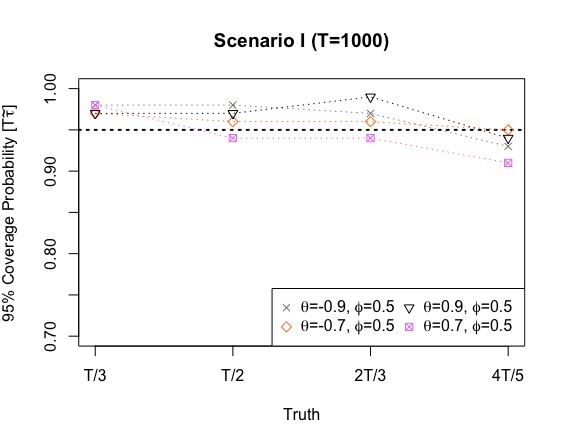} \\
\includegraphics[width=0.4\textwidth]{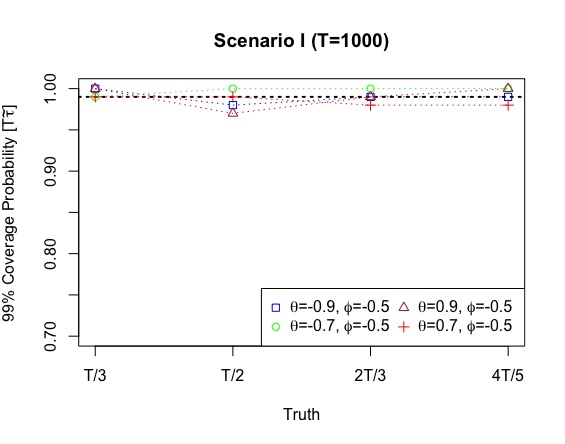} &
\includegraphics[width=0.4\textwidth]{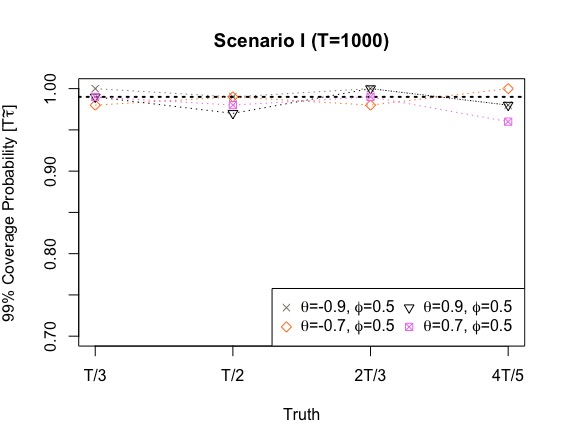} 
\end{tabular}
\caption{($90\%$, $95\%$, $99\%$) CP for $\lfloor T\tilde{\tau} \rfloor$ for Scenario I with $T=1000$.}
\label{Fig.CP_CaseI_1000}
\end{figure}

\begin{figure}[ht]
\centering
\begin{tabular}{ cc }
\includegraphics[width=0.4\textwidth]{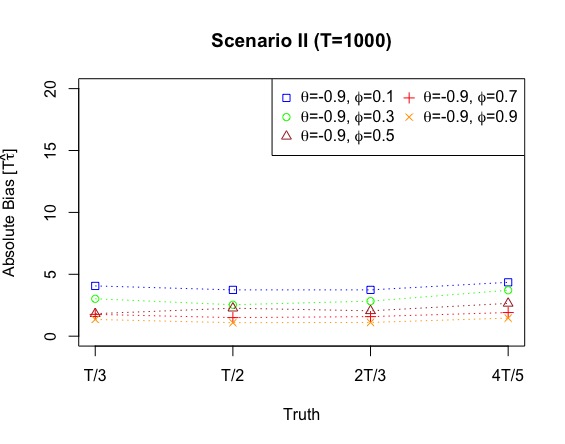} &
\includegraphics[width=0.4\textwidth]{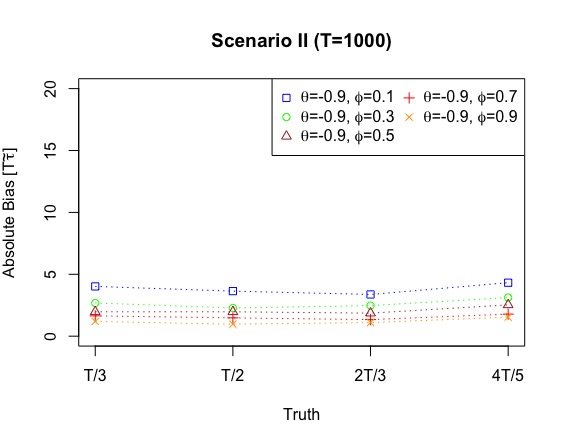} \\
\includegraphics[width=0.4\textwidth]{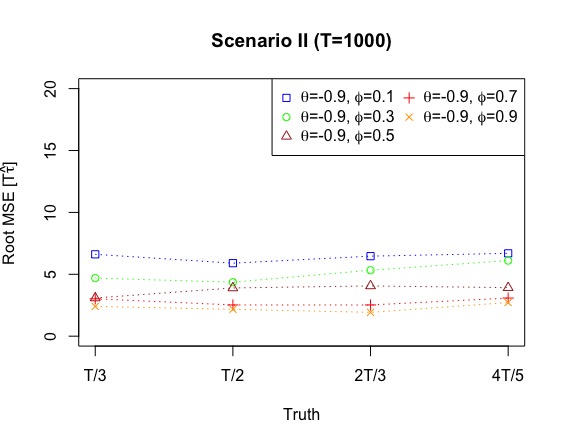} &
\includegraphics[width=0.4\textwidth]{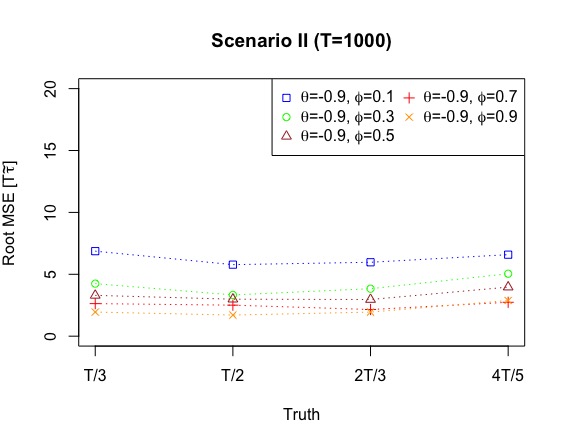} 
\end{tabular}
\caption{AB and RMSE of $\lfloor T\hat{\tau} \rfloor$ and $\lfloor T\tilde{\tau} \rfloor$ for Scenario II with $T=1000$.}
\label{Fig.AB_RMSE_CaseII_1000}
\end{figure}

\begin{figure}[ht]
\centering
\begin{tabular}{ cc }
\includegraphics[width=0.4\textwidth]{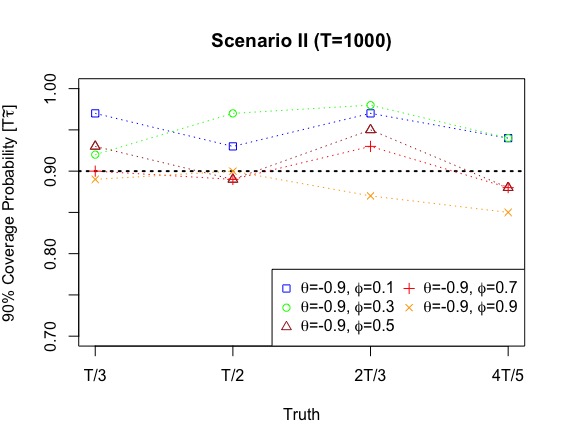} &
\includegraphics[width=0.4\textwidth]{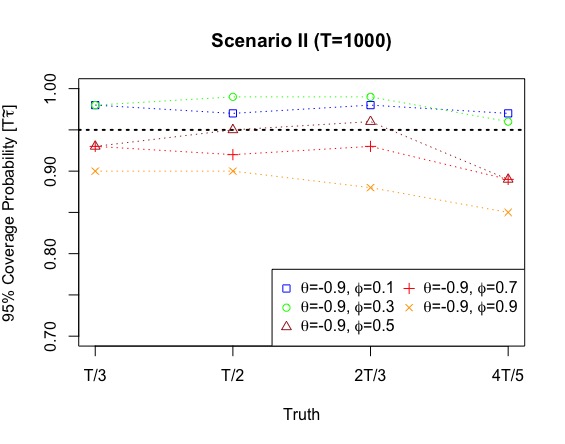} \\
\includegraphics[width=0.4\textwidth]{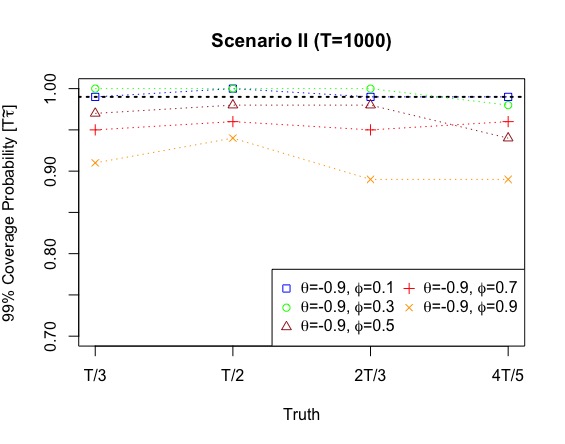}
\end{tabular}
\caption{($90\%$, $95\%$, $99\%$) CP for $\lfloor T\tilde{\tau} \rfloor$ for Scenario II with $T=1000$.}
\label{Fig.CP_CaseII_1000}
\end{figure}

\begin{figure}[ht]
\centering
\begin{tabular}{ cc }
\includegraphics[width=0.4\textwidth]{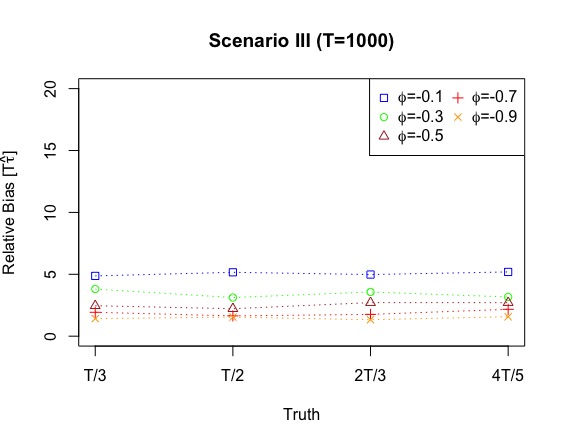} &
\includegraphics[width=0.4\textwidth]{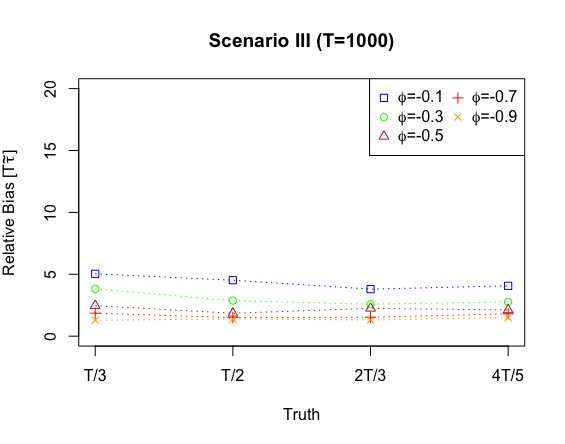} \\
\includegraphics[width=0.4\textwidth]{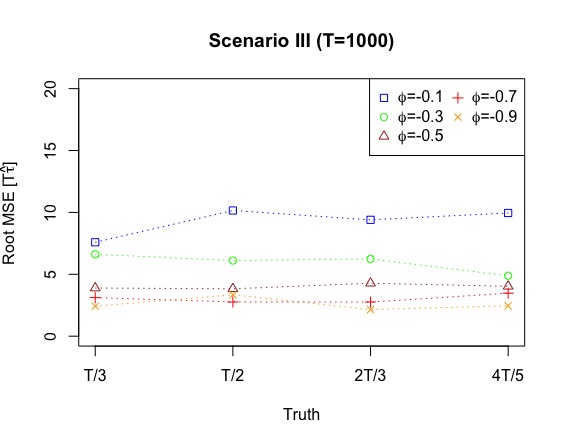} &
\includegraphics[width=0.4\textwidth]{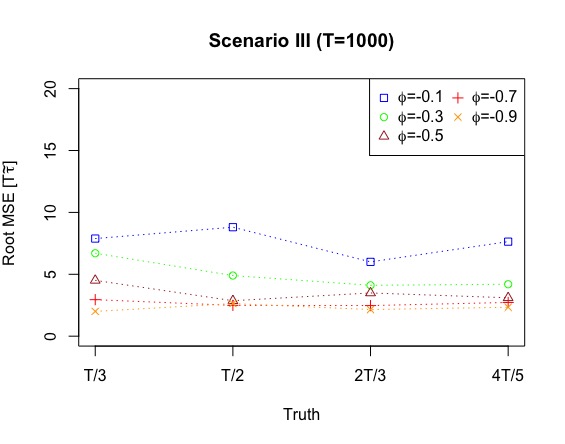} 
\end{tabular}
\caption{AB and RMSE of $\lfloor T\hat{\tau} \rfloor$ and $\lfloor T\tilde{\tau} \rfloor$ for Scenario III with $T=1000$.}
\label{Fig.AB_RMSE_CaseIII_1000}
\end{figure}

\begin{figure}[ht]
\centering
\begin{tabular}{ cc }
\includegraphics[width=0.4\textwidth]
{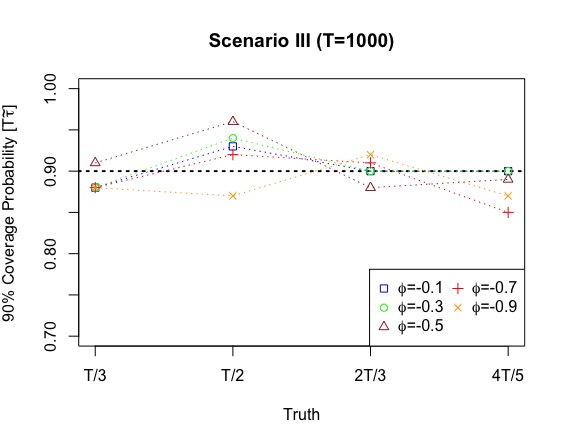} &
\includegraphics[width=0.4\textwidth]{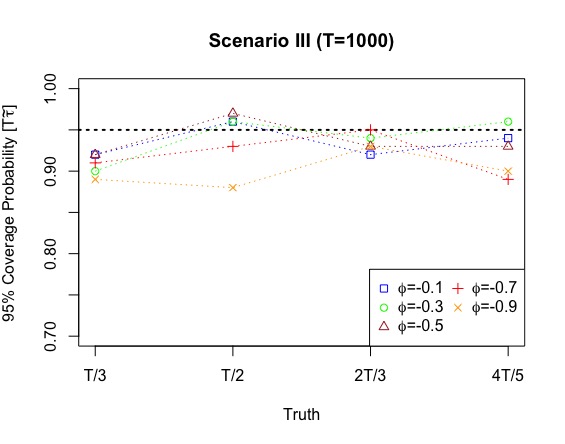} \\
\includegraphics[width=0.4\textwidth]{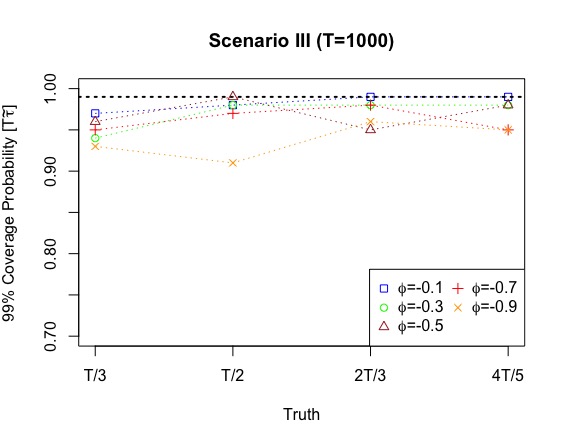}
\end{tabular}
\caption{($90\%$, $95\%$, $99\%$) CP for $\lfloor T\tilde{\tau} \rfloor$ for Scenario III with $T=1000$.}
\label{Fig.CP_CaseIII_1000}
\end{figure}

\begin{figure}[ht]
\centering
\begin{tabular}{ cc }
\includegraphics[width=0.4\textwidth]{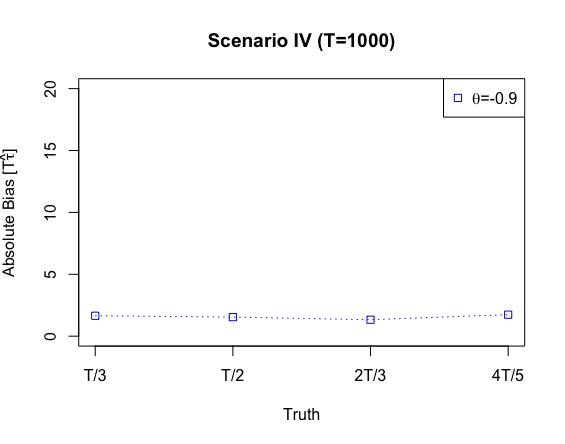} &
\includegraphics[width=0.4\textwidth]{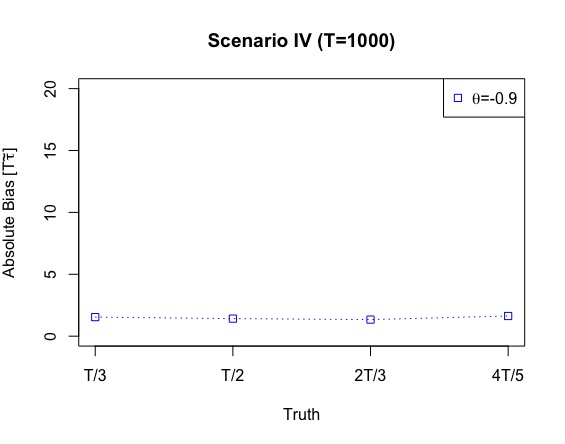} \\
\includegraphics[width=0.4\textwidth]{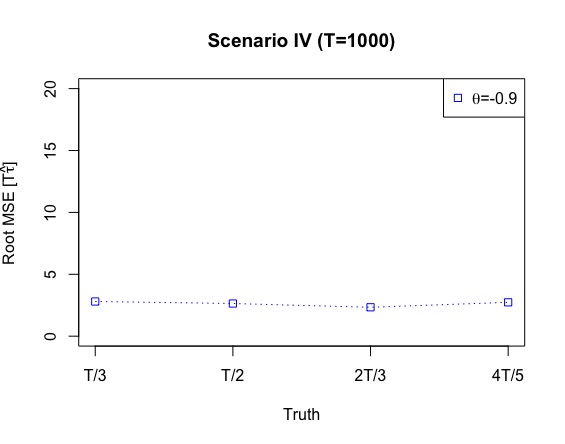} &
\includegraphics[width=0.4\textwidth]{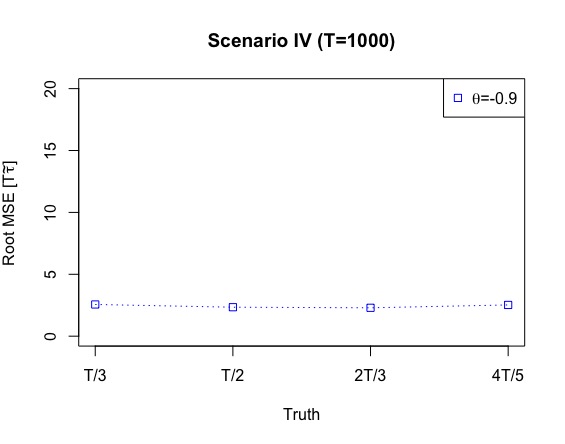} 
\end{tabular}
\caption{AB and RMSE of $\lfloor T\hat{\tau} \rfloor$ and $\lfloor T\tilde{\tau} \rfloor$ for Scenario IV with $T=1000$.}
\label{Fig.AB_RMSE_CaseIV_1000}
\end{figure}

\begin{figure}[ht]
\centering
\begin{tabular}{ cc }
\includegraphics[width=0.4\textwidth]{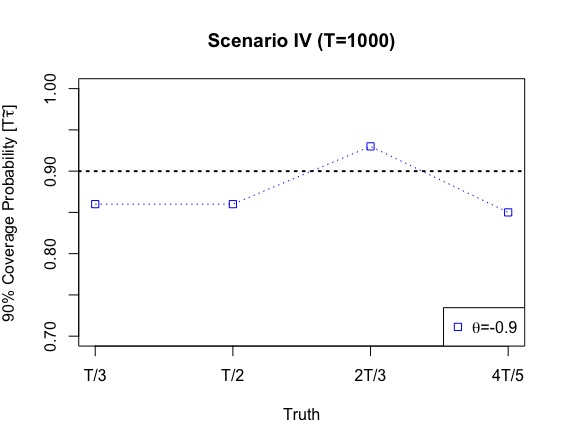} &
\includegraphics[width=0.4\textwidth]{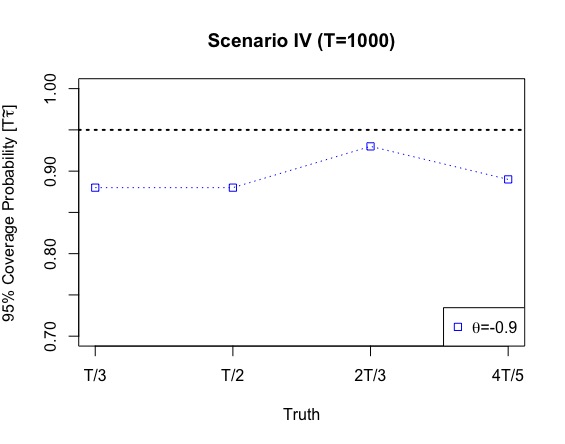} \\
\includegraphics[width=0.4\textwidth]{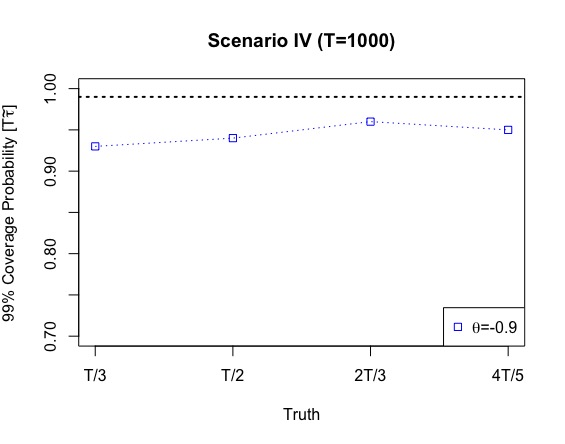}
\end{tabular}
\caption{($90\%$, $95\%$, $99\%$) CP for $\lfloor T\tilde{\tau} \rfloor$ for Scenario IV with $T=1000$.}
\label{Fig.CP_CaseIV_1000}
\end{figure}

\begin{figure}[ht]
\centering
\begin{tabular}{ cc }
\includegraphics[width=0.4\textwidth]{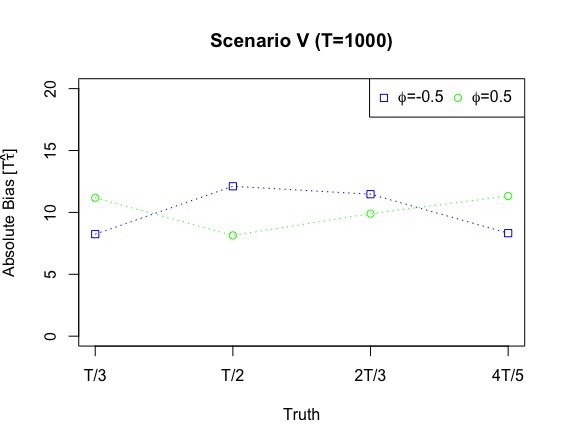} &
\includegraphics[width=0.4\textwidth]{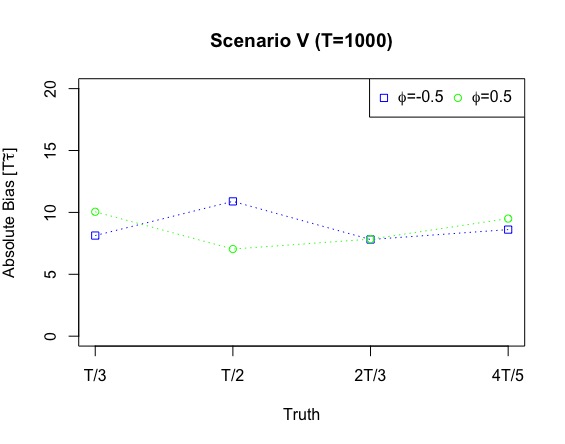} \\
\includegraphics[width=0.4\textwidth]{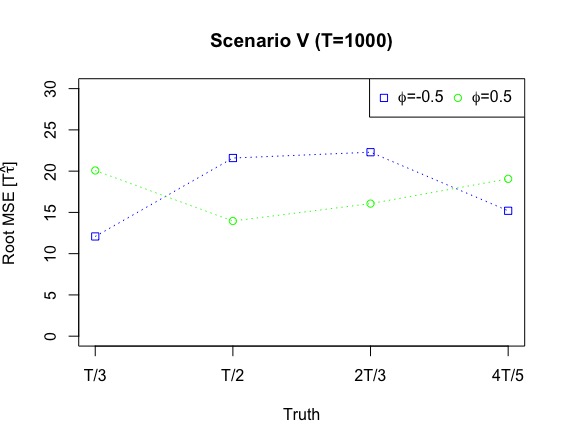} &
\includegraphics[width=0.4\textwidth]{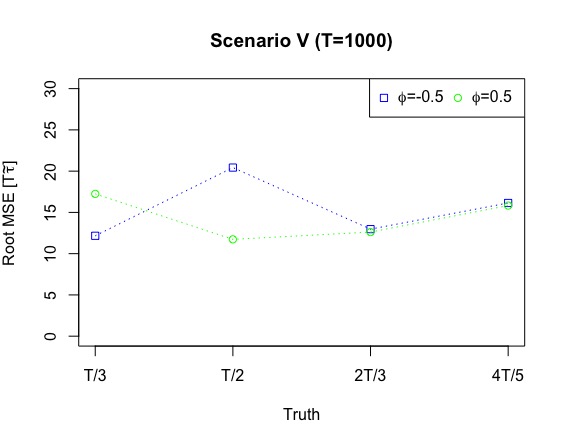} 
\end{tabular}
\caption{AB and RMSE of $\lfloor T\hat{\tau} \rfloor$ and $\lfloor T\tilde{\tau} \rfloor$ for Scenario V with $T=1000$.}
\label{Fig.AB_RMSE_CaseV_1000}
\end{figure}

\begin{figure}[ht]
\centering
\begin{tabular}{ cc }
\includegraphics[width=0.4\textwidth]{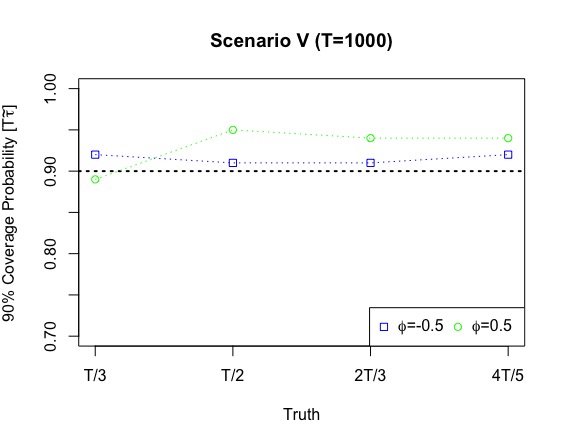} &
\includegraphics[width=0.4\textwidth]{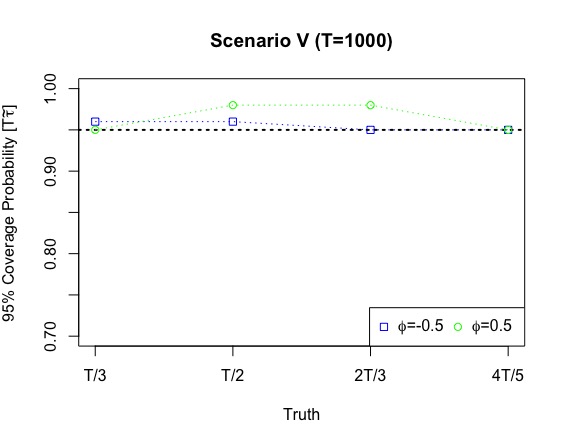} \\
\includegraphics[width=0.4\textwidth]{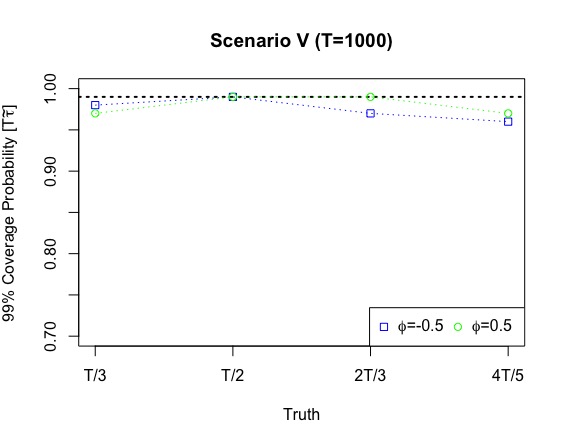}
\end{tabular}
\caption{($90\%$, $95\%$, $99\%$) CP for $\lfloor T\tilde{\tau} \rfloor$ for Scenario V with $T=1000$.}
\label{Fig.CP_CaseV_1000}
\end{figure}

\clearpage

\section{E \quad Additional Application Results} \label{sec:AppE}
\renewcommand{\theequation}{E.\arabic{equation}}
\setcounter{equation}{0}
\renewcommand{\thelemma}{E.\arabic{lemma}}
\renewcommand{\thedefinition}{E.\arabic{definition}}
\setcounter{table}{0}
\renewcommand{\thetable}{E.\arabic{table}}
\setcounter{figure}{0}
\renewcommand{\thefigure}{E.\arabic{figure}}

\subsection{E.1 \quad EEG data set}

Since our main methodology is motivated by the EEG data set, we provide the estimated models before and after the change point as follows.

$\bullet$ Channel P3: Before change point, the estimated process is AR(3) such that $X_t=-0.356 X_{t-1} - 0.263 X_{t-2} -0.209 X_{t-3} + \epsilon_t$, where $\epsilon_t \sim N(0,\sigma^2=0.175)$. After change point, the estimated process is AR(0) such that $X_t=\epsilon_t$, where $\epsilon_t \sim N(0,\sigma^2=0.414)$.
    
$\bullet$ Channel T3: Before change point, the estimated process is AR(0) such that $X_t = \epsilon_t$, where $\epsilon_t \sim N(0,\sigma^2=0.171)$. After change point, the estimated process is AR(1) $X_t=-0.325 X_{t-1}+\epsilon_t$, where $\epsilon_t \sim N(0,\sigma^2=0.337)$.
   
$\bullet$ Channel T5: Before change point, the estimated process is AR(1) such that $X_t=-0.228 X_{t-1}+\epsilon_t$, where $\epsilon_t \sim N(0,\sigma^2=0.261)$. After change point, the estimated process is AR(0) such that $X_t= \epsilon_t$ where $\epsilon_t \sim N(0,\sigma^2=0.519)$. 

We estimate the change points using some alternative methods, and the results are given in Table \ref{Tab.EEG_alters}. The SDE method could accurately estimate the change points, given that the truth is 85 seconds. We also construct the NSP confidence intervals for the EEG data sets, and the results are given in Figure \ref{Fig.NSP_CIs_EEG}. Note that the NSP method is sensitive to the significance level and there is no proper way to select it in practice. 

\begin{table}[ht]
\footnotesize
\caption{ Location of estimated change points in the EEG data set using some selected alternative methods. Note that the true value is around 85 seconds.} \label{Tab.EEG_alters}
\centering
\setlength{\tabcolsep}{10pt} 
\renewcommand{\arraystretch}{1.5}
\begin{tabular}{@{} cccccccc @{}}
\hline
Method & Channel P3 & Channel T3 & Channel T5   \\ \hline\hline
PELT & 92 & 90 & $\{ 18, 20, 33, 35, 85\}$ \\
SDE  & 92 & 90 & 85 \\
DP-UNIVAR & $\{92, 104 \}$ & $\{93, 95, 98, 100, 102 \}$ & $\{96, 98, 100, 103\}$ \\
LOCAL & 104 & 94 & 104 \\
STD & $\{104, 92, 97 \}$ & $\{ 94, 87, 100\}$ & $\{ 104, 90, 82, 97\}$ \\
WBS UNI-NONPAR & $\{ 104, 95\}$ & missing & $\{90, 82, 98 \}$ \\
ROB-WBS & $\{ 104, 100\}$ & missing & $\{104, 90, 85, 93 \}$  \\
\hline
\end{tabular}
\end{table}

\begin{figure}[ht]
\centering
\begin{tabular}{ c }
\includegraphics[width=1\textwidth]{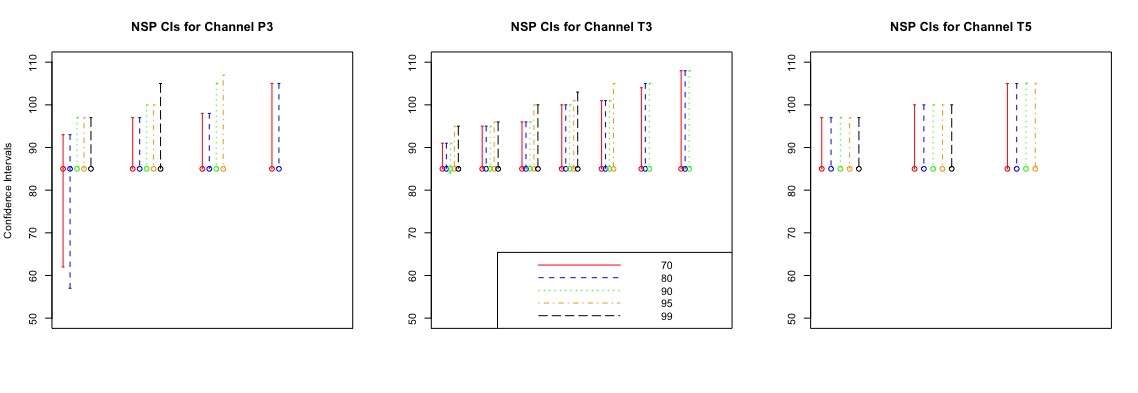} 
\end{tabular}
\caption{NSP confidence intervals for the EEG data set. The true value of $t=85$ is displayed as a point in all intervals.}
\label{Fig.NSP_CIs_EEG}
\end{figure}

\clearpage

\subsection{E.2 \quad Surveillance video data set}

We estimate the change points for applications related to the EEG data set and the surveillance video data set described in the manuscript using some alternative methods.

The images and matrices of the main surveillance video data set are shown in Figures \ref{Fig.image_walkout} and \ref{Fig.image_walkin}.
The results of comparisons across multiple methods from the surveillance video data sets for two of the actions are given in Tables \ref{Tab.Image1&2_PELT} -- \ref{Tab.Image1&2_ROB-WBS}. For the first action, the ROB-WBS method seems to perform better than the rest given the true value of 116. For the second action, we observe that DP-UNIVAR and STD could perform relatively well compared to the rest.

\begin{table}[ht]
\footnotesize
\caption{Location of estimated change points from the video data set for the ``\textit{first person walks out of the lobby} (so-called action 1)" and ``\textit{second person walks into the lobby} (so-called action 2)" using PELT method. Note that the true value is 116 for action 1 and 174 for action 2.} \label{Tab.Image1&2_PELT}
\centering
\setlength{\tabcolsep}{2pt} 
\renewcommand{\arraystretch}{1.5}
\begin{tabular}{@{} ccc @{}}
\hline
Action & Pixel & Estimated change points based on \textbf{PELT} method  \\ \hline\hline
1 & 700 & $\{9, 12, 29, 31, 37, 48, 60, 62, 64, 70, 72, 77, 79, 83, 92, 104, 115, 117, 120, 158, 161 \}$ \\
& 702 & $\{12, 18, 21, 24, 26, 54, 56, 58, 62, 64, 70, 72, 83, 87, 92, 96, 99,$ \\
 & & $103, 108, 110, 123, 125, 135, 138, 144, 147, 151, 168 \}$ \\
 & 731 & $\{4, 6, 8, 10, 14, 17, 43, 45, 53, 55, 57, 91, 94, 109, 114, 125, 127, 131, 141, 143, 161, 163 \}$ \\
& 762 & $\{24, 26, 30, 32, 37, 41, 43, 48, 59, 62, 64, 80, 82, 97, 101, 106, 122, 128 \}$ \\
& 764 & $\{4, 6, 12, 16, 22, 24, 34, 36, 53, 55, 61, 63, 72,  90, 96, 98, 104, 118, 127, 130, 132, 149, 155, 163 \}$ \\
\hline
2 & 48 & $\{131, 134, 214, 216, 219, 221 \}$ \\
& 78 & $\{129, 131, 157, 167, 169, 237, 239, 260, 262, 278 \}$ \\
& 110 & missing \\
& 174 & $\{119, 132, 134, 142, 144, 165, 167, 170, 172, 225, 227, 247, 253, 255, 287 \}$ \\
& 209 & $\{129, 131, 149, 151, 159, 162, 231, 233 \}$ \\
& 241 & $\{130, 132, 134, 137, 139, 167, 169, 184, 186, 209, 211, 216, 218,$ \\
& & $224, 228, 231, 233, 240, 242, 244, 246, 275, 279, 287\}$ \\
\hline
\end{tabular}
\end{table}

\begin{table}[ht]
\footnotesize
\caption{Location of estimated change points from the video data set for the ``\textit{first person walks out of the lobby} (so-called action 1)" and ``\textit{second person walks into the lobby} (so-called action 2)" using SDE method. Note that the true value is 116 for action 1 and 174 for action 2.} \label{Tab.Image1&2_SDE}
\centering
\setlength{\tabcolsep}{2pt} 
\renewcommand{\arraystretch}{1.5}
\begin{tabular}{@{} ccc @{}}
\hline
Action & Pixel & Estimated change points based on \textbf{SDE} method  \\ \hline\hline
1 & 700 & 24 \\
& 702 & 88 \\
 & 731 & 91 \\
& 762 & 96 \\
& 764 & 104 \\
\hline
2 & 48 & 259 \\
& 78 & 246 \\
& 110 & 231 \\
& 174 & 237 \\
& 209 & 236 \\
& 241 & 201 \\
\hline
\end{tabular}
\end{table}

\begin{table}[ht]
\footnotesize
\caption{Location of estimated change points from the video data set for the ``\textit{first person walks out of the lobby} (so-called action 1)" and ``\textit{second person walks into the lobby} (so-called action 2)" using DP-UNIVAR method. Note that the true value is 116 for action 1 and 174 for action 2.} \label{Tab.Image1&2_DP-UNIVAR}
\centering
\setlength{\tabcolsep}{2pt} 
\renewcommand{\arraystretch}{1.5}
\begin{tabular}{@{} ccc @{}}
\hline
Action & Pixel & Estimated change points based on \textbf{DP-UNIVAR} method  \\ \hline\hline
1 & 700 & $\{48, 91, 104, 114\}$ \\
& 702 & $\{83, 90, 101, 108 \}$ \\
 & 731 & $\{93, 109 \}$ \\
& 762 & $\{97, 106 \}$ \\
& 764 & $\{104, 118 \}$ \\
\hline
2 & 48 & $\{126, 144, 151, 182, 188, 207, 218, 224, 230, 260, 274 \}$ \\
& 78 & $\{157, 178, 248, 278, 285, 287 \}$ \\
& 110 & $\{133, 145, 152, 173, 176, 231, 234, 276, 283 \}$ \\
& 174 & $\{121, 128, 157, 170, 227, 246, 262, 272, 284 \}$ \\
& 209 & $\{138, 147, 158, 164, 203, 254, 259 \}$ \\
& 241 & $\{142, 219, 260 \}$ \\
\hline
\end{tabular}
\end{table}

\begin{table}[ht]
\footnotesize
\caption{Location of estimated change points from the video data set for the ``\textit{first person walks out of the lobby} (so-called action 1)" and ``\textit{second person walks into the lobby} (so-called action 2)" using LOCAL method. Note that the true value is 116 for action 1 and 174 for action 2.} \label{Tab.Image1&2_LOCAL}
\centering
\setlength{\tabcolsep}{2pt} 
\renewcommand{\arraystretch}{1.5}
\begin{tabular}{@{} ccc @{}}
\hline
Action & Pixel & Estimated change points based on \textbf{LOCAL} method  \\ \hline\hline
1 & 700 & 84 \\
& 702 & 79 \\
 & 731 & 91 \\
& 762 & 96 \\
& 764 & 98 \\
\hline
2 & 48 & 260 \\
& 78 & 278 \\
& 110 & 276 \\
& 174 & 170 \\
& 209 & 203 \\
& 241 & 142 \\
\hline
\end{tabular}
\end{table}

\begin{table}[ht]
\footnotesize
\caption{Location of estimated change points from the video data set for the ``\textit{first person walks out of the lobby} (so-called action 1)" and ``\textit{second person walks into the lobby} (so-called action 2)" using STD method. Note that the true value is 116 for action 1 and 174 for action 2.} \label{Tab.Image1&2_STD}
\centering
\setlength{\tabcolsep}{2pt} 
\renewcommand{\arraystretch}{1.5}
\begin{tabular}{@{} ccc @{}}
\hline
Action & Pixel & Estimated change points based on \textbf{STD} method  \\ \hline\hline
1 & 700 & $\{84, 48, 109, 37, 65, 91, 120, 102, 161, 132 \}$ \\
& 702 & $\{79, 12, 104, 90, 129, 109 \}$ \\
 & 731 & $\{91, 109, 96, 114 \}$ \\
& 762 & $\{91, 37, 106, 97, 120 \}$ \\
& 764 & $\{98, 118, 104, 134 \}$ \\
\hline
2 & 48 & $\{126, 260, 239, 274, 223, 254, 265, 207, 230, 244, 188, 217, 182, 199, 142, 151, 163, 174 \}$ \\
& 78 & $\{278, 140, 178, 157, 248, 169, 188, 266, 197, 258, 271, 219, 208, 242, 227, 235 \}$ \\
& 110 & $\{133, 122, 273, 234, 283, 171, 260, 152, 176, 239, 145, 203, 197, 208, 228 \}$ \\
& 174 & $\{284, 275, 172, 157, 225, 128, 164, 210, 246, 151, 178, 217, 237, 262, 137, 197, 230, 189, 204 \}$ \\
& 209 & $\{122, 203, 164, 211, 158, 173, 224, 138, 182, 252, 129, 147, 192, 231, 259, 244, 266 \}$ \\
& 241 & $\{142, 122, 272, 129, 219, 164, 260, 158, 171, 224, 266, 200, 229, 189, 209, 236, 250 \}$ \\
\hline
\end{tabular}
\end{table}

\begin{table}[ht]
\footnotesize
\caption{Location of estimated change points from the video data set for the ``\textit{first person walks out of the lobby} (so-called action 1)" and ``\textit{second person walks into the lobby} (so-called action 2)" using WBS UNI-NONPAR method. Note that the true value is 116 for action 1 and 174 for action 2.} \label{Tab.Image1&2_WBS UNI-NONPAR}
\centering
\setlength{\tabcolsep}{2pt} 
\renewcommand{\arraystretch}{1.5}
\begin{tabular}{@{} ccc @{}}
\hline
Action & Pixel & Estimated change points based on \textbf{WBS UNI-NONPAR} method  \\ \hline\hline
1 & 700 & $\{108, 48, 114, 37, 60 \}$ \\
& 702 & $\{89, 83, 102, 36, 96, 108 \}$ \\
 & 731 & $\{91, 109, 96, 114 \}$ \\
& 762 & $\{95, 37, 106, 48, 120 \}$ \\
& 764 & $\{101, 90, 118, 106 \}$ \\
\hline
2 & 48 & $\{260, 239, 274, 142, 254, 265, 134, 151, 244, 163, 223, 174, 230, 182 \}$ \\
& 78 & $\{219, 208, 258, 188, 232, 274, 163, 248, 242 \}$ \\
& 110 & $\{273, 234, 203, 260, 177, 208, 239, 197, 228 \}$ \\
& 174 & $\{284, 172, 157, 275, 128, 164, 223, 151, 197, 246, 137, 189, 210, 237, 262, 178, 204, 217 \}$ \\
& 209 & $\{158, 147, 164, 138, 203, 173, 211, 183, 252, 192, 224, 259, 231, 244 \}$ \\
& 241 & $\{142, 260, 219, 272, 164, 224, 266, 158, 171, 229, 200, 236, 189, 209, 250 \}$ \\
\hline
\end{tabular}
\end{table}

\begin{table}[ht]
\footnotesize
\caption{Location of estimated change points from the video data set for the ``\textit{first person walks out of the lobby} (so-called action 1)" and ``\textit{second person walks into the lobby} (so-called action 2)" using ROB-WBS method. Note that the true value is 116 for action 1 and 174 for action 2.} \label{Tab.Image1&2_ROB-WBS}
\centering
\setlength{\tabcolsep}{2pt} 
\renewcommand{\arraystretch}{1.5}
\begin{tabular}{@{} ccc @{}}
\hline
Action & Pixel & Estimated change points based on \textbf{ROB-WBS} method  \\ \hline\hline
1 & 700 & $\{111, 48, 141, 37, 60, 120 \}$ \\
& 702 & $\{88, 83, 102, 36, 92, 108 \}$ \\
 & 731 & $\{91, 111, 94, 109 \}$ \\
& 762 & $\{91, 37, 117, 48, 97, 106 \}$ \\
& 764 & $\{98, 118, 104 \}$ \\
\hline
2 & 48 & $\{260, 239, 274, 142, 254, 264, 126, 151, 243, 127, 163, 130, 161, 188, 182, 231, 174, 224 \}$ \\
& 78 & $\{219, 208, 258, 188, 232, 278, 157, 227, 235, 242, 248 \}$ \\
& 110 & $\{273, 234, 231, 266, 171, 236, 176, 246, 203, 248, 201, 223, 209, 210 \}$ \\
& 174 & $\{172, 157, 275, 128, 169, 247, 284, 152, 225, 262, 137, 197, 237, 272, 189, 217, 230, 176, 210, 205 \}$ \\
& 209 & $\{156, 147, 164, 138, 200, 181, 211, 178, 192, 262, 175, 254, 224, 220, 231 \}$ \\
& 241 & $\{142, 121, 260, 129, 219, 272, 164, 224, 268, 284, 158, 171,$ \\
& & $228, 282, 145, 200, 236, 149, 195, 209,
239, 189 \}$ \\
\hline
\end{tabular}
\end{table}

\clearpage

\begin{figure}[ht]
\centering
\begin{tabular}{ cc }
\includegraphics[width=0.65\textwidth]{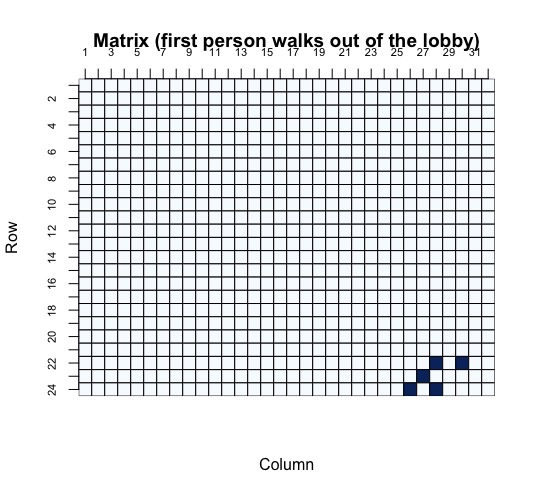} \\
\includegraphics[width=0.65\textwidth]{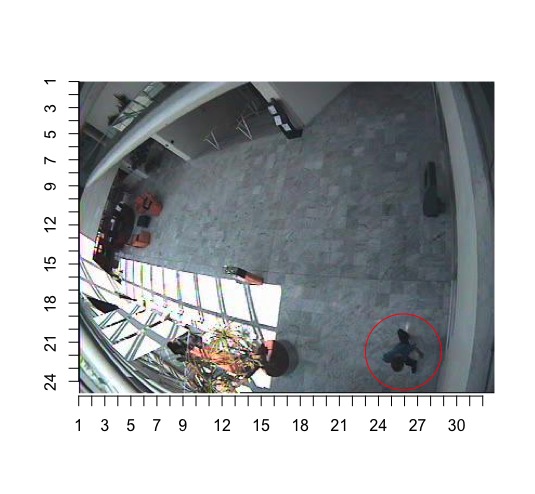} 
\end{tabular}
\caption{The result of image analysis for the ``\textit{first person walks out of the lobby}". The shaded cells from the matrix clearly predict the pixels that the person moves out of the lobby. Note that the true value is 116 for action 1.}
\label{Fig.image_walkout}
\end{figure}

\begin{figure}[ht]
\centering
\begin{tabular}{ cc }
\includegraphics[width=0.65\textwidth]{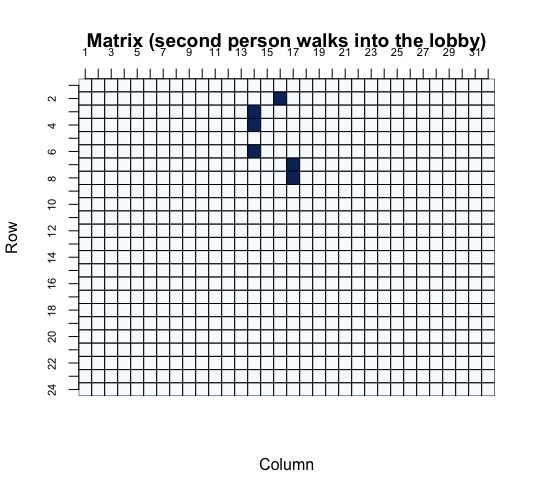} \\
\includegraphics[width=0.65\textwidth]{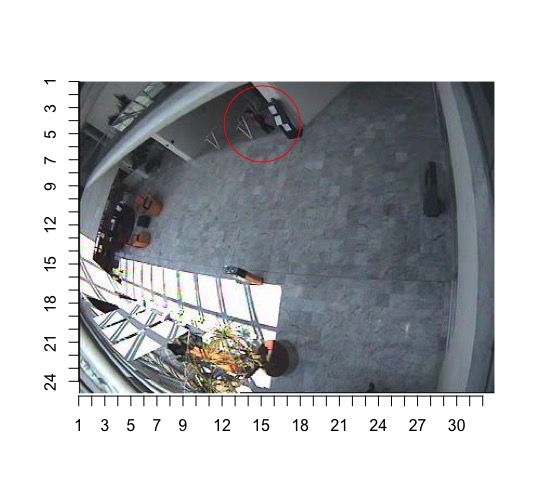} 
\end{tabular}
\caption{The result of image analysis for the ``\textit{second person walks into the lobby}". The shaded cells from the matrix clearly predict the pixels that the person moves into the lobby. Note that the true value is 174 for action 2.}
\label{Fig.image_walkin}
\end{figure}

\end{document}